\documentclass[aos,preprint]{imsart}

\RequirePackage{amsthm,amsmath,amsfonts,amssymb}
\RequirePackage[numbers,sort&compress]{natbib}
\RequirePackage[colorlinks,citecolor=blue,urlcolor=blue]{hyperref}
\RequirePackage{graphicx}
\usepackage{lipsum}
\usepackage{amsfonts}
\usepackage{graphicx}
\usepackage{epstopdf}
\usepackage{ragged2e}

\usepackage{mathtools}
\usepackage{bm}
\usepackage[mathscr]{euscript}
\usepackage[T1]{fontenc}
\usepackage{enumitem}
\usepackage{xcolor}
\usepackage{multirow}
\usepackage{booktabs}
\usepackage{microtype}
\usepackage{fix-cm}
\usepackage{colortbl}
\usepackage[utf8]{inputenc}
\usepackage{tikz}
\usepackage{pbox}
\usepackage{fourier} 
\usepackage{array}
\usepackage{makecell}

\usepackage{algorithm}

\usepackage{bbding}

\startlocaldefs
\theoremstyle{plain}

\newtheorem{theorem}{Theorem}[section]
\newtheorem{lemma}[theorem]{Lemma}
\DeclareMathOperator{\pr}{\mathbb P}
\DeclareMathOperator{\E}{\mathbb E}

\DeclareMathOperator{\ind}{\mathbb I}

\newcommand{\Real}{\mathbb R}

\newcommand{\NatInt}{\mathbb N}

\newcommand{\CalO}{\mathcal O}

\newcommand{\CalH}{\mathcal H}
\newcommand{\CalX}{\mathcal X}
\newcommand{\CalP}{\mathcal P}
\newcommand{\CalU}{\mathcal U}
\newcommand{\CalK}{\mathcal K}
\newcommand{\BX}{\bold X}
\newcommand{\BI}{\bold I}
\newcommand{\Bx}{\bold x}
\newcommand{\Bs}{\bold s}

\newcommand{\BFx}{\bold x}
\newcommand{\BB}{\bold B}

\newcommand{\BU}{\bold U}
\newcommand{\Bu}{\bold u}
\newcommand{\BK}{\bold K}
\newcommand{\im}{\text{i}}
\newcommand{\BPhi}{\boldsymbol{\Phi}}
\newcommand{\CalF}{\mathcal{F}}
\newcommand{\CalW}{\mathcal{W}}
\newcommand{\Bomega}{\boldsymbol{\omega}}

\newtheorem{corollary}{Corollary}

\theoremstyle{remark}
\newtheorem{definition}[theorem]{Definition}

\endlocaldefs

\begin{document}

\begin{frontmatter}
\title{The BdryMat\'ern GP: Reliable incorporation of boundary information on irregular domains for Gaussian process modeling}
\runtitle{BdryMat\'ern GP}

\begin{aug}
\author[A]{\fnms{Liang}~\snm{Ding}\ead[label=e1]{liang\_ding@fudan.edu.cn}},
\author[B]{\fnms{Simon}~\snm{Mak}\ead[label=e2]{sm769@duke.edu}}
\and
\author[C]{\fnms{C. F. Jeff}~\snm{Wu}\ead[label=e3]{jeffwu@cuhk.edu.cn}}
\address[A]{School of Data Science,
Fudan University\printead[presep={ ,\ }]{e1}}

\address[B]{Department of Statistical Science,
Duke University\printead[presep={,\ }]{e2}}
\address[C]{School of Data Science,
The Chinese University of Hong Kong (Shenzhen)\printead[presep={,\ }]{e3}}
\end{aug}

\begin{abstract}
Gaussian processes (GPs) are broadly used as surrogate models for expensive computer simulators of complex phenomena. However, a key bottleneck is that its training data are generated from this expensive simulator and thus can be highly limited. A promising solution is to supplement the learning model with boundary information from scientific knowledge. However, despite recent work on boundary-integrated GPs, such models largely cannot accommodate boundary information on irregular (i.e., non-hypercube) domains, and do not provide sample path smoothness control or approximation error analysis, both of which are important for reliable surrogate modeling. We thus propose a novel BdryMat\'ern GP modeling framework, which can reliably integrate Dirichlet, Neumann and Robin boundaries on an irregular connected domain with a boundary set that is twice-differentiable almost everywhere. Our model leverages a new BdryMat\'ern covariance kernel derived in path integral form via a stochastic partial differential equation formulation. Similar to the GP with Mat\'ern kernel, we prove that sample paths from the BdryMat\'ern GP satisfy the desired boundaries with smoothness control on its derivatives. We further present an efficient approximation procedure for the BdryMat\'ern kernel using finite element modeling with rigorous error analysis. Finally, we demonstrate the effectiveness of the BdryMat\'ern GP in a suite of numerical experiments on incorporating broad boundaries on irregular domains.

\end{abstract}



\end{frontmatter}

\section{Introduction}

With groundbreaking advances in scientific modeling and computing, complex phenomena (e.g., rocket propulsion and particle collisions) can now be reliably simulated via computer code, which typically solve a sophisticated system of partial differential equations (PDEs). Such ``computer experiments'' have had demonstrated success in advancing broad scientific and engineering applications, including particle physics \citep{sengupta2025hybrid,liyanage2022efficient}, aerospace engineering \citep{mak2018efficient,miller2024expected} and robotics \citep{choi2021use,liu2025quip}. A key bottleneck of such virtual experiments, however, is that each run demands a large amount of computing power. For example, a single simulation of a heavy-ion particle collision and its subsequent subatomic interactions can require over thousands of CPU hours \citep{ji2024conglomerate}. Given a limited computational budget, this makes the exploration of the simulation response surface over an input space $\mathcal{X}$ a highly challenging task.

A proven solution is probabilistic surrogate modeling \citep{gramacy2020surrogates}. The idea is simple but effective. First, one runs the expensive simulator (denoted $f$) on a set of $n$ designed input points on $\mathcal{X} \subseteq \mathbb{R}^d$, where $\mathcal{X}$ is the compact domain of its $d$ input parameters. Second, using such data, one trains a probabilistic model that can predict with uncertainty the simulator output $f(\mathbf{x}_{\rm new})$ at an untested input $\mathbf{x}_{\rm new} \in \mathcal{X}$, thus bypassing an expensive simulation run. Gaussian processes (GPs; \citep{gramacy2020surrogates}) and their extensions are widely used as surrogate models as they provide a closed-form predictive distribution on $f(\cdot)$, which facilitates timely downstream scientific decision-making, e.g., optimization \citep{chen2024hierarchical} and inverse problems \citep{ehlers2024bayesian}. In practice, one faces a further obstacle of limited sample sizes (i.e., small $n$) for surrogate training, due to the costly nature of the expensive computer simulator. For complex applications, e.g., in aerospace engineering \citep{narayanan2024misfire}, this can restrict the experimenter to only tens of runs over the domain $\mathcal{X}$, which results in mediocre surrogate performance.

To address this, recent work has explored the integration of known ``physics'' from scientific knowledge within surrogate models; this field of ``physics-integrated machine learning'' \citep{willard2020integrating} is a rising and promising area. For GP surrogates, this includes the incorporation of known shape constraints \citep{WangBerger16,Golchi15}, mechanistic equations \citep{Wheeler14}, manifold embeddings \citep{li2023additive,seshadri2019dimension} and boundary information \citep{tan2018gaussian,vernon2019known,jackson2023efficient,ding2019bdrygp}; we focus on the last topic in this paper. Here, boundary information refers to known information on $f$ along the boundaries of its domain $\mathcal{X}$. Such knowledge may be elicited via a careful analysis of the underlying PDE system or via scientific intuition. Take, e.g., the bending of a metal beam under a fixed force, and consider its displacement $f(\mathbf{x})$ as either beam inverse thickness ($x_1$) and beam length ($x_2$) varies. As thickness grows large (i.e., $x_1 \rightarrow 0$) or length becomes small (i.e., $x_2 \rightarrow 0$), intuition suggests that its displacement $f(\mathbf{x})$ should approach zero; such boundary information can be integrated for surrogate training. Similar information can be elicited in the simulation of viscous flows \cite{white2006viscous}, which are used in climatology and high energy physics. Such flows are dictated by the complex Navier-Stokes equations \cite{temam2001navier}, but at the limits of inputs (e.g., zero viscosity or fluid incompressibility), these equations may admit closed-form solutions \cite{kiehn2001some,humphrey2016introduction}. The integration of this boundary information can greatly improve surrogate modeling with limited simulation runs.

There is a growing literature on integrating boundary information within GPs. Here, Dirichlet, Neumann and Robin boundaries (see \cite{zwillinger2021handbook}) refer to knowledge of $f$, its derivative, or a linear combination of both along the boundaries of $\mathcal{X}$, respectively. \cite{tan2018gaussian, li2022improving} proposed a boundary-modified GP, which modifies the mean and covariance structure of a stationary GP to incorporate Dirichlet boundaries on discrete points. \cite{vernon2019known} and \cite{jackson2023efficient} adopted a projection-based approach for integrating Dirichlet boundaries on axis-symmetric and/or axis-parallel boundaries on a hypercube domain $\mathcal{X}$. \cite{gulian2022gaussian,solin2019know,solin2020hilbert} advocated for a low-rank kernel approximation strategy, where its underlying basis functions explicitly satisfy the prescribed Dirichlet boundary conditions. Finally, \cite{dalton2024boundary} makes use of so-called ``approximate distance functions'' (ADFs) for constructing kernels that can integrate Dirichlet, Robin or Neumann boundaries on non-hypercube domains.

Such models, however, have several limitations. First, aside from \cite{dalton2024boundary}, they consider only the integration of Dirichlet boundaries and not Neumann or Robin boundaries, which can often be elicited in applications \cite{dalton2024boundary}. Second, existing models do not provide a characterization of smoothness (i.e., differentiability) for its resulting sample paths. As such, they do not permit \textit{smoothness control} of the fitted probabilistic model, which is important for effective surrogate modeling \cite{santner2003}. Finally, the input space $\mathcal{X}$ can be ``irregular'', i.e., it takes the form of a non-hypercube domain; see Figure \ref{fig:task} (left) for an example. Such irregular domains are broadly encountered in physical science applications, including glaciology \citep{pratola2017design} and fluid dynamics \cite{mak2018efficient,dalton2024boundary}. Incorporating boundary information on irregular domains requires a careful approximation approach with rigorous \textit{error analysis} for reliable surrogate modeling, which existing methods do not provide. For modern applications that feature rich boundary information on irregular input spaces, there is thus a need for a theoretically grounded GP framework that reliably integrates such information with smoothness and approximation error control for effective surrogate modeling.

The proposed BdryMat\'ern GP targets such a framework (see Figure \ref{fig:task} for a visualization). Recall that a key appeal of the Mat\'ern GP, i.e., a GP with (isotropic) Mat\'ern kernel, is its control of sample path differentiability via a smoothness parameter $\nu$ \citep{stein1999}. To extend this for the boundary-integrated setting, we begin with the stochastic partial differential equation (SPDE) representation of a Mat\'ern GP \citep{LindgrenRueLindstrom11}, and impose the known boundary information on domain $\mathcal{X}$ as a constraint on this SPDE system. Provided $\mathcal{X}$ is connected with a boundary set that is twice-differentiable almost everywhere, we can then solve this system to derive the proposed BdryMat\'ern covariance kernel in path integral form using the well-known Feynman-Kac formula \citep{ito1957fundamental,stroock1971diffusion,papanicolaou1990probabilistic} for elliptical PDEs. With this new BdryMat\'ern kernel, we prove that sample paths from the BdryMat\'ern GP enjoy similar smoothness control as the Mat\'ern GP while satisfying the desired boundaries. We further present an efficient approach for approximating the BdryMat\'ern kernel on an irregular domain $\mathcal{X}$ via finite element modeling (FEM; \cite{huebner2001finite}) with rigorous error analysis. We finally demonstrate in a suite of numerical experiments the effectiveness of the BdryMat\'ern GP in reliably integrating broad boundary information on irregular domains.

This paper is organized as follows: Section \ref{sec:GP_review} provides an overview of GPs and existing work on boundary-integrated GPs. Section \ref{sec:Bdrymat_SPDE} derives the new BdryMat\'ern kernel in path integral form, and proves smoothness control on sample paths from its corresponding BdryMat\'ern GP. Section \ref{sec:kernel_approx} outlines a finite-element-method approach for kernel approximation with supporting error analysis. Section \ref{sec:tensor_mat} presents a tensor form of the BdryMat\'ern GP that provides closed-form kernels on the unit hypercube domain. Section \ref{sec:numeric} demonstrates the effectiveness of the BdryMat\'ern GP in numerical experiments, and Section \ref{sec:conclusion} concludes the paper.

\begin{figure}[!t]
\begin{center}
\includegraphics[width=\textwidth]
{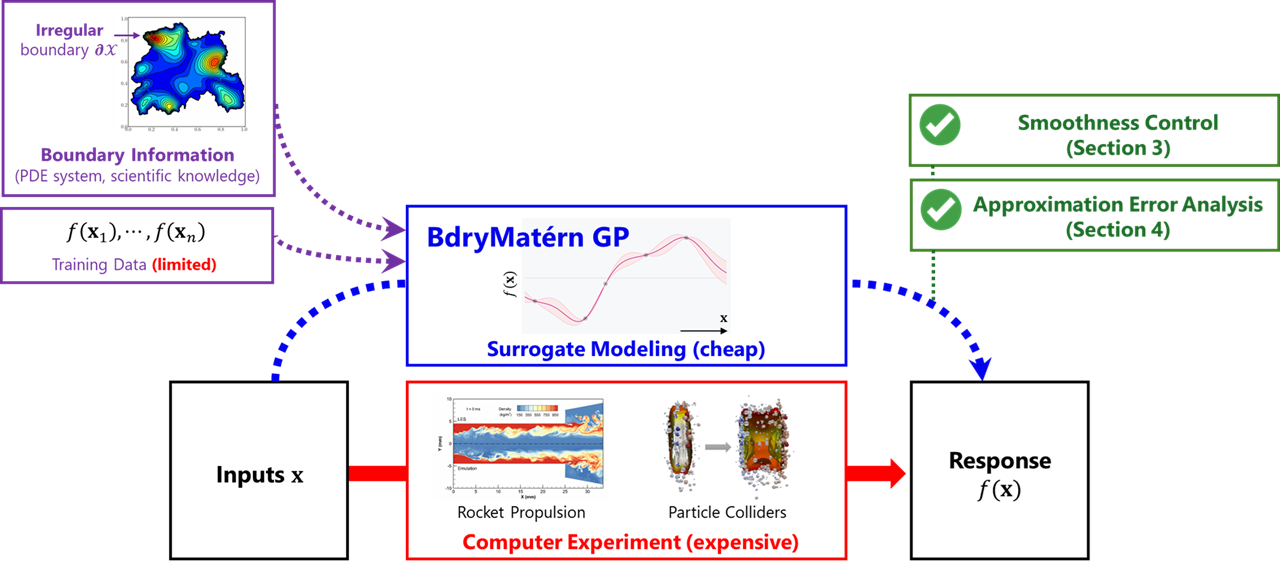}
\end{center}
\caption{Visualizing the proposed BdryMat\'ern GP for reliable incorporation of boundary information for probabilistic surrogate modeling. Figures in this workflow adapted from \cite{dalton2024boundary}, \cite{mak2018efficient} and \cite{ji2024graphical}.}
\label{fig:task}
\end{figure}

\section{Background \& Motivation}
\label{sec:GP_review}
We first provide in this section a brief overview on GP modeling, the isotropic Mat\'ern kernel and its representation as an SPDE. We then review existing GP models that integrate boundaries, and their potential limitations.

\subsection{GP Modeling and the Mat\'ern Kernel}
Let $\Bx \in \mathcal{X}\subseteq\Real^d$ be an input vector lying on a compact domain $\mathcal{X}$, and let $f(\Bx)$ denote its black-box output from the expensive computer simulator. We adopt the following Gaussian process model on $f$:
\begin{equation}
    f(\cdot) \sim \text{GP}\{\mu(\Bx),k(\Bx,\Bx')\}.
\end{equation}
Here, $\mu(\Bx)$ is its mean function, and $k(\Bx,\Bx')$ is its covariance function or kernel. With little prior information on $f$, $\mu(\Bx)$ is typically specified as a constant $\mu$ (to be estimated from data), and $k(\Bx,\Bx')$ is selected to ensure a desirable degree of smoothness for GP sample paths (e.g., the Mat\'ern kernel later). Given known boundaries on $f$, however, both $\mu(\Bx)$ and $k(\Bx,\Bx')$ should be carefully specified so that its sample paths satisfy such boundaries. We will later denote this \textit{boundary-integrated} mean and covariance function as $\mu_{\mathcal{B}}(\mathbf{x})$ and $k_{\mathcal{B}}(\Bx,\Bx')$, respectively.


Next, suppose the computer simulator is run at a set of $n$ designed input points $\BX =\{\Bx_1,\cdots,\Bx_n\}\subset \mathcal{X}$, yielding observations $f(\BX)=[f(\Bx_1),\cdots,f(\Bx_n)]^T$. In what follows, we presume that the simulator $f$ is \textit{deterministic}, meaning the data $f(\BX)$ are observed without noise. This is often assumed in the computer experiments literature \citep{santner2003}, when the simulator solves a deterministic PDE system and returns the same solution $f(\Bx)$ given the same inputs $\Bx$. One can easily accommodate for Gaussian simulation noise via the use of a nugget term \citep{peng2014choice} in the below predictive equations.

Conditional on data $f(\BX)$, the posterior predictive distribution of $f(\Bx_{\rm new})$ at an unevaluated input point $\Bx_{\rm new}$ takes the form of a Gaussian distribution $[f(\mathbf{x}_{\rm new})|f(\mathbf{X})] \sim \mathcal{N}\{\hat{f}_n(\Bx_{\rm new}),k_n(\Bx_{\rm new},\Bx_{\rm new})\}$, with posterior mean and variance given by:
\begin{align}
\begin{split}
    \hat{f}_n(\Bx)&=\mu(\Bx)+\mathbf{k}(\Bx,\BX)\mathbf{K}^{-1}(\BX,\BX)\left[f(\BX)-\mu(\BX)\right],\\
    k_n(\Bx,\Bx')&=k(\Bx,\Bx')-\mathbf{k}(\Bx,\BX)\mathbf{K}^{-1}(\BX,\BX)\mathbf{k}(\BX,\Bx').
    \end{split}
    \label{eq:kriging}
\end{align}
Here, $\mathbf{k}(\BX,\Bx_{\rm new}) = [k(\Bx_{\rm new},\Bx_1),\cdots,k(\Bx_{\rm new},\Bx_n)]^T$ is the covariance vector between the design points $\BX$ and a new point $\Bx_{\rm new}$, $k(\Bx,\BX) = k(\BX,\Bx)^T$, and $\mathbf{K}(\BX,\BX) = {{[k(\Bx_i},\Bx_j)]_{i=1}^n}_{j=1}^n$ is the covariance matrix over design points. Equation \eqref{eq:kriging} provides the means for ``emulating'' the expensive simulator $f$ with probabilistic uncertainty quantification over the domain $\mathcal{X}$.


For GP modeling, a popular covariance function choice is the (isotropic) Mat\'ern kernel \citep{stein1999}, given by:
\begin{equation}
    \label{eq:matern}
    k_{\nu}(\BFx,\BFx')=\frac{\sigma^2}{2^{\nu-1}\Gamma(\nu)}(\kappa\|\BFx-\BFx'\|)^\nu K_\nu(\kappa\|\BFx-\BFx'\|).
\end{equation}
Here, $\nu > 0$ is its smoothness parameter, $\kappa >0$ is its inverse length-scale parameter, $\sigma^2 > 0$ is its variance parameter, and $K_\nu$ is the modified Bessel function of the second kind \citep{abramowitz1965handbook}. A key appeal for the Mat\'ern kernel with smoothness $\nu$ is that its sample path smoothness can be controlled by $\nu$; in particular, its sample paths are \textup{($\lceil\nu\rceil$-1)}-differentiable almost surely \citep{da2023sample}. Such smoothness control is desirable in broad scientific applications, e.g., climatology \citep{wiens2020modeling} and geostatistics \citep{pardo2008geostatistics}. It is particularly important for our motivating surrogate modeling application, where the black-box function $f$ (which often represents the solution of a PDE system) may have known smoothness, i.e., differentiability, properties \citep{Evans15}. Integrating this smoothness information via a careful selection of $\nu$ in the Mat\'ern kernel permits interpretable surrogate modeling with improved predictive performance \citep{stein1999}.


Consider now $\mathcal{W}_{\nu}(\Bx)$, the zero-mean Mat\'ern GP with smoothness parameter $\nu$, inverse length-scale parameter $\kappa$, and process variance $\sigma^2 = \Gamma(\nu)[\Gamma(\alpha)(4\pi)^{d/2}\kappa^{2\nu}\tau^2]^{-1}$, where $\alpha=\nu+d/2$ and $\tau^2 > 0$ is given. (For brevity, all parameters except for $\nu$ are suppressed in the notation $\mathcal{W}_{\nu}(\Bx)$.) One can show (see, e.g., \cite{whittle1954stationary,whittle1963stochastic,LindgrenRueLindstrom11}) that $\mathcal{W}_{\nu}(\Bx)$ is the unique stationary solution to the following SPDE system:
\begin{equation}
    \label{eq:SPDE}
    \tau (\kappa^2-\bigtriangleup)^{\alpha/2}f(\BFx)=\mathcal{W}(\Bx),\quad \BFx\in\Real^d.
\end{equation}
Here, $\CalW(\Bx)$ is a Gaussian white noise process with unit spectral, and $\bigtriangleup=\sum_{j=1}^d\partial_{x_j^2}$ is the Laplacian operator. This SPDE representation is leveraged in various contexts for scalable computation and/or approximation, e.g., within Integrated Nested Laplace Approximations (INLA; \cite{rue2009approximate}). Furthermore, from \eqref{eq:SPDE}, it follows that $\mathcal{W}_{\nu+2}(\Bx)$ is a solution to the SPDE:
\begin{equation}
    (\kappa^2-\bigtriangleup)f(\BFx)=\mathcal{W}_\nu(\Bx),\quad \BFx\in\Real^d.
    \label{eq:SPDE2}
\end{equation}
We will use a modification of this representation later to derive the BdryMat\'ern GP model.




\subsection{Existing Work on Boundary-Integrated GPs}

Next, we briefly review existing literature on boundary-integrated GPs, and highlight their potential limitations. This literature can be broadly categorized into projection-based, eigenfunction-based and ADF-based approaches; we discuss each of these categories below.


\textbf{Projection-based approaches:} One class of boundary-integrated GPs (see \cite{vernon2019known,jackson2023efficient}) makes use of careful projections onto known boundaries on the input space $\mathcal{X}$. Suppose $\mathcal{X} = [0,1]^d$ is the unit hypercube, and further suppose we know the Dirichlet boundaries of $f$ along the $(d-1)$-dimensional hyperplane $\mathcal{B} = \{\Bx: x_1 = 0\}$. Let $f(\cdot) \sim \text{GP}\{\mu,k(\Bx,\Bx')\}$, where the covariance kernel takes the product form $k(\mathbf{x},\mathbf{x}') = \sigma^2 \prod_{l=1}^d k_l(x_l,x_l')$. Consider the prediction of $f(\mathbf{x})$ at a new point $\mathbf{x}$, conditional on evaluations of $f$ on a finite set of points $\mathbf{B} = \{\mathbf{b}_1, \cdots, \mathbf{b}_m,\Bx_{\mathcal{B}}\} \subset \mathcal{B}$ on the boundary, where $\Bx_{\mathcal{B}}$ is the projection of $\Bx$ onto $\mathcal{B}$. Then \cite{vernon2019known} shows that the posterior predictive distribution $f(\mathbf{x})|f(\mathbf{B})$ takes the form $\mathcal{N}\{\mu_{\mathcal{B}}(\mathbf{x}),k_{\mathcal{B}}(\Bx,\Bx)\}$, where:
\begin{align}
\begin{split}
\mu_{\mathcal{B}}(\mathbf{x}) &= \mu + k_1(x_1)[f(\Bx_{\mathcal{B}}) - \mu],\\
k_{\mathcal{B}}(\Bx,\Bx')&= \sigma^2 [k_1(x_1-x_1')-k_1(x_1)k_1(x_1')] \prod_{l=2}^d k_l(x_l,x_l').
\label{eq:proj}
\end{split}
\end{align}

Note that the above predictive equations depend only on the projected point $\Bx_{\mathcal{B}}$ and not the other boundary points $\mathbf{b}_1, \cdots, \mathbf{b}_m$. Thus, the use of the projected point $\Bx_{\mathcal{B}}$ in the boundary-integrated mean function and an appropriately modified covariance function is sufficient for capturing boundary information in this setting. One can then directly use $\mu_{\mathcal{B}}$ and $k_{\mathcal{B}}$ within the GP equations \eqref{eq:kriging} for probabilistic surrogate modeling. This projection-based framework can be extended (see \cite{jackson2023efficient}) to accommodate multiple hyperplane boundaries of varying dimensions on the unit hypercube domain $\mathcal{X} = [0,1]^d$.

Despite its elegance, there are several limitations for such approaches. First, it is unclear whether projection-based methods extend for the integration of Neumann and Robin boundaries, which are present in many problems \citep{dalton2024boundary}. Second, the above framework may not extend for the case of irregular input spaces $\mathcal{X}$, which arises in complex applications. We will address both limitations via the proposed BdryMat\'ern GP.


\textbf{Eigenfunction-based approaches:} Another class of boundary-integrated GPs (see \cite{solin2019know,solin2020hilbert,gulian2022gaussian}) makes use of an eigenfunction approach for kernel construction. The key idea is to employ basis functions that inherently satisfy the provided boundary conditions. Here, such works consider the setting of zero (i.e., homogeneous) boundary conditions, thus $\mu_{\mathcal{B}}(\mathbf{x})$ is set as 0. The covariance kernel $k$ is then constructed as:
\[
k_{\mathcal{B}}(\mathbf{x},\mathbf{x}') = \sum_{j=1} \lambda_j \phi_j(\mathbf{x}) \phi_j(\mathbf{x}'),
\]
where \(\{\phi_j\}\) are basis functions that satisfy known boundary conditions. For rectangular domains $\mathcal{X}$, these basis functions can be taken as the eigenfunctions of the Laplace operator under Dirichlet boundary conditions, which can be derived in closed form \citep{dalton2024boundary}. 

There are two limitations of such eigenfunction-based approaches for boundary-integrated surrogates. First, the derivation of closed-form eigenfunctions is possible only under restricted boundary conditions on simple domains $\mathcal{X}$, and thus may not be possible in broad applications. Second, there is little work investigating smoothness control (of GP sample paths) for such approaches. As mentioned previously, this smoothness control is desirable for surrogate modeling, as one may have prior knowledge on the smoothness of $f$ from the considered PDE system.


\textbf{ADF-based approaches:} The last class of boundary-integrated GPs makes use of the so-called ``approximate distance functions''. Such a term was coined by \cite{dalton2024boundary}, but related concepts have been explored in \cite{tan2018gaussian,li2022improving}. The key idea is to construct the boundary-integrated kernel $k_{\mathcal{B}}$ by modifying a base kernel $k$ (e.g., the Mat\'ern kernel) by a distance function $\varsigma(\Bx)$:
\begin{equation}
k_{\mathcal{B}}(\mathbf{x},\mathbf{x}') = \varsigma(\Bx) \varsigma(\Bx') k(\mathbf{x},\mathbf{x}').
\label{eq:adf}
\end{equation}
Here, $\varsigma(\Bx)$ should be specified such that $\lim_{\Bx \rightarrow \mathbf{b}} \varsigma(\Bx) = 0$ for any boundary point $\mathbf{b} \in \mathcal{B}$. This condition ensures that the prior variance of the GP, i.e., $k_{\mathcal{B}}(\mathbf{x},\mathbf{x})$, converges to zero as $\Bx$ approaches a known boundary, thus reflecting a modeler's certainty of boundary information. Such distance functions can further be used to construct a mean function $\mu_{\mathcal{B}}(\mathbf{x})$ that interpolates known boundaries. \cite{tan2018gaussian} makes use of a product distance function that incorporates Dirichlet boundaries at discrete points. \cite{dalton2024boundary} employs the so-called $R$-function method \citep{rvachev1995r} to construct ADFs for incorporating Dirichlet boundaries on irregular domains $\mathcal{X}$. The same paper further extends this approach (with appropriate modifications on \eqref{eq:adf}) for incorporating Neumann and Robin boundaries.

While such approaches offer a flexible framework for boundary integration, they share a similar limitation with earlier methods: they do not provide a means for investigating smoothness control of the resulting GP sample paths, which is desirable for surrogate modeling. Furthermore, given the need for approximation with complex irregular boundaries, these approaches do not provide analysis of such approximation error. This error analysis is useful in surrogate modeling, where it is important to quantify and aggregate all sources of uncertainty for confident scientific inference. The proposed BdryMat\'ern GP presented next addresses these needs for smoothness control and approximation error analysis.




\section{The BdryMat\'ern GP}
\label{sec:Bdrymat_SPDE}

We now present the proposed BdryMat\'ern GP model. We first outline its kernel specification for the Dirichlet boundary setting via a boundary-modified formulation of the SPDE \eqref{eq:SPDE2}. Using an extension of the Feynman-Kac formula for elliptical PDEs, we prove that this BdryMat\'ern kernel admits a path integral form on a connected domain with a boundary set that is twice-differentiable almost everywhere. Such a form permits efficient kernel approximation via finite-element-based algorithms (see Section \ref{sec:kernel_approx}). We then derive analogous BdryMat\'ern kernels with path integral forms for the Neumann and Robin boundary settings. 

\subsection{The BdryMat\'ern Kernel for Dirichlet boundaries}
\label{sec:BdryGP_Green}

Consider first the Dirichlet boundary setting, where the black-box function $f$ is known on $\mathcal{B} = \partial \mathcal{X}$, the boundary set of a potentially irregular domain $\mathcal{X}$. As before, in the homogeneous setting with $f(\partial \mathcal{X}) = 0$, $\mu_{\mathcal{B}}(\mathbf{x})$ can be set as 0. In the inhomogeneous setting, $\mu_{\mathcal{B}}(\mathbf{x})$ can be specified via a careful interpolation of known boundaries, using, e.g., Remark 10 of \cite{dalton2024boundary}. Thus, without loss of generality, we presume in this subsection (unless otherwise stated) the homogeneous setting of $f(\partial \mathcal{X}) = 0$.

To incorporate this boundary information, we investigate the SPDE system in \eqref{eq:SPDE2} with an additional boundary condition on domain $\mathcal{X}$:
\begin{align}
\begin{split}
    (\kappa^2-\bigtriangleup)f(\BFx)&=\CalW_\nu(\BFx), \quad \BFx\in\CalX, 
    \\
    f(\Bx)&=0,\quad \quad \; \; \; \; \; \Bx\in\partial\CalX,\\
\end{split}
\label{eq:bdry_SPDE}
\end{align}
where $\CalW_\nu$ is the Mat\'ern GP with smoothness $\nu>0$ as defined earlier. Constrained to a general irregular domain $\mathcal{X}$ and not $\Real^d$, the SPDE \eqref{eq:bdry_SPDE} may not have a closed-form solution. However, since $(\kappa^2-\bigtriangleup)$ is a linear operator, a unique solution $f$ of \eqref{eq:bdry_SPDE} should also a GP (see \cite[Chapter 4]{lifshits2012lectures}). Our goal is then to represent the covariance function of $f$ in a form that can be efficiently approximated. To do so, we leverage the Green's function representation \citep{folland2009fourier} of the solution $f$; further details can be found in \cite{sakellaris2021scale}. Here, the Green's function associated with the operator $(\kappa^2-\bigtriangleup)$ on domain $\mathcal{X}$ is defined as a function $G(\cdot)$ satisfying:
\begin{align}
\begin{split}
    (\kappa^2-\bigtriangleup)\int_\CalX G(\BFx,\Bs)\phi(\Bs) d\Bs&=\int_\CalX G(\BFx,\Bs)(\kappa^2-\bigtriangleup)\phi(\Bs) d\Bs=\phi(\Bs),\\
    G(\BFx,\Bx')&=G(\BFx',\Bx),\\
    G(\Bx,\cdot)&=0, \quad \quad \Bx\in \partial\CalX.
\end{split}
\label{eq:greendef}
\end{align}
for any integrable function $\phi : \mathcal{X} \rightarrow \Real$.

While the Green's function $G$ has no closed-form solution here, it does provide a useful integral representation of $f$ as:
\begin{equation}
\label{eq:Green_SPDE}
    f(\BFx)=\int_{\CalX}G(\BFx,\Bs)\CalW_\nu(\Bs)d\Bs.
\end{equation}
Using such a representation, we can write the covariance function of $f$ as:
\begin{align}
    k_{\nu+2,\mathcal{B}}(\BFx,\BFx') 
    &= \E\left[\int_{\CalX} G(\BFx, \Bs) \CalW_\nu(\Bs) \, d\Bs \int_{\CalX} G(\BFx', \bold{s}') \CalW_\nu(\bold{s}') \, d\bold{s}' \right] \nonumber \\
    &= \int_{\CalX \times \CalX} G(\BFx, \Bs) k_\nu(\Bs, \bold{s}') G(\bold{s}', \BFx') \, d\Bs \, d\bold{u}.
    \label{eq:kernel_green}
\end{align}
Here, $k_\nu(\Bs, \bold{s}')$ is the Mat\'ern kernel from \eqref{eq:matern}, and the second line follows from applying Fubini's theorem to exchange the integral and expectation. The notation $k_{\nu+2,\mathcal{B}}(\BFx,\BFx')$ will be explained later. Combining \eqref{eq:greendef} and \eqref{eq:kernel_green}, we see that if we let the linear operator \((\kappa^2 - \Delta)\) act on one variable of \(k_{\nu+2,\mathcal{B}}(\BFx, \BFx')\) and then act on the other variable, the resulting function reduces to the Mat\'ern kernel \(k_\nu(\BFx, \BFx')\).



Using \eqref{eq:kernel_green}, we can then leverage the Feynman-Kac formula for Dirichlet boundaries \citep{ito1957fundamental,stroock1971diffusion,papanicolaou1990probabilistic} to derive a path integral form for the new kernel $k_{\nu+2,\mathcal{B}}(\BFx, \BFx')$. This is given by the following theorem:
\begin{theorem}[BdryMat\'ern kernel for Dirichlet boundary]
\label{thm:bdryMatern_path_integral_dirichlet}
    Suppose the domain $\CalX$ is connected, with its boundary twice-differentiable almost everywhere. Then the kernel $k_{\nu+2,\mathcal{B}}(\BFx,\BFx')$ in \eqref{eq:kernel_green} takes the path integral form:
    \begin{align}
    \begin{split}
        k_{\nu+2,\mathcal{B}}(\BFx,\BFx')
        =  &k_{\nu+2}(\BFx,\BFx')-\E\left[k_{\nu+2}(\BFx,\bold{B}'_{\tau})e^{-\kappa^2\tau}\right]\\
        &-\E\left[k_{\nu+2}(\bold{B}_{\gamma},\BFx')e^{-\kappa^2\gamma}\right]+\E\left[k_{\nu+2}(\bold{B}_{\gamma},\bold{B}'_{\tau})e^{-\kappa^2(\tau+\gamma)}\right], \quad \mathbf{x},\mathbf{x}' \in \mathcal{X}.
        \end{split}
        \label{eq:bdryMatern_path_integral_dirichlet}
    \end{align}
    Here, $\bold{B}_t$ and $\bold{B}_s'$ are mutually independent $d$-dimensional Brownian motions initialized at $\bold{B}_0=\Bx$ and $\bold{B}'_0=\Bx'$, respectively, and $\tau$ and $\gamma$ are the hitting times to the boundary of the Brownian motions $\bold{B}_t$ and $\bold{B}_s'$, i.e.:
    \begin{align*}
        &\tau=\inf\{t:\bold{B}_t\in\partial\CalX\},\quad \gamma=\inf\{s:\bold{B}'_s\in\partial\CalX\}.
    \end{align*}
\end{theorem}




\begin{proof}
Our proof first decomposes \eqref{eq:kernel_green} into four terms corresponding to the terms in \eqref{eq:bdryMatern_path_integral_dirichlet}. For each term, we then rewrite its Green's function representation in path integral integral form. Recall that the Mat\'ern kernel $k_\nu$ admits the spectral representation \cite{Wendland10}:
\begin{equation}
    \label{eq:pf_path_integral_1}
    k_\nu(\BFx,\BFx')=\CalF^{-1}[{(\kappa^2+\|\Bomega\|^2)^{-\alpha}}](\BFx-\BFx')=C_{\CalF^{-1}}\int_{\Real^d}\frac{e^{ \im (\BFx-\BFx')^T\Bomega }}{(\kappa^2+\|\Bomega\|^2)^{\alpha}}d\Bomega,
\end{equation}
where $\mathcal{F}$ is the Fourier transform, and $C_{\CalF^{-1}}$ is the inverse Fourier transform constant that depends only on dimension $d$. Substituting \eqref{eq:pf_path_integral_1} into  \eqref{eq:kernel_green} and applying Fubini's theorem, we get:
\begin{align}
    k_{\nu+2,\mathcal{B}}(\BFx,\BFx')=&C_{\CalF^{-1}}\int_{\Real^d}\left[\int_{\CalX\times\CalX}\frac{G(\BFx,\Bs)e^{ \im (\Bs-\Bu)^T\Bomega}G(\BFx',\Bu)}{(\kappa^2+\|\Bomega\|^2)^{\alpha}}\right]d\Bs d\Bu d\Bomega\nonumber\\
    =& C_{\CalF^{-1}}\int_{\Real^d}(\kappa^2+\|\Bomega\|^2)^{-\alpha}\left[\int_{\CalX}G(\BFx,\Bs)e^{ \im\Bs^T\Bomega} d\Bs\right]\overline{\left[ \int_{\CalX}G(\BFx',\Bu)e^{ \im\Bu^T\Bomega}d\Bu\right]} d\Bomega\nonumber\\
    =& C_{\CalF^{-1}}\int_{\Real^d}\psi(\Bomega) g_{\Bomega}(\BFx)\overline{g_{\Bomega}(\BFx')}d\Bomega\label{eq:pf_path_integral_2},
\end{align}
where $\psi(\Bomega)=(\kappa^2+\|\Bomega\|^2)^{-\alpha}$. Here, following the definition of Green's function, $g_{\Bomega}$ is the solution to the following PDE with zero boundary condition on $\partial\CalX$:
\begin{align*}
    (\kappa^2-\bigtriangleup)g_{\Bomega}=e^{ \im \BFx^T\Bomega},\quad \BFx\in\CalX, \quad g_{\Bomega}(\partial\CalX)=0. 
\end{align*}


From this, $g_{\Bomega}$ can then be decomposed into a homogeneous solution $h_{\Bomega}$ and a particular solution $v_{\Bomega} = h_{\Bomega} - g_{\Bomega}$ such that:
\begin{align}
    &(\kappa^2-\bigtriangleup) h_{\Bomega}(\BFx)=e^{ \im \BFx^T\Bomega},\quad  \BFx\in\Real^d, \label{eq:pf_path_integral_3}\\
    &(\kappa^2-\bigtriangleup) v_{\Bomega}(\BFx)=0,\quad \quad \quad \BFx\in\CalX, \quad v_{\Bomega}(\partial\CalX)=h_\omega(\partial\CalX).\label{eq:pf_path_integral_4}
\end{align}
Substituting the decomposition $g_{\Bomega}=h_{\Bomega}-v_{\Bomega}$ into \eqref{eq:pf_path_integral_2}, the kernel $k_{\nu+2,\mathcal{B}}(\BFx,\BFx')$ can then be decomposed into four different terms:
\begin{equation}
    k_{\nu+2,\mathcal{B}}(\BFx,\BFx')=K_{hh}(\BFx,\BFx')-K_{hv}(\BFx,\BFx')-K_{hv}(\BFx',\BFx)+K_{vv}(\BFx,\BFx'),\label{eq:pf_path_integral_5}
\end{equation}
where:
\begin{align}
\begin{split}
    &K_{hh}(\BFx,\BFx')=C_{\CalF^{-1}}\int_{\Real^d}\psi(\Bomega)h_{\Bomega}(\BFx)\overline{h_{\Bomega}(\BFx')}d\Bomega,\\
    &K_{hv}(\BFx,\BFx')=C_{\CalF^{-1}}\int_{\Real^d}\psi(\Bomega)h_{\Bomega}(\BFx)\overline{v_{\Bomega}(\BFx')}d\Bomega,\\
    &K_{vv}(\BFx,\BFx')=C_{\CalF^{-1}}\int_{\Real^d}\psi(\Bomega)v_{\Bomega}(\BFx)\overline{v_{\Bomega}(\BFx')}d\Bomega.
    \label{eq:khh_int}
\end{split}
\end{align}

Consider the first term $K_{hh}(\BFx,\BFx')$. From \eqref{eq:pf_path_integral_3}, we see that $h_{\Bomega}$ is the inverse Fourier transform of the Green's function of the operator $(\kappa^2-\bigtriangleup)$ over the whole domain $\Real^d$. From our direct calculations, we have:
\begin{equation}
    h_{\Bomega}(\BFx)=\frac{e^{\im \BFx^T\Bomega}}{\kappa^2+\|\Bomega\|^2}\label{eq:pf_path_integral_6}.
\end{equation}
Substituting \eqref{eq:pf_path_integral_6} into the integral representation of $K_{hh}$, we can show that $K_{hh}$ corresponds to the Mat\'ern-$(\nu+2)$ kernel, since:
\begin{equation}
    K_{hh}(\BFx,\BFx')=C_{\CalF^{-1}}\int_{\Real^d}\frac{e^{\im(\Bx-\Bx')^T\Bomega}}{(1+\|\Bomega\|^2)^{\alpha+2}}d\Bomega=k_{\nu+2}(\BFx,\BFx'). \label{eq:pf_path_integral_7}
\end{equation}

Consider the next three terms in \eqref{eq:pf_path_integral_5}. Note that $v_{\Bomega}$ has no closed-form solution in general. However, since it is the solution of the elliptical PDE \eqref{eq:pf_path_integral_4} and given the twice-differentiability of the boundary set $\partial \mathcal{X}$ almost everywhere, we can apply the Feynman-Kac formula for Dirichlet boundaries \citep{Chung1981} to obtain the representation:
\begin{equation}
    \label{eq:pf_path_integral_8}
    v_{\Bomega}(\BFx)=\E\left[h_{\Bomega}(\bold{B}_\gamma)e^{-\kappa^2\gamma}|\bold{B}_0=\BFx\right]=\frac{\E\left[e^{\im \bold{B}_\gamma^T\Bomega}e^{-\kappa^2\gamma}|\bold{B}_0=\BFx\right]}{\kappa^2+\|\Bomega\|^2}.
\end{equation} 
Substituting \eqref{eq:pf_path_integral_8} into the integral representation of $K_{hv}$ and $K_{vv}$ in \eqref{eq:khh_int}, we can apply Fubini's theorem for Gaussian measures to obtain: 
\small
\begin{equation}
\label{eq:pf_path_integral_9}
    K_{hv}(\BFx,\BFx')=\E\left[C_{\CalF^{-1}}\int_{\Real^d}\frac{e^{\im(\Bx-\bold{B}_\tau)^T\Bomega}}{(1+\|\Bomega\|^2)^{\alpha+2}}d\Bomega e^{-\kappa^2\tau}\bigg|\bold{B}_0=\Bx'\right]=\E\left[k_{\nu+2}(\BFx,\bold{B}_\tau)e^{-\kappa^2\tau}|\bold{B}_0=\Bx'\right],
\end{equation}
\normalsize
and:
\begin{align}
\label{eq:pf_path_integral_10}
\begin{split}
    K_{vv}(\BFx,\BFx')=&\E\left[C_{\CalF^{-1}}\int_{\Real^d}\frac{e^{\im(\bold{B}_\tau-\bold{B}'_\gamma)^T\Bomega}}{(1+\|\Bomega\|^2)^{\alpha+2}}d\Bomega e^{-\kappa^2(\tau+\gamma)}\bigg|\bold{B}_0=\Bx, \bold{B}'_0=\Bx'\right]\\
    =&\E\left[k_{\nu+2}(\bold{B}_{\tau},\bold{B}'_\gamma)e^{-\kappa^2(\tau+\gamma)}|\bold{B}_0=\Bx,\bold{B}'_0=\Bx'\right].
\end{split}
\end{align}
The proof is then complete by substituting \eqref{eq:pf_path_integral_7}-\eqref{eq:pf_path_integral_10} into \eqref{eq:pf_path_integral_5}.

\end{proof}

In what follows, we define the kernel $k_{\nu,\mathcal{B}}(\BFx, \BFx')$ in \eqref{eq:bdryMatern_path_integral_dirichlet} as the \textit{BdryMat\'ern kernel} with smoothness parameter $\nu > 2$ for \textit{Dirichlet} boundaries, with process variance $\sigma^2$ and inverse length-scale $\kappa > 0$ inherited from the base Mat\'ern kernel $k_\nu(\BFx, \BFx')$. Here, the restriction of $\nu > 2$ is needed for the Feynman-Kac representation \eqref{eq:pf_path_integral_8}; we thus presume $\nu > 2$ for the BdryMat\'ern kernel for the remainder of the paper. With the mean function $\mu_{\mathcal{B}}(\mathbf{x})$ specified as described at the start of the subsection, the BdryMat\'ern GP for Dirichlet boundaries is then defined as $f \sim \text{GP}\{\mu_{\mathcal{B}}(\mathbf{x}),k_{\nu,\mathcal{B}}(\BFx, \BFx')\}$; such a model embeds the desired boundary information. We will discuss later in Section \ref{sec:kernel_approx} how the posterior predictive distribution of this model can be efficiently approximated.


A key appeal of the BdryMat\'ern kernel is that it not only ensures the resulting GP sample paths satisfy known Dirichlet boundary conditions, but also provides smoothness control of such paths via the smoothness parameter $\nu$. The latter is analogous to the smoothness control from the Mat\'ern kernel \eqref{eq:matern}. The following theorem proves this for the zero-mean BdryMat\'ern GP given the homogeneous Dirichlet boundary condition $f(\partial \mathcal{X}) = 0$:
\begin{theorem}[Smoothness control for Dirichlet BdryMat\'ern GP]
Suppose the domain $\CalX$ is connected, with its boundary twice-differentiable almost everywhere. Consider the zero-mean BdryMat\'ern GP $f(\cdot) \sim \textup{GP}\{0,k_{\nu,\mathcal{B}}(\BFx, \BFx')\}$, where $k_{\nu,\mathcal{B}}(\BFx, \BFx')$ is the BdryMat\'ern kernel defined in \eqref{eq:bdryMatern_path_integral_dirichlet}. Then:
\begin{enumerate}[label=(\alph*)]
    \item Its sample paths satisfy the zero Dirichlet boundary condition $f(\partial \mathcal{X}) = 0$ almost surely,
    \item Its sample paths are \textup{($\lceil\nu\rceil$-1)}-differentiable almost surely, where $\lceil\cdot\rceil$ denotes the ceiling function.
\end{enumerate}
\label{thm:smoothdir}
\end{theorem}
\noindent The proof of this theorem is provided in Appendix \ref{pf:smoothdir}. This theorem holds analogously for inhomogeneous Dirichlet boundary information on $\partial \mathcal{X}$ (i.e., $f(\partial \mathcal{X}) \neq 0$), as long as the specified mean function $\mu_{\mathcal{B}}(\Bx)$ (see Remark 10 of \cite{dalton2024boundary}) satisfies the same smoothness properties from part (b).

\subsection{The BdryMat\'ern Kernel for Neumann and Robin Boundaries}
Consider next the Neumann and Robin boundary setting. Homogeneous Robin boundaries typically presume (see \cite{zwillinger2021handbook}) known boundaries of the form:
\begin{equation}
\frac{\partial f(\Bx)}{\partial \bold{n}} + c(\Bx)f(\Bx) = 0, \quad \quad \Bx \in \partial \mathcal{X},
\label{eq:robinbdry}
\end{equation}
where $c(\mathbf{x}) \geq 0$ is given and $\bold{n}$ is the inward unit
normal vector on $\partial\CalX$. Again, without loss of generality, we presume in this subsection the homogeneous setting; the inhomogeneous setting can be accommodated via a careful specification of the mean function $\mu_{\mathcal{B}}(\mathbf{x})$ via interpolation; see Remark 10 of \cite{dalton2024boundary}. Note that Neumann boundaries correspond to a specific case of Robin boundaries with $c(\mathbf{x})=0$; the following development on Robin boundaries thus holds for Neumann boundaries as well. 


To incorporate the above Robin boundary information, we investigate the SPDE system in \eqref{eq:SPDE2} with an additional boundary condition on domain $\mathcal{X}$:
\begin{align}
\label{eq:bdry_SPDE3}
\begin{split}
    (\kappa^2-\bigtriangleup)f(\BFx) & =\CalW_\nu(\BFx), \quad \BFx\in\CalX, 
    \\
    \frac{\partial f(\Bx)}{\partial \bold{n}}+c(\mathbf{x})f(\Bx) & =0,\quad \quad \quad  \Bx\in\partial \CalX,
\end{split}
\end{align}
where $\mathcal{W}_\nu$ is the Mat\'ern GP with smoothness parameter $\nu > 0$. Again, the SPDE \eqref{eq:bdry_SPDE3} may not have a closed-form solution on an irregular domain $\mathcal{X}$. We thus apply a similar Green's function approach as for the earlier Dirichlet boundary setting. Here, the Green's function $G(\cdot$) for operator $(\kappa^2 - \bigtriangleup)$ on $\mathcal{X}$ satisfies:
\begin{align}
\begin{split}
    (\kappa^2-\bigtriangleup)\int_\CalX G(\BFx,\Bs)\phi(\Bs) d\Bs & =\int_\CalX G(\BFx,\Bs)(\kappa^2-\bigtriangleup)\phi(\Bs) d\Bs=\phi(\Bs),\\
    G(\BFx,\Bx')&=G(\BFx',\Bx),\\
    \partial_{\bold{n}}G(\Bx,\cdot)+c(\mathbf{x})G(\Bx,\cdot)&=0, \quad \Bx\in \partial\CalX,
\end{split}
\label{eq:greendef2}
\end{align}
for any integrable function $\phi: \mathcal{X} \rightarrow \Real$.

Again, while $G$ has no closed-form solution here, it does provide a useful representation of a solution $f$ of \eqref{eq:bdry_SPDE3}, which should be a GP since $(\kappa^2 - \bigtriangleup)$ is a linear operator. In particular, the covariance function of such a solution $f$, again denoted $k_{\nu+2,\mathcal{B}}(\BFx, \BFx')$, can be written in the same form as \eqref{eq:kernel_green}. Using this representation with the Feynman-Kac formula for Robin boundaries \citep{ito1957fundamental,stroock1971diffusion,papanicolaou1990probabilistic}, we can then derive a path integral form for the kernel $k_{\nu+2,\mathcal{B}}(\BFx, \BFx')$. This is given by the following theorem:
\begin{theorem}[BdryMat\'ern kernel for Robin boundary]
\label{thm:bdryMatern_path_integral_robin}
    Suppose the domain $\CalX$ is connected, with its boundary twice-differentiable almost everywhere. Then the kernel $k_{\nu+2,\mathcal{B}}(\BFx, \BFx')$ defined above takes the path integral form:
    \begin{align}
    \begin{split}
        k_{\nu+2,\mathcal{B}}(\BFx,\BFx')
        =  &k_{\nu+2}(\BFx,\BFx')-\E\left[\int_0^{\infty} k_{\nu+2}(\BFx,\bold{B}'_{t})e^{-\kappa^2t-\int_0^tc(\bold{B}'_{s})dL_s}dL_t\right]\\
        &-\E\left[\int_0^{\infty} k_{\nu+2}(\bold{B}_{\tau},\BFx')e^{-\kappa^2\tau-\int_0^\tau c(\bold{B}_{s})dL_s}dL_\tau\right]\\
        &+\E\left[\int_0^{\infty}\int_0^{\infty} k_{\nu+2}(\bold{B}_{\tau},\bold{B}'_{t})e^{-\kappa^2(t+\tau)-\int_0^\tau c(\bold{B}_{s})dL_s-\int_0^tc(\bold{B}'_{s})dL_s}dL_tdL_\tau\right].
        \end{split}
        \label{eq:bdryMatern_path_integral_robin}
    \end{align}
    Here, $\bold{B}_t$ and $\bold{B}_s'$ are mutually independent $d$-dimensional Brownian motions initialized at $\bold{B}_0=\Bx$ and $\bold{B}'_0=\Bx'$, respectively, and $L_t$ and $L_\tau$ are the so-called boundary local times defined as:
    \begin{equation}
    L_t=\lim_{\delta\to 0}\frac{\int_0^t\ind_{\{\bold{B}_s\in\CalX_\delta\}}ds}{\delta},\quad L_\tau =\lim_{\delta\to 0}\frac{\int_0^\tau \ind_{\{\bold{B}'_s\in\CalX_\delta\}}ds}{\delta},
        \label{eq:local}
    \end{equation}
    where $\CalX_\delta = \{\BFx\in\CalX: \|\BFx-\partial\CalX\|_2 \leq\delta\}$ is a region of width $\delta$ around the boundary $\partial\CalX$.
\end{theorem}

\begin{proof}

This proof is similar to that for Theorem \ref{thm:bdryMatern_path_integral_dirichlet}. The kernel solution of the form \eqref{eq:kernel_green} can again be decomposed into the four terms:
\begin{equation}
    k_{\nu+2,\mathcal{B}}(\BFx,\BFx')=K_{hh}(\BFx,\BFx')-K_{hv}(\BFx,\BFx')-K_{hv}(\BFx',\BFx)+K_{vv}(\BFx,\BFx'),
    \label{eq:pf_path_integral_robin_2}
\end{equation}
where:
\begin{align}
\begin{split}
    &K_{hh}(\BFx,\BFx')=C_{\CalF^{-1}}\int_{\Real^d}\psi(\Bomega)h_{\Bomega}(\BFx)\overline{h_{\Bomega}(\BFx')}d\Bomega,\\
    &K_{hv}(\BFx,\BFx')=C_{\CalF^{-1}}\int_{\Real^d}\psi(\Bomega)h_{\Bomega}(\BFx)\overline{v_{\Bomega}(\BFx')}d\Bomega,\\
    &K_{vv}(\BFx,\BFx')=C_{\CalF^{-1}}\int_{\Real^d}\psi(\Bomega)v_{\Bomega}(\BFx)\overline{v_{\Bomega}(\BFx')}d\Bomega.
    \label{eq:krep}
\end{split}
\end{align}
Here, $h_{\Bomega}(\BFx)$ is the homogeneous solution of \eqref{eq:pf_path_integral_3}, and $v_{\Bomega}(\BFx)$ is the solution to:
\begin{align}
\begin{split}
    (\kappa^2-\bigtriangleup) v_{\Bomega}(\BFx)&=0,\quad \BFx\in\CalX,\\
    \frac{\partial v_{\Bomega}(\BFx)}{\partial\bold{n}}+c(\BFx)v_{\Bomega}(\BFx)&=0,\quad \BFx\in\partial\CalX. \label{eq:pf_path_integral_robin_1}
\end{split}
\end{align}

As before, the first term in \eqref{eq:pf_path_integral_robin_2} reduces to the Mat\'ern kernel $k_{\nu+2,\mathcal{B}}(\mathbf{x},\mathbf{x}')$. For the remaining terms, note that $v_{\Bomega}$ again has no closed-form solution in general. However, as it is the solution to the elliptical PDE \eqref{eq:pf_path_integral_robin_1} and given the twice-differentiability of $\partial \mathcal{X}$ almost everywhere, one can apply the Feynman-Kac formula for Robin boundaries \cite[Corollary 4.4]{papanicolaou1990probabilistic} to obtain the representation:
\begin{equation}
    v_{\Bomega}(\BFx)=\frac{\E\left[\int_0^\infty \exp\{-\kappa^2t-\int_0^tc(\bold{B}_s)dL_s +i\bold{B}_t^T{\Bomega}\}dL_t|\bold{B}_0=\BFx\right]}{\kappa^2+\|\Bomega\|^2}. \label{eq:pf_path_integral_robin_3}
\end{equation}
Plugging this expression into \eqref{eq:krep}, we can then apply Fubini's theorem (in a similar fashion to \eqref{eq:pf_path_integral_9} and \eqref{eq:pf_path_integral_10}) to exchange the order of integration and expectation, under our assumption that $c(\mathbf{x}) \geq 0$. This yields:
\begin{equation}
    \label{eq:pf_path_integral_robin_4}
    K_{hv}(\BFx,\BFx')=\E\left[\int_0^{\infty} k_{\nu+2}(\bold{B}_{\tau},\BFx)e^{-\kappa^2\tau-\int_0^\tau c(\bold{B}_{s})dL_s}dL_\tau\big|\bold{B}_0=\BFx'\right]
\end{equation}
and:
\begin{align}
\label{eq:pf_path_integral_robin_5}
    K_{vv}(\BFx,\BFx')=\E\left[\int_0^{\infty}\int_0^{\infty} k_{\nu+2}(\bold{B}_{\tau},\bold{B}'_{t})e^{-\kappa^2(t+\tau)-\int_0^\tau c(\bold{B}_{s})dL_s-\int_0^tc(\bold{B}'_{s})dL_s}dL_tdL_\tau\right].
\end{align}
The result is then proven by plugging in \eqref{eq:pf_path_integral_robin_4} and \eqref{eq:pf_path_integral_robin_5} into \eqref{eq:pf_path_integral_robin_2}.
\end{proof}

We define the kernel $ k_{\nu,\mathcal{B}}(\BFx,\BFx')$ in \eqref{eq:bdryMatern_path_integral_robin} as the \textit{BdryMat\'ern kernel} with smoothness $\nu > 2$ for \textit{Robin} boundaries, with process variance $\sigma^2$ and inverse length-scale $\kappa > 0$ inherited from the base Mat\'ern kernel $k_\nu(\BFx, \BFx')$. (In what follows, the notation $k_{\nu,\mathcal{B}}(\BFx,\BFx')$ does not distinguish between Dirichlet and Robin boundaries for brevity; this should be clear from context.) Again, the restriction of $\nu > 2$ is needed for the Feynman-Kac representation \eqref{eq:pf_path_integral_8}. The BdryMat\'ern GP for Robin boundaries then takes the form $f(\cdot) \sim \text{GP}\{\mu_{\mathcal{B}}(\Bx),k_{\nu,\mathcal{B}}(\BFx,\BFx')\}$, where $\mu_{\mathcal{B}}(\Bx)$ is the mean function specified at the start of the subsection. We will discuss in Section \ref{sec:kernel_approx} how to efficiently approximate the posterior predictive distribution of this model.

As before, the above BdryMat\'ern GP not only guarantees its sample paths satisfy known Robin boundary conditions, but also ensures the smoothness of such paths can be controlled by the parameter $\nu$. The following theorem shows this for the zero-mean BdryMat\'ern GP given a zero Robin boundary condition:
\begin{theorem}[Smoothness control for Robin BdryMat\'ern GP]
Suppose the domain $\CalX$ is connected, with its boundary twice-differentiable almost everywhere. Consider the zero-mean BdryMat\'ern GP for Robin boundaries $f(\cdot) \sim \textup{GP}\{0,k_{\nu,\mathcal{B}}(\BFx, \BFx')\}$, where $k_{\nu,\mathcal{B}}(\BFx, \BFx')$ is the BdryMat\'ern kernel defined in \eqref{eq:bdryMatern_path_integral_robin}. Then:
\begin{enumerate}[label=(\alph*)]
    \item Its sample paths satisfy the zero Robin boundary condition $\partial f(\Bx)/\partial \bold{n} + c(\Bx) f(\Bx) = 0$ on $\Bx \in \partial \mathcal{X}$ almost surely,
    \item Its sample paths are \textup{($\lceil\nu\rceil$-1)}-differentiable almost surely.
\end{enumerate}
\label{thm:smooth}
\end{theorem}
\noindent The proof of this theorem is provided in Appendix \ref{pf:smoothrobin}. This theorem again holds analogously for inhomogeneous Robin boundaries on $\partial \mathcal{X}$, as long as the specified mean function $\mu_{\mathcal{B}}(\Bx)$ (see Remark 10 of \cite{dalton2024boundary}) satisfies the required smoothness properties.

For general irregular domains $\mathcal{X}$, the BdryMat\'ern kernel (both for the Dirichlet setting \eqref{eq:bdryMatern_path_integral_dirichlet} and the Robin setting \eqref{eq:bdryMatern_path_integral_robin}) cannot be obtained in closed form and thus requires careful approximation. We present next a finite-element-method approach that performs this kernel approximation to facilitate efficient posterior predictions with reliable error analysis.


\section{Approximation of the Posterior Predictive Distribution}
\label{sec:kernel_approx}


We first consider the approximation of the BdryMat\'ern kernel, which is needed within the GP equations \eqref{eq:kriging} for posterior predictions. This kernel approximation approach requires two steps. First, we construct a coupled Monte Carlo estimator for estimating the Gram kernel matrix for the BdryMat\'ern kernel. Next, we employ a finite-element-method approach that leverages such an estimator for efficient approximation of the posterior predictive distribution with error analysis. In what follows, we again presume the domain $\mathcal{X}$ is connected with its boundary twice-differentiable almost everywhere, as required for defining the BdryMat\'ern GP.


\subsection{Kernel Matrix Estimator}\label{sec:kermat}

Consider first the BdryMat\'ern kernel $k_{\nu,\mathcal{B}}$ in \eqref{eq:bdryMatern_path_integral_dirichlet} for the Dirichlet setting. Let $\mathbf{U} = \{\mathbf{u}_1, \cdots, \mathbf{u}_m\}$ denote a set of $m$ points on $\mathcal{X}$, and consider the approximation of its kernel matrix $\mathbf{K}:=\mathbf{K}(\mathbf{U},\mathbf{U}) = [k_{\nu,\mathcal{B}}(\mathbf{u}_i,\mathbf{u}_j)]_{i,j=1}^m$. A naive approach might be to \textit{separately} approximate each matrix entry $K_{ij} = k_{\nu,\mathcal{B}}(\mathbf{u}_i, \mathbf{u}_j)$ via the Monte Carlo (MC) estimator:
\small
   \begin{align}
    \begin{split}
        \tilde{K}_{ij}
        =  &k_{\nu}(\Bu_i,\Bu_j)-k_{\nu}(\Bu_i,\mathbf{W}'_{ij})e^{-\kappa^2\tau_{ij}} - k_{\nu}(\mathbf{W}_{ij},\Bu_j)e^{-\kappa^2\gamma_{ij}} + k_{\nu}(\mathbf{W}_{ij},\mathbf{W}'_{ij})e^{-\kappa^2(\tau_{ij}+\gamma_{ij})},
        \end{split}
        \label{eq:naive_dir}
    \end{align}
    \normalsize
where $(\mathbf{W}_{ij},\tau_{ij})$ and $(\mathbf{W}'_{ij},\gamma_{ij})$ are the simulated hitting locations and hitting times on $\partial \mathcal{X}$ of mutually independent Brownian motions initialized at $\mathbf{u}_i$ and $\mathbf{u}_j$, respectively. The Brownian motions over different indices $(i,j)$ are also mutually independent. Here, only one Monte Carlo draw is used, for reasons we explain later. However, there is a key drawback to such an approach: the estimated kernel matrix $\tilde{\mathbf{K}} = [\tilde{K}_{ij}]_{i,j=1}^m$ may no longer be positive definite, as the eigenvalues of the Monte Carlo error matrix $\tilde{\mathbf{K}}-\mathbf{K}$ can be larger than the minimum eigenvalue of the desired matrix $\mathbf{K}$. A similar issue arises if an analogous naive MC estimator is used for the BdryMat\'ern kernel \eqref{eq:bdryMatern_path_integral_robin} in the Robin setting.

To address this in the Dirichlet setting, we adopt the following \textit{coupled} Monte Carlo approximation of $\mathbf{K}$. For each point $\mathbf{u}_i$, let $(\mathbf{W}_{i},\tau_{i})$ denote the simulated hitting locations and times on $\partial \mathcal{X}$ of a Brownian motion initialized at $\mathbf{u}_i$; the Brownian motions over different indices $i$ are also mutually independent. We can then estimate each matrix entry $K_{ij}$ via the coupled MC estimator:
   \begin{align}
        \hat{K}_{ij}^{\rm C}
        =  &k_{\nu}(\mathbf{u}_i,\mathbf{u}_j)-k_{\nu}(\Bu_i,\mathbf{W}_{j})e^{-\kappa^2\tau_{j}} - k_{\nu}(\mathbf{W}_{i},\Bu_j)e^{-\kappa^2\tau_{i}}+k_{\nu}(\mathbf{W}_{i},\mathbf{W}_{j})e^{-\kappa^2(\tau_{i}+\tau_{j})}.
       \label{eq:coupled_dir}
    \end{align}
For two entries on the same row $i$ (or same column $i$), its estimators from \eqref{eq:coupled_dir} depend on the same hitting time $\tau_{i}$; such estimators are thus ``coupled'' and dependent.

The following theorem shows that this coupled MC estimator $\hat{\mathbf{K}}^{\rm C}(\mathbf{U},\mathbf{U}) = [\hat{K}_{ij}^{\rm C}]_{i,j=1}^m$ for Dirichlet boundary can preserve positive definiteness almost surely:
\begin{theorem}[Positive Definiteness of Dirichlet Coupled MC Estimator]
\label{thm:empirical_BdryMat_Dirichlet}
   Suppose the points \(\BU = \{\Bu_1, \cdots, \Bu_m\} \subset \mathcal{X}\) are distinct and within the interior of $\mathcal{X}$. Further let $\{(\mathbf{W}_{i},\tau_{i})\}_{i=1}$ denote the simulated hitting locations and times on $\partial \mathcal{X}$ of the mutually independent Brownian motions initialized at the points $\{\mathbf{u}_i\}_{i=1}^m$. Then the coupled matrix estimator $\hat{\mathbf{K}}^{\rm C}(\mathbf{U},\mathbf{U})$ with entries $\hat{K}_{ij}^{\rm C}$ defined in \eqref{eq:coupled_dir} is positive definite almost surely. 
\end{theorem}

\begin{proof}
The positive definiteness of $\hat{\mathbf{K}}^{\rm C}(\mathbf{U},\mathbf{U})$ can be shown as follows. Define $\bold{W}_i := \Bu_i+\boldsymbol{\delta}_i$, and let $\mathbf{v}=(v_1,\cdots,v_m)^T \in \mathbb{R}^m$ be an arbitrary vector. Using the path integral representation of $K_{hh}$, $K_{hv}$, and $K_{vv}$ in \eqref{eq:pf_path_integral_7}-\eqref{eq:pf_path_integral_10}, the quadratic form $\mathbf{v}^T\hat{\BK}^{\rm C}(\mathbf{U},\mathbf{U})\mathbf{v}$ can be written as:
\small
\begin{align}
    \begin{split}
    \mathbf{v}^T\hat{\BK}^{\rm C}(\mathbf{U},\mathbf{U})\mathbf{v}=&C_{\CalF^{-1}}\int_{\Real^d}\frac{\sum_{i,j}v_iv_je^{\im(\mathbf{u}_i-\mathbf{u}_j)^T\Bomega}\left[1-e^{\im \boldsymbol{\delta}_i^T\Bomega-\kappa^2 \tau_i}-e^{-\im \boldsymbol{\delta}_j^T\Bomega-\kappa^2 \tau_j}+e^{\im(\boldsymbol{\delta}_i-\boldsymbol{\delta}_j)^T\Bomega-\kappa^2 (\tau_i+\tau_j)}\right]}{(\kappa^2+\|\Bomega\|^2)^{\alpha+2}}d\Bomega\\
    =& C_{\CalF^{-1}}\int_{\Real^d}\frac{{\mathbf{v}}^T\left[\bold{I}-\text{diag}[e^{\im \boldsymbol{\delta}_i^T\Bomega-\kappa^2 \tau_i}]\right][e^{\im (\mathbf{u}_i-\mathbf{u}_j)^T\Bomega}]_{i,j}\overline{\left[\bold{I}-\text{diag}[e^{\im \boldsymbol{\delta}_i^T\Bomega-\kappa^2 \tau_i}]\right]{\mathbf{v}}}}{(\kappa^2+\|\Bomega\|^2)^{\alpha+2}}d\Bomega\\
    \geq & C_{\CalF^{-1}}\int_{\Real^d}\frac{{\mathbf{v}}^T\left[\bold{I}-\text{diag}[e^{-\kappa^2\tau_*}]\right][e^{\im (\mathbf{u}_i-\mathbf{u}_j)^T\Bomega}]_{i,j}\overline{\left[\bold{I}-\text{diag}[e^{-\kappa^2\tau_*}]\right]{\mathbf{v}}}}{(\kappa^2+\|\Bomega\|^2)^{\alpha+2}}d\Bomega,\label{eq:pf_empirical_BdryMat_1}
\end{split}
\end{align}
\normalsize
where $\tau_*=\min_{i}\tau_i>0$   and $\lambda_*=1-e^{-\kappa^2\tau_*}>0 $ almost surely for any $\mathbf{U}$ in the interior of $\CalX$. We thus have:
\begin{align}
    \left[\bold{I}-\text{diag}[e^{\im \boldsymbol{\delta}_i-\kappa^2\tau_i}]\right]^2> \left[\bold{I}-\text{diag}[e^{-\kappa^2\tau_*}]\right]^2=\lambda_*^2\bold{I}.
\end{align}
Using this with the last line of \eqref{eq:pf_empirical_BdryMat_1}, it follows that:
\begin{equation}
\mathbf{v}^T\hat{\BK}^{\rm C}(\mathbf{U},\mathbf{U})\mathbf{v} \geq \lambda^2_* C_{\CalF^{-1}}\int_{\Real^d}\frac{\|\sum_{i=1}v_ie^{\im \Bu_i^T\Bomega}\|^2}{(\kappa^2+\|\Bomega\|^2)^{\alpha+2}}d\Bomega>0\quad\quad \text{a.s.,} 
\end{equation}
which completes the proof.

\end{proof}

A similar coupled estimator can be used in the Robin setting for the kernel matrix $\mathbf{K} = [K_{ij}]_{i,j=1}^m = k_{\nu ,\mathcal{B}}(\mathbf{U},\mathbf{U})$. For each point $i$, let $\{(\mathbf{B}_{t}^{(i)}, L_{t}^{(i)})\}_{i=1}^m$ be the mutually independent Brownian motions initialized at $\{\mathbf{u}_i\}_{i=1}^m$ along with its boundary local times as defined in \eqref{eq:local}. We then estimate each matrix entry $K_{ij}^{\rm C}$ via the coupled MC estimator:
\begin{align}
\begin{split}
     \hat{K}_{ij}^{\rm C} &= k_{\nu}(\Bx_i,\Bx_j)-  \int_0^{\infty} k_{\nu}(\Bx_i,\bold{B}^{(j)}_{t})e^{-\kappa^2t-\int_0^tc(\bold{B}^{(j)}_{s})dL_s^{(j)}}dL_t^{(j)}  \\
        & \quad - \int_0^{\infty} k_{\nu}(\bold{B}^{(i)}_{\tau},\Bx_j)e^{-\kappa^2\tau-\int_0^\tau c(\bold{B}^{(i)}_{s})dL_s}dL^{(i)}_\tau  \\
        &\quad +  \int_0^{\infty}\int_0^{\infty} k_{\nu}(\bold{B}^{(i)}_{\tau},\bold{B}^{(j)}_{t})e^{-\kappa^2(t+\tau)-\int_0^\tau c(\bold{B}^{(i)}_{s})dL_s^{(i)}-\int_0^tc(\bold{B}^{(j)}_{s})dL_s^{(j)}}dL_t^{(i)}dL_\tau^{(j)}.
        \label{eq:coupled_robin}
\end{split}
\end{align}
Again, note that for matrix entries on the same row $i$ (or same column $i$), its estimators from \eqref{eq:coupled_robin} depend on the same simulated Brownian motions, and thus are dependent.

The following theorem shows that this coupled MC estimator $\hat{\mathbf{K}}^{\rm C}(\mathbf{U},\mathbf{U}) = [\hat{K}^{\rm C}_{ij}]_{i,j=1}^m$ for Robin boundary again preserves positive definiteness almost surely:
\begin{theorem}[Positive Definiteness of Robin Coupled MC Estimator]
\label{thm:empirical_BdryMat_Robin}
   Suppose the points \(\BU = \{\Bu_1, \cdots, \Bu_m\} \subset \mathcal{X}\) are distinct and within the interior of $\mathcal{X}$. Further, let $\{\mathbf{B}_{t}^{(i)}\}_{i=1}^m$ denote the mutually independent standard reflecting Brownian motions (SRBMs; see \cite{harrison1981reflected}) initialized at the points $\{\mathbf{u}_i\}_{i=1}^m$, with $\{L_{t}^{(i)}\}_{i=1}^m$ their corresponding boundary local times as defined in \eqref{eq:local}. If the following inequality is satisfied: 
   \begin{equation}
    \label{eq:empirical_BdryMat_Robin}
    \sup_{\BFx\in\CalX}\E[L_t|\bold{B}_0=\BFx] < (1-e^{-1})\kappa \sqrt{t},
\end{equation}
then the coupled matrix estimator $\hat{\mathbf{K}}^{\rm C}(\mathbf{U},\mathbf{U})$ with entries $\hat{K}_{ij}^{\rm C}$ defined in \eqref{eq:coupled_robin} is positive definite almost surely. 
\end{theorem}

\begin{proof}
Let $\Bu_i+\bold{W}^{(i)}_t=\bold{B}^{(i)}_t$. Similar to the proof for Theorem \ref{thm:empirical_BdryMat_Dirichlet}, note that for any $\bold{v}\in\Real^m$, the quadratic form $\bold{v}^T\BK^{\rm C}(\BU,\BU)\bold{v}$ can be written as:
\begin{align}
\begin{split}
     \bold{v}^T\BK(\BU,\BU)\bold{v}    =&C_{\CalF^{-1}}\int_{\Real^d}\frac{\sum_{i,j}v_iv_je^{\im(\Bu_i-\Bu_j)^T\Bomega}\left[1-S_i-S_j+S_{i}S_j\right]}{(\kappa^2+\|\Bomega\|^2)^{\alpha+2}}d\Bomega\\
    =& C_{\CalF^{-1}}\int_{\Real^d}\frac{{\bold{V}}^T\left[\bold{I}-\text{diag}[S_i]\right][e^{\im (\Bu_i-\Bu_j)^T\Bomega}]_{i,j}\overline{\left[\bold{I}-\text{diag}[S_i]\right]{\bold{V}}}}{(\kappa^2+\|\Bomega\|^2)^{\alpha+2}}d\Bomega,
    \end{split}
\end{align}
where:
\begin{equation*}
    S_i=\int_0^\infty\exp\left\{-\kappa^2t-\int_0^tc(\bold{B}^{(i)}_s)dL_s+\im{\Bomega}^T{\bold{W}^{(i)}_t}\right\}dL_t.
\end{equation*}
Further note that:
\begin{align}
\begin{split}
    \left|S_i\right|=&\left|\int_0^\infty\exp\left\{-\kappa^2t-\int_0^tc(\bold{B}^{(i)}_s)dL_s+\im{\Bomega}^T{\bold{W}^{(i)}_t}\right\}dL_t\right|\\
    \leq & \int_0^\infty\exp\left\{-\kappa^2t-\int_0^tc(\bold{B}^{(i)}_s)dL_s\right\}dL_t\\
    =&\sum_{l=0}^\infty e^{-\kappa^2\vartriangle_t l-\int_0^{\vartriangle_t l}c(\bold{B}^{(i)}_s)dL_s}G_{\vartriangle_t}(\bold{B}^{(i)}_{\vartriangle_t l}),\label{eq:pf_empirical_BdryMat_robin_1}
\end{split}
\end{align}
where the last line of \eqref{eq:pf_empirical_BdryMat_robin_1} follows from the Markov property of the SRBM, which holds for any $\vartriangle_t>0$. Moreover, we have:
\begin{align}
    G_{\vartriangle_t}(\BFx)=&\E\left[\int_0^{\vartriangle_t}\exp\left\{-\kappa^2t-\int_0^tc(\bold{B}_s)dL_s\right\}dL_t|\bold{B}_0=\BFx \right]\leq \sup_{\BFx\in\CalX}\E[L_t|\bold{B}_0=\BFx],\label{eq:pf_empirical_BdryMat_robin_2}
\end{align}
where $\bold{B}_s$ is an SRBM and $L_s$ its boundary local time.

Substituting \eqref{eq:pf_empirical_BdryMat_robin_2} into \eqref{eq:pf_empirical_BdryMat_robin_1}, and using the fact that $c(\BFx)\geq 0$, we have:
\begin{align*}
    |S_i|\leq &\sup_{\BFx\in\CalX}\E[L_t|\bold{B}_0=\BFx]\sum_{l\geq 0}e^{-\kappa^2\vartriangle_t l}
    \leq \frac{\sup_{\BFx\in\CalX}\E[L_t|\bold{B}_0=\BFx]}{1-e^{-\kappa^2 \vartriangle_t}}
    < \frac{(1-e^{-1})\kappa \sqrt{\vartriangle_t} }{1-e^{-\kappa^2 \vartriangle_t}}
    \leq 1,
\end{align*}
where the second-last inequality follows from \eqref{eq:empirical_BdryMat_Robin} and the last inequality is obtained by setting $\vartriangle_t=\kappa^2$. It follows that $\left[\bold{I}-\text{diag}[S_i]\right]^2\geq \lambda_*^2\bold{I}> 0$ and thus $\bold{v}^T\BK^{\rm C}(\BU,\BU)\bold{v} \succ 0$.
\end{proof}

It should be noted that the condition \eqref{eq:empirical_BdryMat_Robin} in Theorem \ref{thm:empirical_BdryMat_Robin} is quite mild. As shown in \cite{hsu1984reflecting,hsu1985probabilistic}, the boundary local time \( L_t \) satisfies:
\[
\sup_{\BFx \in \CalX} \E[L_t | \bold{B}_0 = \BFx] \leq A\sqrt{t},
\]
where \( A \) is a constant dependent only on the domain \( \CalX \). Thus, our condition \eqref{eq:empirical_BdryMat_Robin} easily satisfied for any \( \kappa > {A}(1 - 1/e) \); see \citep{grebenkov2019probability}. 


Finally, it should be noted that the simulation of hitting times and locations (for the Dirichlet setting) and the simulation of boundary local times (for the Robin setting) can be efficiently performed with high accuracy and with theoretical guarantees. Efficient Monte Carlo simulation algorithms for such tasks have been proposed and investigated in the literature; see, e.g., \cite{roger_RSDE,milstein2004stochastic,leimkuhler2023simplest}. In our later experiments, we make use of a MATLAB implementation of the approach in \cite{milstein2004stochastic} for simulating hitting times and local times.





\subsection{Finite-Element-Method Kernel Approximation} \label{sec:fea}

With the kernel matrix estimation method in Section \ref{sec:kermat}, one may be tempted to directly apply this for estimating the kernel matrix $\mathbf{K}(\mathbf{X},\mathbf{X})$ then compute the posterior predictive distribution $[f(\mathbf{x}_{\rm new})|f(\mathbf{X})]$ from Equation \eqref{eq:kriging}. There are, however, several limitations with this approach. First (i), even with $\mathbf{K}(\mathbf{X},\mathbf{X})$ estimated, one still needs to approximate the kernel $\mathbf{k}(\mathbf{x}_{\rm new},\mathbf{X})$ within the predictive equations \eqref{eq:kriging}. Note that each prediction point $\mathbf{x}_{\rm new}$ requires a separate MC approximation, thus the prediction of multiple new points can be computationally expensive. Second (ii), there is little guarantee that the MC-approximated kernel satisfies the desired boundary conditions in any sense. Using the earlier coupled MC estimator, we present next an FEM-based approach that addresses these limitations for effective posterior approximation with error analysis. In what follows, we combine the discussion for both Dirichlet and Robin boundaries; the distinction should be clear from context.

Finite element methods are broadly used for numerically solving differential equation systems \citep{huebner2001finite}. Let $V^Q=\{\phi_q(\mathbf{x})\}_{q=1}^Q$ be the set of basis functions, with each basis having support set $S_q = {\rm supp}\{\phi_q\}$. FEM leverages the so-called finite elements $\{(\phi_q, S_q)\}_{q=1}^Q$ to approximate the system over a desired domain $\mathcal{X}$. The domain for such support sets, denoted $\hat{\CalX}_Q := \bigcup_{q=1}^Q S_q$, should cover the desired domain $\mathcal{X}$, i.e., $\CalX \subseteq \hat{\CalX}_Q$. There are established strategies for selecting the finite elements $\{(\phi_q, S_q)\}_{q=1}^Q$, e.g., via simplicial methods \citep{dang2012triangulations} and B-spline approaches \citep{hollig2003finite}. Error analysis for the recommended B-spline approach is provided in the next subsection. For this subsection, we adopt the following general form for the finite elements on the Sobolev space  \citep[Chapter 2]{Brenner08}:
\begin{equation}
\ \|\phi_q\|_{\mathscr{H}_{\nu}}<\infty \;\;\; \text{for all} \;\;\; \phi_q\in V^Q, Q\in\NatInt, \quad \overline{\lim_{Q\to\infty}{\rm span}\{V^Q\}} = \mathscr{H}_{\nu}(\CalX).
\label{eq:femreg}
\end{equation}
Here, $\|\cdot\|_{\mathscr{H}_{\nu}}$ is the norm of the reproducing kernel Hilbert space $\mathscr{H}_{\nu}(\CalX)$ (see \cite{Wendland10}) induced by the Mat\'ern kernel with smoothness $\nu$, i.e., $k_{\nu}$. Such a function space can be shown to be equivalent to the $(\nu+d/2)$-order Sobolev space; see \cite[Corollary 1]{Tuo16}.

With the choice of finite elements $\{(\phi_q, S_q)\}_{q=1}^Q$ satisfying \eqref{eq:femreg} in hand, consider the following FEM-based approximation of the BdryMat\'ern kernel:
\begin{equation}
\label{eq:kernel_regression}
    \hat{k}_{m,Q}(\BFx, \BFx') = \BPhi(\BFx)^T [\BPhi(\BU)]^\dagger \hat{\BK}^{\rm C}(\BU, \BU) [\BPhi(\BU)^T]^\dagger \BPhi(\BFx'), \quad \mathbf{u}_i \overset{\text{i.i.d.}}{\sim}\text{Unif}(\mathcal{X}), \quad i = 1, \cdots, m,
\end{equation}
where $\hat{\BK}^{\rm C}(\BU, \BU)$ is the earlier coupled MC estimator: \eqref{eq:coupled_dir} for the Dirichlet setting, \eqref{eq:coupled_robin} for the Robin setting. We refer to the point set $\mathbf{U}$ as \textit{inducing points}, following \cite{snelson2005sparse,li2025prospar}. Here, $\BPhi(\BFx) = [\phi_1(\BFx), \cdots, \phi_Q(\BFx)]^T$ is the vector of basis functions at $\mathbf{x}$, $\BPhi(\BU) =[\BPhi(\mathbf{u}_1), \cdots, \BPhi(\mathbf{u}_m)]^T$ consists of basis functions for inducing points $\mathbf{U}$, and $\mathbf{M}^\dagger$ denotes the pseudo-inverse of matrix $\mathbf{M}$. This kernel approximation offers a solution to the earlier two limitations. For (i), note that $\mathbf{M} = [\BPhi(\BU)]^\dagger \hat{\BK}^{\rm C}(\BU, \BU) [\BPhi(\BU)^T]^\dagger$ requires just a one-shot evaluation, as it does not depend on kernel inputs $\mathbf{x}$ and $\mathbf{x}'$. With $\mathbf{M}$ evaluated, $\hat{k}_{m,Q}(\mathbf{x},\mathbf{x}')$ can be quickly computed for any choice of $\mathbf{x}$ and $\mathbf{x}'$ via matrix multiplication, without need for MC simulation or pseudo-inverses. We discuss later an implementation using B-splines that further reduces this computation. For (ii), the following proposition shows that $\hat{k}_{m,Q}(\mathbf{x},\mathbf{x}')$ satisfies the desired boundary conditions on the approximated domain $\hat{\CalX}_Q$:
\begin{corollary}[Boundary Conditions on FEM-Based Kernel]
\label{coro:posotive_hat_K}
    Let $\hat{\CalX}_Q = \bigcup_{q=1}^Q S_q$ be the approximated domain of \(\CalX\) via the finite elements \(\{(\phi_q, S_q)\}_{q=1}^Q\). Then, for the Dirichlet setting, the approximated kernel $\hat{k}_{m,Q}(\mathbf{x},\mathbf{x}')$ in \eqref{eq:kernel_regression} almost surely satisfies: 
    \begin{equation}
        \hat{k}_{m,Q}(\BFx, \cdot) = 0, \quad \BFx \in \partial \hat{\CalX}.
    \end{equation}
    Similarly, for the Robin setting, the approximated kernel $\hat{k}_{m,Q}(\mathbf{x},\mathbf{x}')$ almost surely satisfies:
    \begin{equation}
    \frac{\partial\hat{k}_{m,Q}(\BFx, \cdot)}{\partial\bold{n}}+c(\BFx)\hat{k}_{m,Q}(\BFx, \cdot) = 0,\quad \BFx \in \partial \hat{\CalX}.
    \end{equation}    
\end{corollary}
\noindent This corollary follows directly from Theorems \ref{thm:empirical_BdryMat_Dirichlet} and \ref{thm:empirical_BdryMat_Robin}.

With this, the posterior predictive distribution $[f(\mathbf{x}_{\rm new})|f(\mathbf{X})]$ can be approximated as follows. First, sample a large number of i.i.d. samples $\mathbf{U} = \{\mathbf{u}_1, \cdots, \mathbf{u}_m\}$ uniformly-at-random on $\mathcal{X}$. Next, given the finite elements \(\{(\phi_q, S_q)\}_{q=1}^Q\), evaluate the intermediate matrix $\mathbf{M} = [\BPhi(\BU)]^\dagger \hat{\BK}^{\rm C}(\BU, \BU) [\BPhi(\BU)^T]^\dagger$. Finally, using the form $\hat{k}_{m,Q}(\mathbf{x},\mathbf{x}') = \BPhi(\BFx)^T \mathbf{M} \BPhi(\BFx')$ in \eqref{eq:kernel_regression}, evaluate the GP predictive mean and variance in \eqref{eq:kriging} using the approximated kernel $k = \hat{k}_{m,Q}$. The following corollary guarantees that, with a sufficiently large number of finite elements $Q$, the kernel matrix $\hat{\mathbf{K}}_{m,Q}(\mathbf{X},\mathbf{X}) = [\hat{k}_{m,Q}(\mathbf{x}_i,\mathbf{x}_j)]_{i,j=1}^n$ required in the GP equations \eqref{eq:kriging} is almost surely positive definite:
\begin{corollary}[Positive Definiteness of FEM-Based Kernel Matrix]
Let $\mathbf{X}$ be a set of $n$ distinct design points on the interior of $\mathcal{X}$, i.e., $\CalX \setminus \partial \CalX$. Choose the number of finite elements $Q$ to be greater than $n$. Then, using the approximated kernel in \eqref{eq:kernel_regression}, its kernel matrix $\hat{\mathbf{K}}_{m,Q}(\mathbf{X},\mathbf{X}) = [\hat{k}_{m,Q}(\mathbf{x}_i,\mathbf{x}_j)]_{i,j=1}^n$ is positive definite almost surely. 
\label{coro:pd}
\end{corollary}
\noindent This corollary again follows from Theorems \ref{thm:empirical_BdryMat_Dirichlet} and \ref{thm:empirical_BdryMat_Robin}. Algorithm \ref{alg:fem} summarizes the above steps for this FEM-based approximation of the posterior predictive distribution.

\begin{algorithm}[!t]
    \caption{FEM-Based Approximation of the Posterior Predictive Distribution}

\justifying

        \noindent \textbf{Inputs}: Number of inducing points $m$, number of finite elements $Q$, simulated data points $\{(\mathbf{x}_i,f(\mathbf{x}_i)\}_{i=1}^n$, prediction point $\mathbf{x}_{\rm new}$. Ensure $Q > n$.\\
        \vspace{-0.35cm}
        
        \noindent $\bullet$ Generate the finite elements $\{(\phi_q, S_q)\}_{q=1}^Q$ satisfying regularity conditions \eqref{eq:femreg}.\\
        \vspace{-0.2cm}

        \noindent $\bullet$ Sample inducing points $\mathbf{U} = \{\mathbf{u}_1, \cdots, \mathbf{u}_m\} \overset{\text{i.i.d.}}{\sim}\text{Unif}(\mathcal{X})$.\\
        \vspace{-0.35cm}

        \noindent $\bullet$ Compute $\mathbf{M} = [\BPhi(\BU)]^\dagger \hat{\BK}^{\rm C}(\BU, \BU) [\BPhi(\BU)^T]^\dagger$, where $ \hat{\BK}^{\rm C}(\BU, \BU)$ is the coupled MC kernel matrix estimator (Equation \eqref{eq:coupled_dir} for Dirichlet boundary; Equation \eqref{eq:coupled_robin} for Robin boundary).\\
        \vspace{-0.35cm}

        \noindent $\bullet$ Compute the FEM-based approximated kernel $\hat{k}_{m,Q}(\mathbf{x},\mathbf{x}') = \BPhi(\BFx)^T \mathbf{M} \BPhi(\BFx')$, and its approximated kernel matrix $\hat{\mathbf{K}}_{m,Q}(\mathbf{X},\mathbf{X}) = [\hat{{k}}_{m,Q}(\mathbf{x}_i,\mathbf{x}_j)]_{i,j=1}^n$.\\
        \vspace{-0.35cm}

        \noindent $\bullet$ Evaluate the approximated posterior predictive distribution ${\Pi}(\mathbf{x}_{\rm new})$ via Equation \eqref{eq:kriging} with $k = \hat{k}_{m,Q}$.\\
        \vspace{-0.35cm}
        
   \noindent \textbf{Output}: Return ${\Pi}(\mathbf{x}_{\rm new})$.
    \label{alg:fem}
\end{algorithm}


Another appeal of the FEM-based kernel $\hat{k}_{m,Q}$ in \eqref{eq:kernel_regression} is that it facilitates approximation error analysis. The following theorem provides an error bound on how well such a kernel approximates the desired BdryMat\'ern kernel $k_{\nu,\mathcal{B}}$:
\begin{theorem}[FEM-Based Kernel Approximation Error]
    Consider the FEM-based approximated kernel in \eqref{eq:kernel_regression}. Define the \(L^2\)-error between two kernels \(k_1\) and \(k_2\) as:
\begin{equation}
     \|k_1 - k_2\|^2_{L^2(\CalX \times \CalX)} = {\int_{\CalX \times \CalX} \left|k_1(\BFx, \BFx') - k_2(\BFx, \BFx')\right|^2 d\BFx d\BFx'}. \label{eq:HS_norm}
\end{equation}
Then, with probability at least \(1 - v_* - Q(\sqrt{2}e^{-1/2})^{m\underline{\lambda}/S_\phi} - Q\left(\frac{e^{1/2}}{(3/2)^{3/2}}\right)^{m\overline{\lambda}/S_\phi}\), we have:
\begin{equation}
    \label{eq:convergence_general}
    \|\hat{k}_{m,Q} - k_{\nu,\mathcal{B}}\|_{L^2(\CalX \times \CalX)}^2 \leq C \left(\frac{\overline{\lambda}}{\underline{\lambda}}\right)^2 \left[\int_\CalX (\bold{I} - \CalP)[k_\nu + \bar{k}](\bold{I} - \CalP)^T(\Bx, \Bx) d\Bx + \frac{\overline{\lambda}S_\phi}{\underline{\lambda}^2m}\right],
\end{equation}
where $C>0$ is a constant independent of basis choice. Here, $\overline{\lambda}$ and $\underline{\lambda}$ are the maximum and minimum eigenvalues of the Gram matrix $\boldsymbol{\Lambda} = [\lambda_{ij}]_{i,j=1}^Q$, $\lambda_{ij} = \int_\CalX \phi_i(\BFx) \phi_j(\BFx) d\BFx$. Further, \(S_\phi = \max_{\Bx \in \CalX} \sum_{j=1}^p \phi_j^2(\Bx)\), $\bar{k}$ is defined as the kernel:
\small
\begin{equation}
    \bar{k}(\mathbf{x},\mathbf{x}')=
    \begin{cases}
        &\E\left[k_\nu(\BB_\tau,\BB'_\gamma)\bigg|\BB_0=\mathbf{x},\BB'_0=\mathbf{x}'\right]\quad\textup{(Dirichlet)},\\
        &\E\left[\int_{\Real_+^2}e^{-\kappa^2(t+\tau)-\int_0^tc(\BB_s)dL_s-\int_0^\tau c(\BB'_s)dL'_s}k_\nu(\BB_t,\BB'_\tau)dL_tdL'_\tau\bigg|\BB_0=\mathbf{x},\BB'_0=\mathbf{x}'\right]\  \textup{(Robin)},
    \end{cases}
    \label{eq:kbar}
\end{equation}
\normalsize
$\mathcal{P}$ is the $L^2$-projection operator $\CalP f = \arg\min_{h \in \textup{span}\{V^Q\}} \|h - f\|_{L^2(\CalX)}^2$, $\bold{I}$ is the identity operator, and $(\bold{I} - \CalP)k(\cdot, \cdot)(\bold{I} - \CalP)^T(\BFx, \BFx')$ is an operator that applies $\bold{I} - \CalP$ on the first variable of $k(\mathbf{x},\mathbf{x}')$, then $\bold{I} - \CalP$ on its second variable. Finally, $v_* = \int_\CalX \left|(\bold{I} - \CalP)[k_\nu + \bar{k}](\bold{I} - \CalP)^T(\Bx, \Bx)\right|^2 d\Bx$.
    \label{thm:convergence_general}

\end{theorem}
\noindent The proof of this theorem is provided in Appendix \ref{pf:convergence_general}. Note that the norm defined in \eqref{eq:HS_norm} is equivalent to the Hilbert–Schmidt norm if we treat $k_1-k_2$ as an integral operator. Equation \eqref{eq:convergence_general} shows that the approximation error bound depends on the base Mat\'ern kernel $k_{\nu}$ and its projection onto the employed finite element basis. This error bound is derived from the decomposition in Theorem \ref{thm:empirical_BdryMat_Dirichlet}, and we make use of such a bound next to derive a kernel approximation rate using the specific choice of B-spline bases.

\subsection{B-Spline Implementation}
\label{sec:Bspline}
Next, we consider the choice of B-splines \cite{hollig2003finite} for the finite elements $\{(\phi_q,S_q)\}_{q=1}^Q$, which we recommend in implementation. Without loss of generality, we presume that in the following that $\mathcal{X} \subseteq [0,1]^d$; this can easily be done by rescaling the domain. The use of B-splines here has three key advantages. First (i), such bases have been shown to provide effective approximation of both Dirichlet and Robin boundaries in finite-element modeling (see \cite{shu2012differential}). Second (ii), by construction, the support sets of most B-spline basis functions are disjoint, which facilitates an efficient evaluation of the approximated kernel $\hat{k}_{m,Q}$ \eqref{eq:kernel_regression} via sparse matrix operations. Finally (iii), using the famous de Boor's conjecture proven in \cite{shadrin2001norm}, we can nicely characterize the kernel approximation quality of $\hat{k}_{m,Q}$ to the BdryMat\'ern kernel ${k}_{\nu,\mathcal{B}}$ using B-splines.

We first briefly review B-splines, following \cite{hollig2003finite}. Let $ N(x) = \boldsymbol{1}_{x \in [0,1]}$, and define:
\begin{equation}
    \label{eq:B_spline}
    N_s(x) = \underbrace{N \ast N \ast \cdots \ast N}_{ \text{$s+1$ times}}(x),
\end{equation}
where \( f \ast g(x) := \int f(x - t)g(t) \, dt \) is the convolution operator. One can show that the support of $N_s$ is $[0,s+1]$. For given multi-indices \( \boldsymbol{\zeta} = (\zeta_1, \ldots, \zeta_d) \in \mathbb{N}^d \) and \( \mathbf{j} = (j_1, \ldots, j_d) \in \mathbb{N}^d \), the multi-dimensional tensor B-spline function indexed by $(\boldsymbol{\zeta}, \mathbf{j})$ is defined as 
$M_{\boldsymbol{\zeta}, \mathbf{j}}^d(x) = \prod_{l=1}^d N_s(2^{\zeta_l}x_l - j_l)$.
Here, \( \boldsymbol{\zeta} \) is a degree vector that controls the resolution of the basis function in each dimension, and \( \mathbf{j} \) specifies the location on which the basis is put. We can similarly define this spline function for a given scalar degree $\zeta \in \mathbb{N}$ as $M_{\zeta, \mathbf{j}}^d(x) = \prod_{l=1}^d N_s(2^\zeta x_l - j_l)$. As the support of $N_s$ is $[0,s+1]$, it is straightforward to show that the support of $M_{\zeta, \mathbf{j}}^d$ is:
\begin{equation}
    \label{eq:B_spline_support}
    {\rm supp}\{M_{\zeta, \mathbf{j}}^d\}=[0,2^{-\zeta}(s+1)]^d+\bold{j}2^{-\zeta}.
\end{equation}
 For a given order $s$, the set of cardinal B-spline basis functions is given by $F_\zeta=\{M_{\zeta, \mathbf{j}}^d, \mathbf{j}\in \mathcal{J}(\zeta)\}$, where $\mathcal{J}(\zeta)=\{-s,-s+1,...,2^\zeta-1,2^\zeta\}^d$. This serves as our set of basis functions $\{\phi_q\}_{q=1}^Q$, with the number of basis functions given as $Q= | F_{\zeta}|$.

Consider now point (ii) above. Recall that our kernel approximator takes the form:
\begin{equation}
\label{eq:kernel_reg2}
 \hat{k}_{m,Q}(\BFx, \BFx') = \sum_{q,q'=1}^Q m_{q,q'}{\phi}_q(\BFx) {\phi}_{q'}(\BFx'), \; \mathbf{M} = [m_{q,q'}]_{q,q'=1}^Q = [\BPhi(\BU)]^\dagger \hat{\BK}^{\rm C}(\BU, \BU) [\BPhi(\BU)^T]^\dagger.
\end{equation}
The minimal support overlap between B-spline basis functions (see \cite{hollig2003finite}) can greatly speed up this computation. We see from \eqref{eq:B_spline_support} that (a) many basis functions in $\{\phi_q(\mathbf{x})\}_{q=1}^Q$ are zero for a given $\mathbf{x}$, and (b) the product terms $\{\phi_q(\mathbf{x})\phi_{q'}(\mathbf{x})\}_{q,q'=1}^Q$ are also largely zero for any pair of distinct points $\mathbf{x}$ and $\mathbf{x}'$. Property (a) permits the use of sparse algorithms \citep{ding2020generalization,ding2024sample} for efficient evaluation of the pseudoinverses and matrix products in $\mathbf{M}$ within \eqref{eq:kernel_reg2}. Property (b) permits the efficient evaluation of $\hat{k}_{m,Q}$, as most terms in the sum within \eqref{eq:kernel_reg2} equal zero. With this setup (see Appendix \ref{sec:B_spline_time_complexity} for details), after a one-time pre-processing step with computational cost $\mathcal{O}(Q^3)$, the approximated kernel $\hat{k}_{m,Q}(\BFx,\BFx')$ can be evaluated in $\mathcal{O}\{(2s)^{2d}\}$ work for any pair $(\BFx, \BFx')$. The latter thus provides a constant-time complexity independent of both $Q$ and $n$.


Consider next point (iii) above. Using B-spline bases, one can use Theorem \ref{thm:convergence_general} to characterize the approximation quality of the kernel $\hat{k}_{m,Q}$ to the desired BdryMat\'ern kernel $k_{\nu,\mathcal{B}}$. This relies on showing that the ratio $\overline{\lambda}/\underline{\lambda} = \mathcal{O}(1)$, which follows from the famous de Boor's conjecture proven in \cite{shadrin2001norm}. This approximation rate using B-splines is presented below:

\begin{theorem}[B-Spline Kernel Approximation Error Rate]
\label{thm:b_spline_regression}
Suppose $\CalX\subseteq[0,1]^d$ and $\nu > d/2$. Let $\{\phi_q\}_{q=1}^Q$ be the set of B-spline basis functions with order $s>\nu+d/2+1/2$ and degree $\zeta = \lfloor (\log_2m)/(2\nu+2d) \rfloor$. Using this basis, the approximated kernel $\hat{k}_{m,Q}$ in \eqref{eq:kernel_regression} satisfies:
\begin{equation}
    \|\hat{k}_{m,Q}-k_{\nu,\mathcal{B}}\|_{L^2(\CalX\times\CalX)}= \CalO(m^{-\frac{\nu}{2\nu+d}}),
    \label{eq:approx_rate}
\end{equation}
     with probability at least $1-C_1m^{-\frac{2\nu-d}{2\nu+d}}-e^{-C_2(m^{\frac{2\nu}{2\nu+d}}-\log m)}$, where $C_1$ and $C_2$ are constants in $m$.
 \end{theorem}
\noindent The proof of this theorem is provided in Appendix \ref{pf:b_spline_regression}. This approximation rate holds for both the Dirichlet and Robin boundary settings. Note that the rate in \eqref{eq:approx_rate} coincides with the $L^2$-minimax lower bound for GP regression \citep[Theorem 10]{wang2022gaussian}. This is not too surprising, as the task of kernel approximation is related to GP regression in the sense that a (zero-mean) GP is fully characterized by its covariance kernel specification.

An important practical use of such error analysis is that it can guide the specification of $m$, i.e., the number of inducing points, for reliable kernel approximation. Note that, given a desired smoothness parameter $\nu$, only $m$ needs to be specified using the approximation set-up in Theorem \ref{thm:b_spline_regression}; the order and degree of the B-spline bases are provided given this choice of $\nu$ and $m$. One way to specify $m$ is to match the approximation rate in \eqref{eq:approx_rate} with known minimax contraction rates (similarly in $L^2$) for GPs. Using the Mat\'ern-$\nu$ kernel, the latter is known to be of order $\mathcal{O}(n^{-\nu/d})$, where $n$ is the number of evaluation runs on $f$. Setting both rates together and solving for $m$, we get (ignoring constants) that $m = n^{(2\nu+d)/d}$. This specification of $m$ can be used as a rule-of-thumb to ensure the kernel approximation error does not affect the overall GP contraction rate in an asymptotic sense.

\section{The Tensor BdryMat\'ern GP Model}
\label{sec:tensor_mat}



Finally, we present a tensor form of the BdryMat\;ern GP. Such a form may be useful when the domain is a high-dimensional unit hypercube, i.e., $\mathcal{X} = [0,1]^d$ for large $d$, and boundary information is known on $\mathcal{X}$. For this high-dimensional setting, the earlier approximation approach may suffer from a ``curse-of-dimensionality'', in that approximation becomes more difficult as $d$ increases. Here, the tensor BdryMat\'ern GP offers a key advantage of a closed-form kernel, which can directly integrate within the GP equations \eqref{eq:kriging} for efficient posterior predictions.


As before, without loss of generality, we consider the homogeneous boundary setting below; the inhomogeneous setting can be accounted for via a careful specification of its mean function $\mu_{\mathcal{B}}$ (see Section \ref{sec:Bdrymat_SPDE}). Suppose we know the following Robin boundary conditions on $f$:
\begin{align}
\label{eq:boundar_mixed_full}
f(\BFx)+c_l\partial_{x_l}f(\BFx)=0,\quad c_l\geq 0, \quad x_l=0,1, \quad \text{for all } l=1,\cdots,d.
\end{align}
To incorporate boundary information of this form, we employ the so-called tensor BdryMat\'ern kernel:
\begin{equation}
k_{\nu,\mathcal{B}}^{\rm T}(\mathbf{x},\mathbf{x}') = \prod_{l=1}^d k_{\nu,\mathcal{B},c_l}({x_l},{x_l}').
\label{eq:tensorbdry}
\end{equation}
Here, $k_{\nu,\mathcal{B},c}$ is the one-dimensional BdryMat\'ern kernel from \eqref{eq:bdryMatern_path_integral_robin} with smoothness $\nu$, incorporating homogeneous Robin boundaries of the form:
\begin{align}
\label{eq:boundar_mixed_1D}
   f(x)+cf'(x)=0, \quad c\geq 0, \quad x=0,1,
\end{align}
on the one-dimensional domain $[0,1]$.




With this, we can derive a closed-form expression for the one-dimensional BdryMat\'ern kernel $k_{\nu,\mathcal{B},c}$ in \eqref{eq:tensorbdry}. This is given in the following theorem:
\begin{theorem}
\label{thm:tensor}
Consider the one-dimensional BdryMat\'ern kernel $k_{\nu,\mathcal{B},c}$ with smoothness $\nu$ for the homogeneous Robin boundaries \eqref{eq:boundar_mixed_1D} on domain $[0,1]$. This kernel has the closed-form representation:
\begin{align}
\begin{split}
k_{\nu,\mathcal{B},c}(x,x') = k_{\nu}(x,x') - k_{hp}(x,x') - k_{hp}(x',x) + k_{hh}(x,x'),
\end{split}
\label{eq:closedform}
\end{align}
\label{thm:closedform}
where $k_{\nu}(\cdot,\cdot)$ is the base Mat\'ern kernel with smoothness $\nu > 2$ and inverse length-scale $\kappa > 0$, and:
\small
\begin{align}
\begin{split}
k_{hp}(x,x') &= \frac{1/2}{\sinh(\kappa)}\left[K_1(x-1)\left(\frac{e^{\kappa x'}}{1+c\kappa}-\frac{e^{-\kappa x'}}{1-c\kappa}\right)+K_1(x)\left(\frac{e^{\kappa(1- x')}}{1-c\kappa}-\frac{e^{-\kappa(1- x')}}{1+c\kappa}\right)\right],\\
k_{hh}(x,x')&=\frac{1}{\sinh({\kappa})^{2}} \times\\
& \quad \left(\left[\cosh(\kappa)K_2(0)-K_2(1)\right]\left(\frac{e^{\kappa(1-x-x')}}{2(1-c\kappa)^2}+\frac{e^{-\kappa(1-x-x')}}{2(1+c\kappa)^2}\right) + \left[\cosh( {\kappa})K_2(1)-K_2(0)\right]\frac{\cosh( {\kappa}(x-x'))}{1-c^2\kappa^2}\right).
\end{split}
\label{eq:closedform2}
\end{align}
\normalsize
Here, $K_1(x)= k_{\nu}(x,x)-c\partial_x k_{\nu}(x,x)$ and $K_2(x)=k_{\nu}(x,x)-c^2\partial_{x x}k_{\nu}(x,x)$.
\end{theorem}

\begin{proof}
Since the BdryMat\'ern kernel (with $\nu > 2$) can be represented as the convolution between Green's functions and the Mat\'ern kernel, we have:
\begin{align}
    k_{\nu,\mathcal{B}}(x,x')=&\int_{\CalX\times\CalX}G(x,s)k_{\nu-2}(s,u)G(u,x')dsdu
    =C_{\CalF^{-1}}\int_\Real\frac{f(x,\omega)\overline{f(x',\omega)}}{(\kappa^2+\omega^2)^{\nu-3/2}}d\omega.
     \label{eq:kernel_green_1d}
\end{align}
Here, the second equality follows from the spectral density representation of $k_\nu$ and switch of integral, and $f(x,\omega)=\int_{\CalX}G(x,s)e^{i\omega s}ds$ is the solution to the one-dimensional PDE:
\begin{align}
\begin{split}
   & \left(\kappa^2-\frac{\partial^2}{\partial x^2}\right)f(x)=e^{i\omega x},\\
   &f(x)+cf'(x)=0,\quad x=0,1,\quad c\geq 0.
   \end{split}
\label{eq:green_1_character}
\end{align}
The solution to \eqref{eq:green_1_character} can be represented as a particular solution $f_p$ minus a homogeneous $f_h$, i.e., $f(x,\omega)=f_p(x,\omega)-f_h(x,\omega)$, with each part satisfying:
\begin{align*}
    &\left(\kappa^2-\frac{\partial^2}{\partial x^2}
    \right)u_p(x)=e^{i\omega x},\quad \left(\kappa^2-\frac{\partial^2}{\partial x^2}\right)u_h(x)=0,\\
   & u_h(x)+cu_h'(x)=u_p(x),\quad x=0,1.
\end{align*}
Here, both $f_p$ and $f_h$ can be solved analytically as:
\begin{align}
\begin{split}
    & f_p(x)=\frac{e^{i\omega x}}{\kappa^2+\omega^2},\\
    & f_h(x)=\frac{(e^\kappa-e^{i\omega})(1+ic\omega)e^{-\kappa x}}{2(1-c\kappa)\sinh(\kappa)(\kappa^2+\omega^2)}+\frac{(e^{i\omega}-e^{-\kappa})(1+ic\omega)e^{\kappa x}}{2(1+c\kappa)\sinh(\kappa)(\kappa^2+\omega^2)}.
\end{split}
\label{eq:u_p_u_h_solution}
\end{align}
Substitute \eqref{eq:u_p_u_h_solution} into \eqref{eq:kernel_green_1d}, we have the closed form of $k_{\nu,\mathcal{B},c}$ in Equations \eqref{eq:closedform} and \eqref{eq:closedform2}.
\end{proof}


\noindent An analogous closed-form kernel for the Dirichlet boundary setting can be obtained by setting $c=0$. Note that the closed-form kernel from \eqref{eq:closedform} and \eqref{eq:closedform2} aligns with the BdryMat\'ern kernel form in \eqref{eq:bdryMatern_path_integral_dirichlet} and \eqref{eq:bdryMatern_path_integral_robin}; $k_{hp}$ and $k_{hh}$ provide the exact Brownian motion expectations in \eqref{eq:bdryMatern_path_integral_dirichlet} for Dirichlet boundaries and \eqref{eq:bdryMatern_path_integral_robin} for Robin boundaries.

With this, the kernel from \eqref{eq:closedform} and \eqref{eq:closedform2} can then be integrated within \eqref{eq:tensorbdry} to obtain a closed-form form tensor BdryMat\'ern kernel $k_{\nu,\mathcal{B}}^{\rm T}$. One can show that the GP with such a tensor kernel has corresponding sample paths that satisfy the desired boundary information \eqref{eq:boundar_mixed_full}. This closed-form tensor kernel further permits closed-form GP predictive equations from \eqref{eq:kriging} with $k = k_{\nu,\mathcal{B}}^{\rm T}$. Thus, in this simpler setting of a unit hypercube domain $\mathcal{X} = [0,1]^d$ with boundary information of the form \eqref{eq:boundar_mixed_full}, one can bypass the need for kernel approximation via a carefully-specified tensor kernel.



Figure \ref{fig:bdrymat_plot} visualizes different sample paths for the BdryMat\'ern GP using the closed-form kernel from Theorem \ref{thm:tensor}. Here, two choices of smoothness parameters ($\nu = 5/2$ and $7/2$) are tested with inverse length-scale parameter fixed at $\kappa = 5$, under homogeneous Dirichlet, Robin ($c=10$) and Neumann boundary conditions. There are two observations of interest. First, for the Dirichlet (or Neumann) settings, all sample paths appear to go through zero (or have a derivative of zero) at left and right end-points, which satisfies provided boundaries; Robin boundaries are more difficult to see visually. Second, for larger $\nu$, the sample paths appear to have a higher degree of smoothness (in terms of differentiability), which is as desired.




\begin{figure}
\begin{center}
\includegraphics[width=0.3\textwidth]{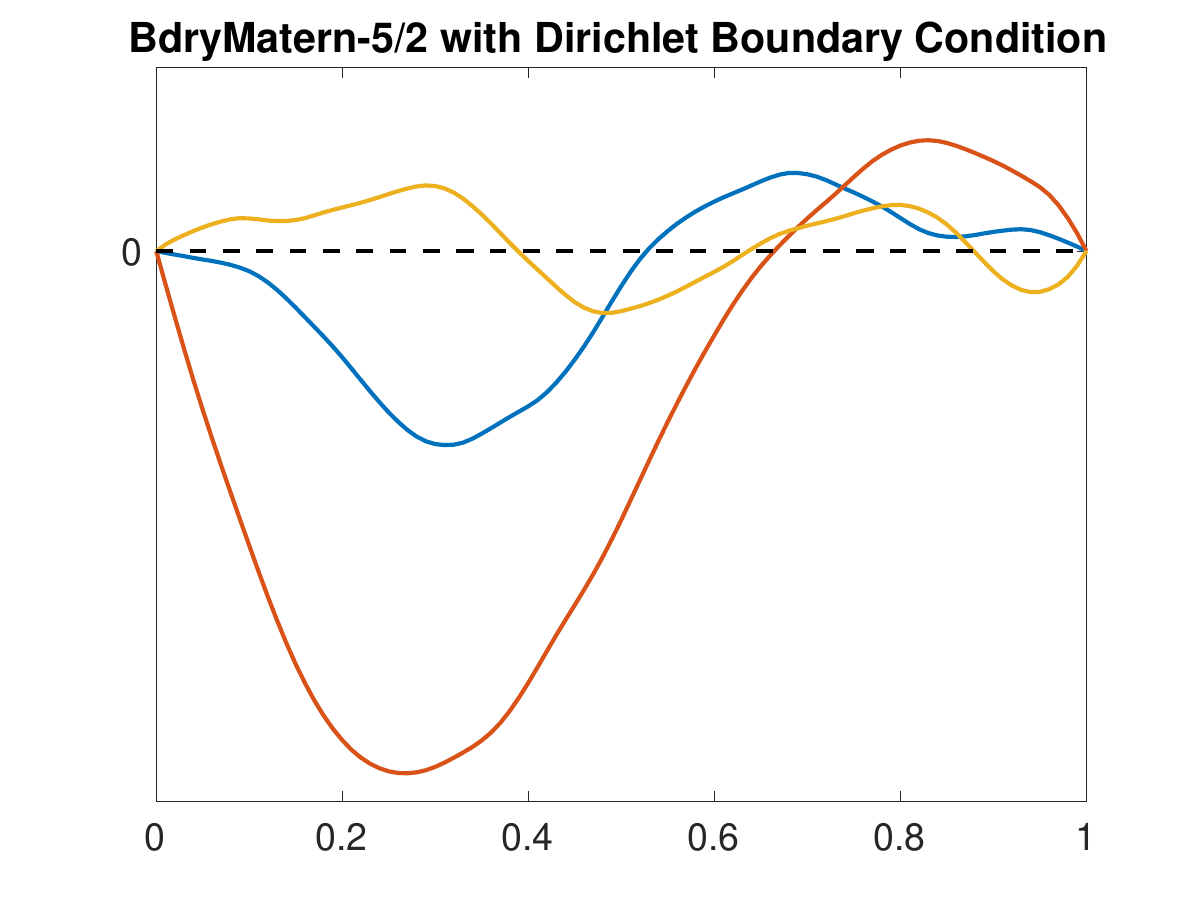}
\includegraphics[width=0.3\textwidth]{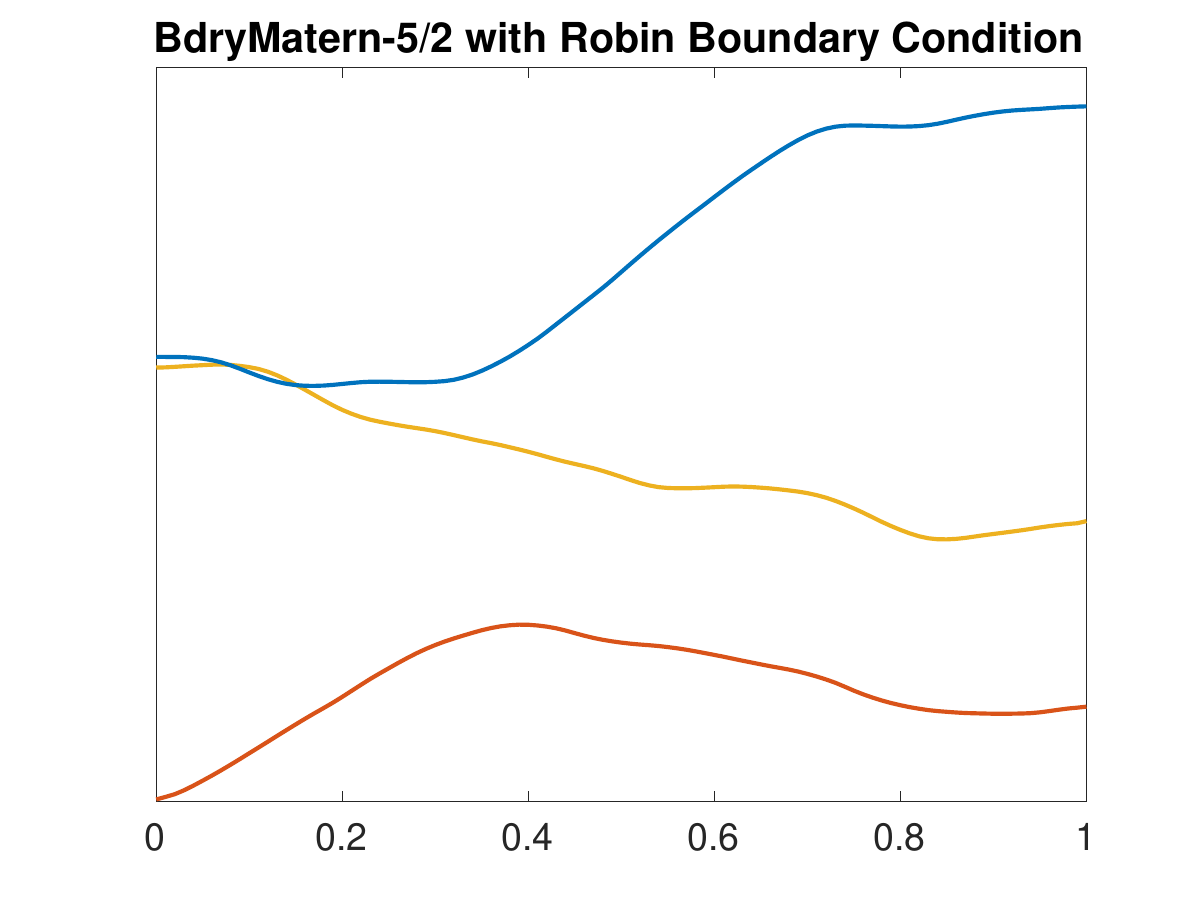}
\includegraphics[width=0.3\textwidth]
{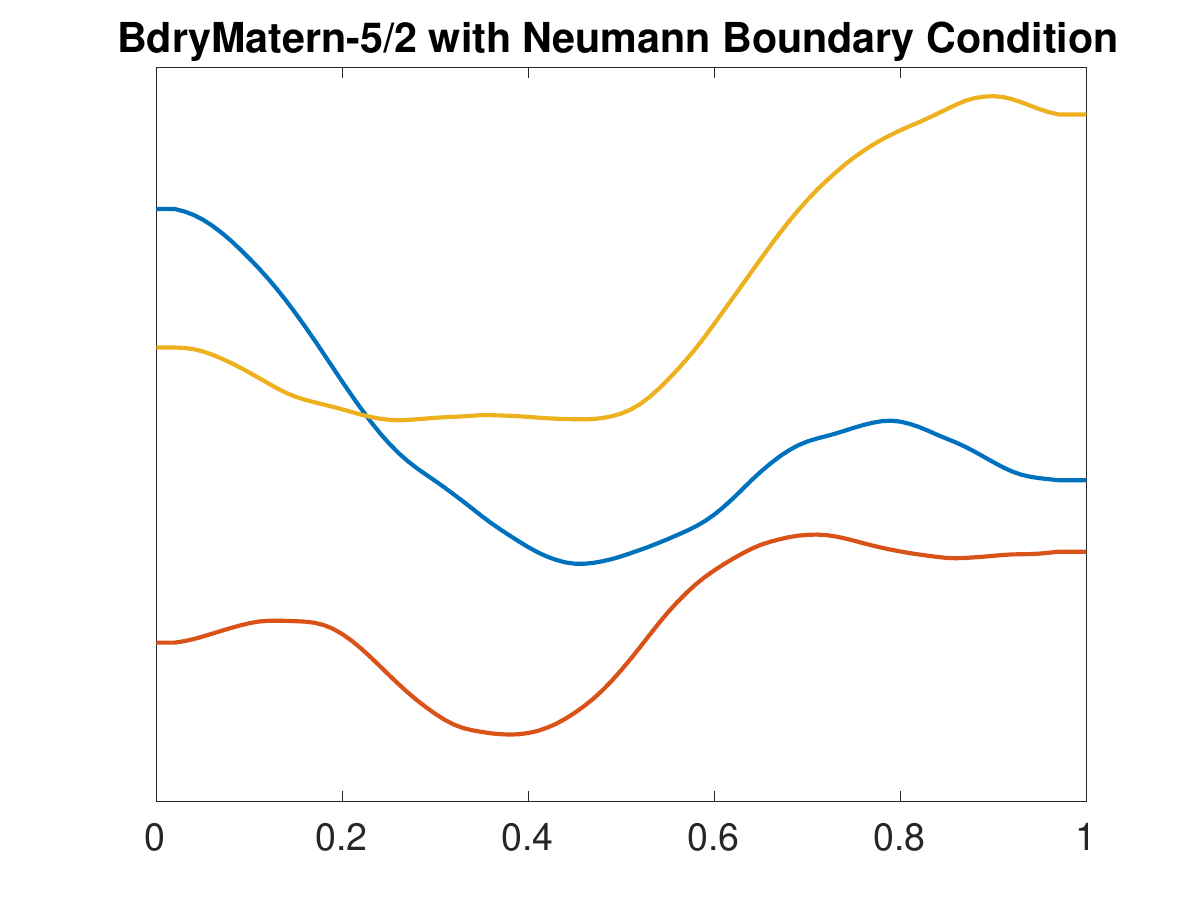}
\includegraphics[width=0.3\textwidth]
{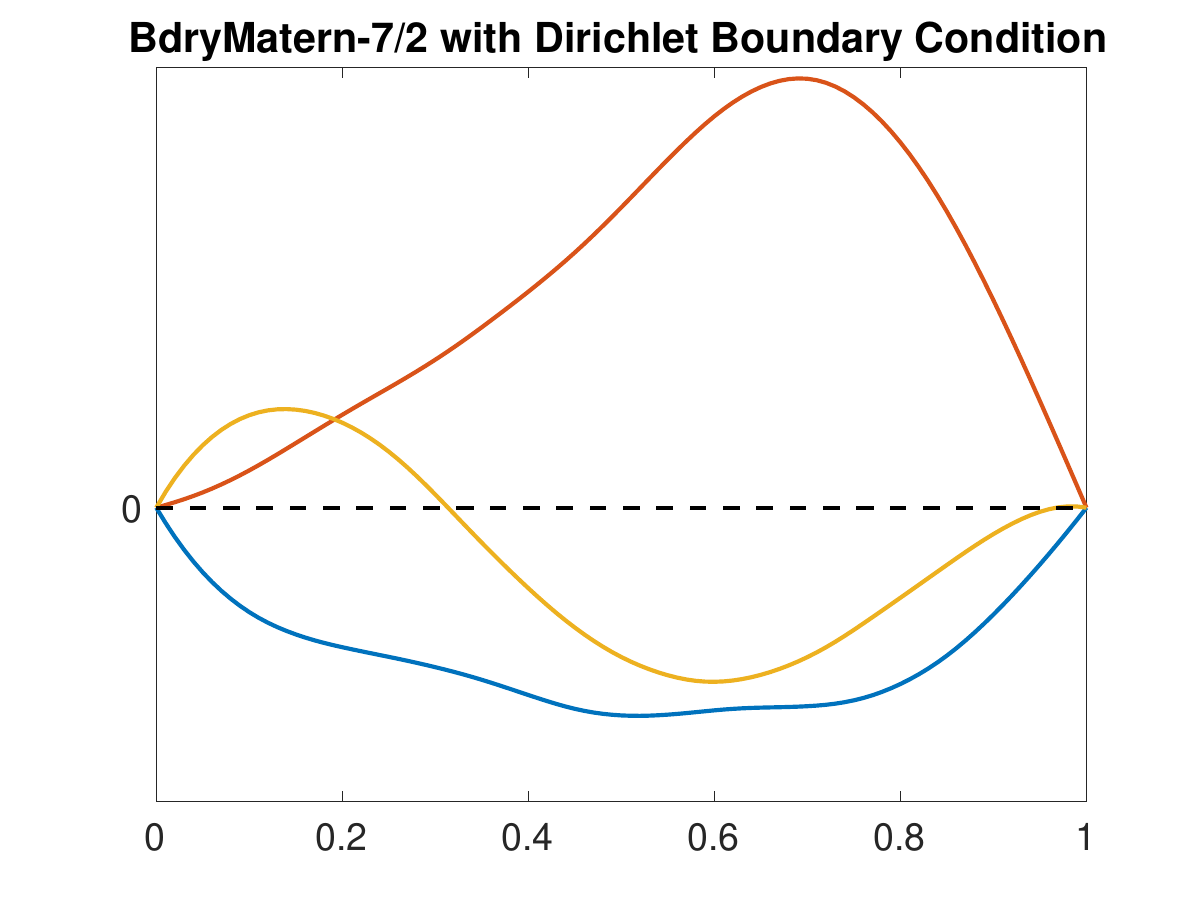}
\includegraphics[width=0.3\textwidth]{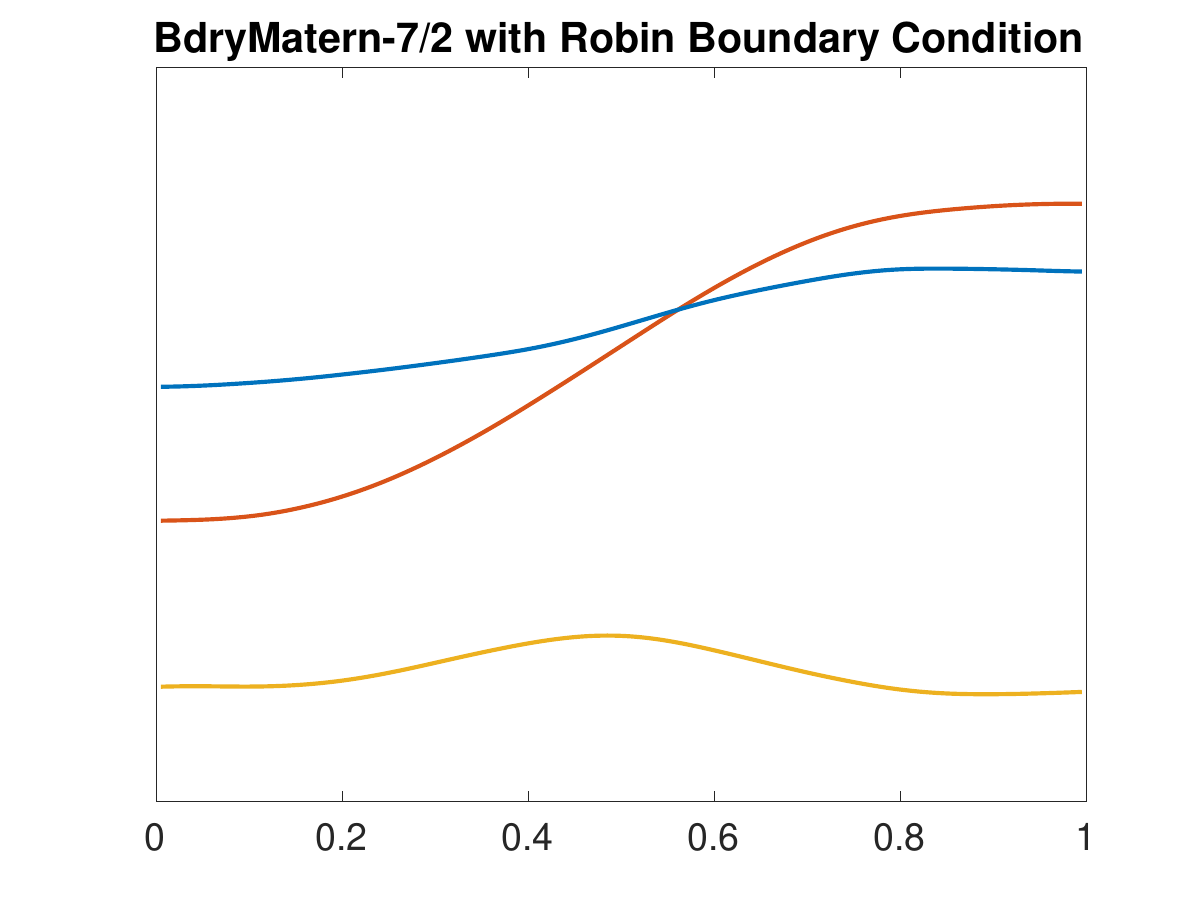}
\includegraphics[width=0.3\textwidth]{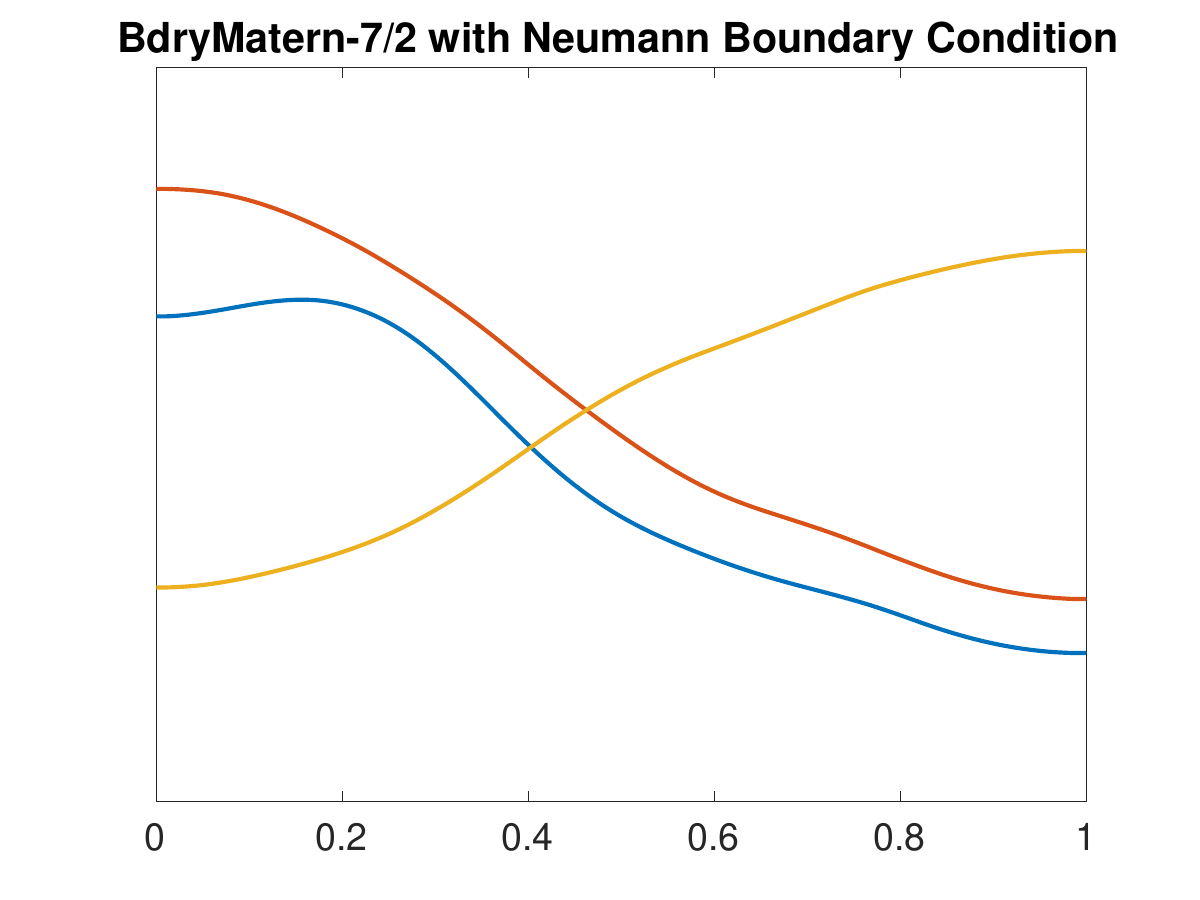}
\end{center}
\caption{[Top] Sample paths from the BdryMat\'ern GP with smoothness parameter $\nu = 5/2$ given Dirichlet, Robin and Neumann boundary conditions on $[0,1]$. Different colors correspond to different sample paths. [Bottom] Sample paths from the BdryMat\'ern GP with smoothness parameter $\nu = 7/2$ given Dirichlet, Robin and Neumann boundary conditions on $[0,1]$.}
\label{fig:bdrymat_plot}
\end{figure}



\section{Numerical Experiments}
\label{sec:numeric}

Finally, we investigate the performance of the proposed boundary-integrated GP model in a suite of numerical experiments. We first explore the BdryMat\'ern GP (Sections \ref{sec:Bdrymat_SPDE} and \ref{sec:kernel_approx}) for a variety of irregular domains in two dimensions, then examine the tensor BdryMat\'ern GP (Section \ref{sec:tensor_mat}) on the higher-dimensional unit hypercube domain $[0,1]^{30}$. 

The general experimental set-up is as follows. In both cases, we compare the proposed model with its corresponding Mat\'ern GP counterpart without boundary conditions (details provided later), to gauge the effect of incorporating boundary information. All models share the same training design points, and model parameters are estimated via maximum likelihood \citep{casella2002}. Using the posterior mean $\hat{f}_n(\cdot)$ as our predictor, predictive performance is evaluated via its log-mean-squared-error (log-MSE) on an independent test set of $10^4$ randomly-selected points on $\mathcal{X}$. All experiments are performed in Matlab R2023a on a MacBook with Apple M2 Max Chip and 32GB of RAM.



\subsection{BdryMat\'ern GP}
We first explore the BdryMat\'ern GP on a variety of irregular (i.e., non-hypercube) domains. The following four domains are considered:
\begin{itemize}
\item T-shaped: $\mathcal{X}_{\rm T} = ([-1, 1] \times [0, 1]) \cup ([-0.5, 0.5] \times [0, 1])$,
\item Ring-shaped: $\mathcal{X}_{\rm ring} = \left\{ (x_1, x_2) : 1 \leq \sqrt{x_1^2+x_2^2} \leq 2 \right\}$,
\item Disk-shaped: $\mathcal{X}_{\rm disk} = \left\{ (x_1, x_2): \sqrt{x_1^2+x_2^2} \leq 1 \right\}$,
\item Holed-rectangle: $\mathcal{X}_{\rm hrec} = \left\{ (x_1, x_2): |x_1|\leq 1,-1\leq x_2\leq \frac{1}{2}, \sqrt{(x_1-0.5)^2+(x_2+0.5)^2}\geq 0.25\right\}$. 
\end{itemize}
For each domain choice, we generate the true function $f$ from the BdryMat\'ern GP with kernel $k_{\nu,\mathcal{B}}$ with smoothness $\nu = 5/2$, inverse length-scale $\kappa = 5 $ and normalized variance parameter $\sigma^2 = \max_{\BFx\in\CalX} \hat{k}_{m,Q}(\BFx,\BFx)$.  For the T-shaped and ring-shaped domains, we adopt the homogeneous (i.e., zero) Dirichlet boundary conditions $f(\partial \mathcal{X}) = 0$, which we presume to be known for surrogate modeling. For the disk and holed-rectangle domains, we adopt Robin boundary conditions of the form \eqref{eq:robinbdry} with $c(\mathbf{x})=20$; these are again presumed to be known for surrogate modeling. 



Next, we sampled $n=50$ training design points $\mathbf{x}_1, \cdots, \mathbf{x}_n$ uniformly-at-random on $\mathcal{X}$, and collected training data $f(\mathbf{x}_1), \cdots, f(\mathbf{x}_n)$. Using such data, we fit the proposed BdryMat\'ern GP (from Section \ref{sec:Bdrymat_SPDE}) with smoothness parameter $\nu = 5/2$, using the FEM-based approximation procedure from Algorithm \ref{alg:fem} for posterior predictions. Here, B-splines are used, with $m = 153, 107, 97, 127$ inducing points for the T-shaped, ring-shaped, disk and holed-rectangle domains, following the nodes of their respective triangulations. We then compare with an (isotropic) Mat\'ern GP with smoothness $\nu = 5/2$, which is the counterpart for the BdryMat\'ern GP from Section \ref{sec:Bdrymat_SPDE} without boundary information. This simulation is replicated 100 times to ensure reproducibility.


Figure \ref{fig:2d_log_err} shows the corresponding mean log-MSEs of the compared models for each domain choice. For all domains, we see that the integration of known boundary information using the BdryMat\'ern GP (with the kernel approximation approach in Section \ref{sec:kernel_approx}) indeed improves predictive performance. Furthermore, as sample size $n$ increases, the BdryMat\'ern GP appears to offer improved prediction error rates in $n$ over its Mat\'ern GP counterpart. This suggests that the integration of boundary information can perhaps accelerate GP prediction rates in the irregular domain setting, much like the improved rates for boundary-integrated surrogates shown in \cite{ding2019bdrygp} for unit hypercube domains. We aim to investigate this rate improvement theoretically as future work. Figure \ref{fig:2d_log_err} visualizes the true response surfaces $f$ for a single experiment replication, and the posterior mean predictor $\hat{f}_n$ from the two compared models. We see that, without boundary information, the Mat\'ern GP may yield erratic predictions particularly when the domain is highly irregular. By incorporating known boundary information on $\mathcal{X}$, the BdryMat\'ern GP provides a considerably improved visual fit of $f$ with limited samples, which corroborates the results in Figure \ref{fig:2d_log_err}.

\begin{figure}[!t]
\hspace{2cm}
\begin{center}
\includegraphics[width=0.23\textwidth]
{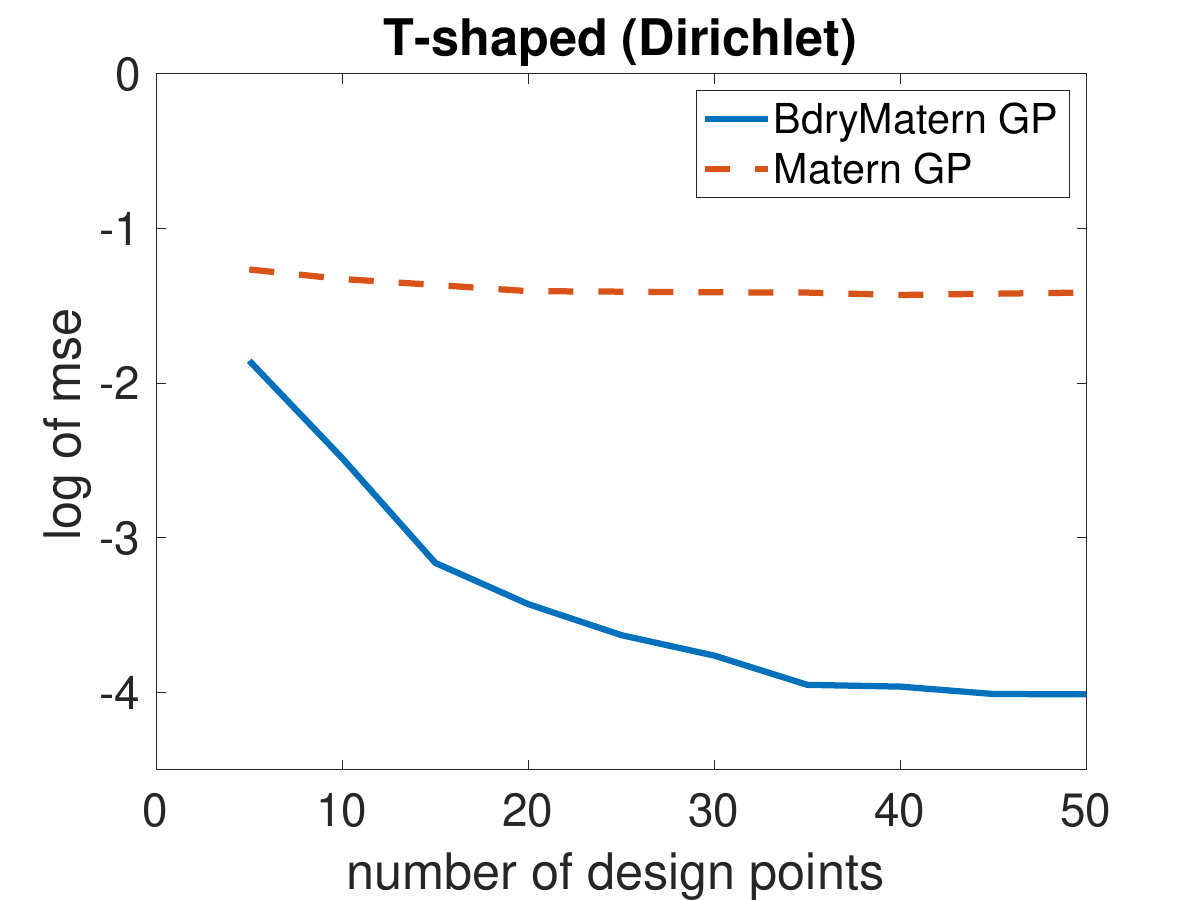}
\includegraphics[width=0.23\textwidth]{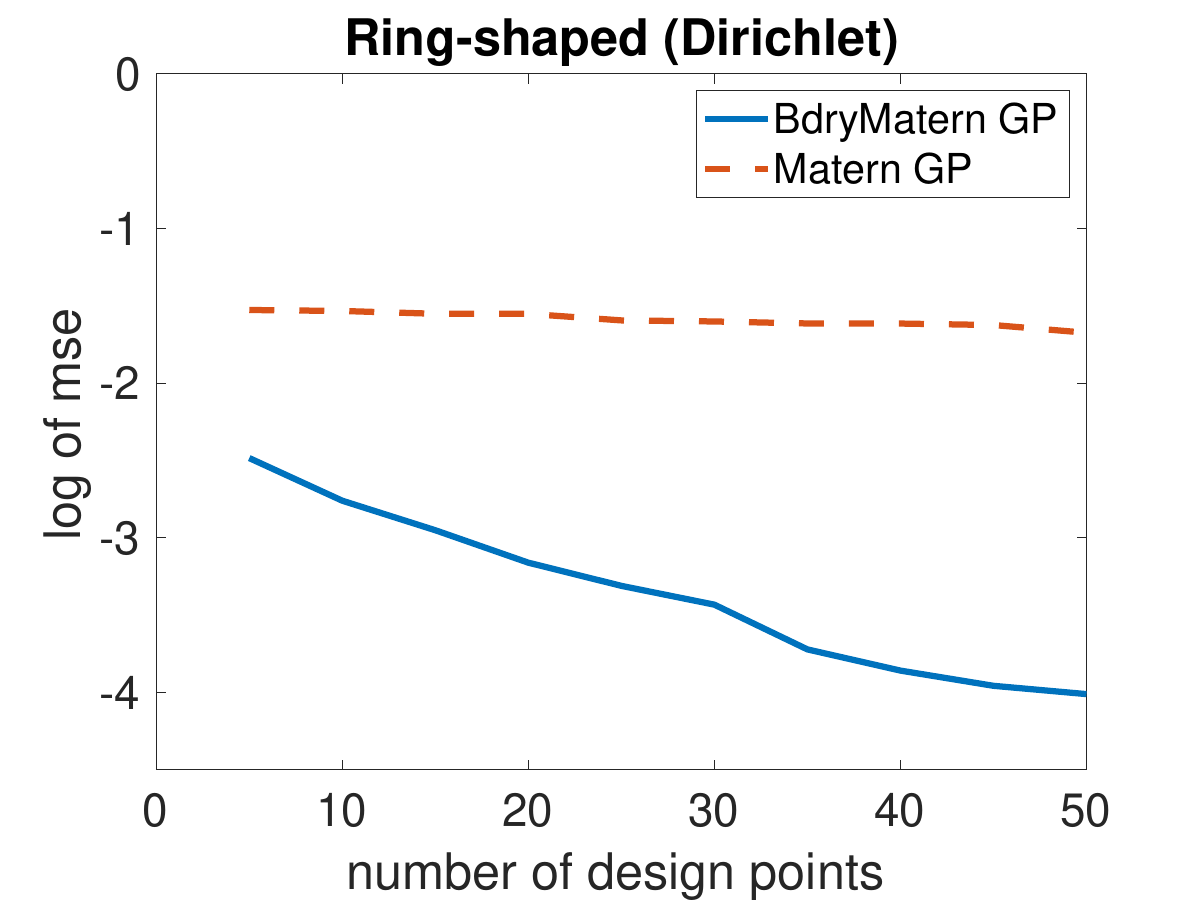}
\includegraphics[width=0.23\textwidth]{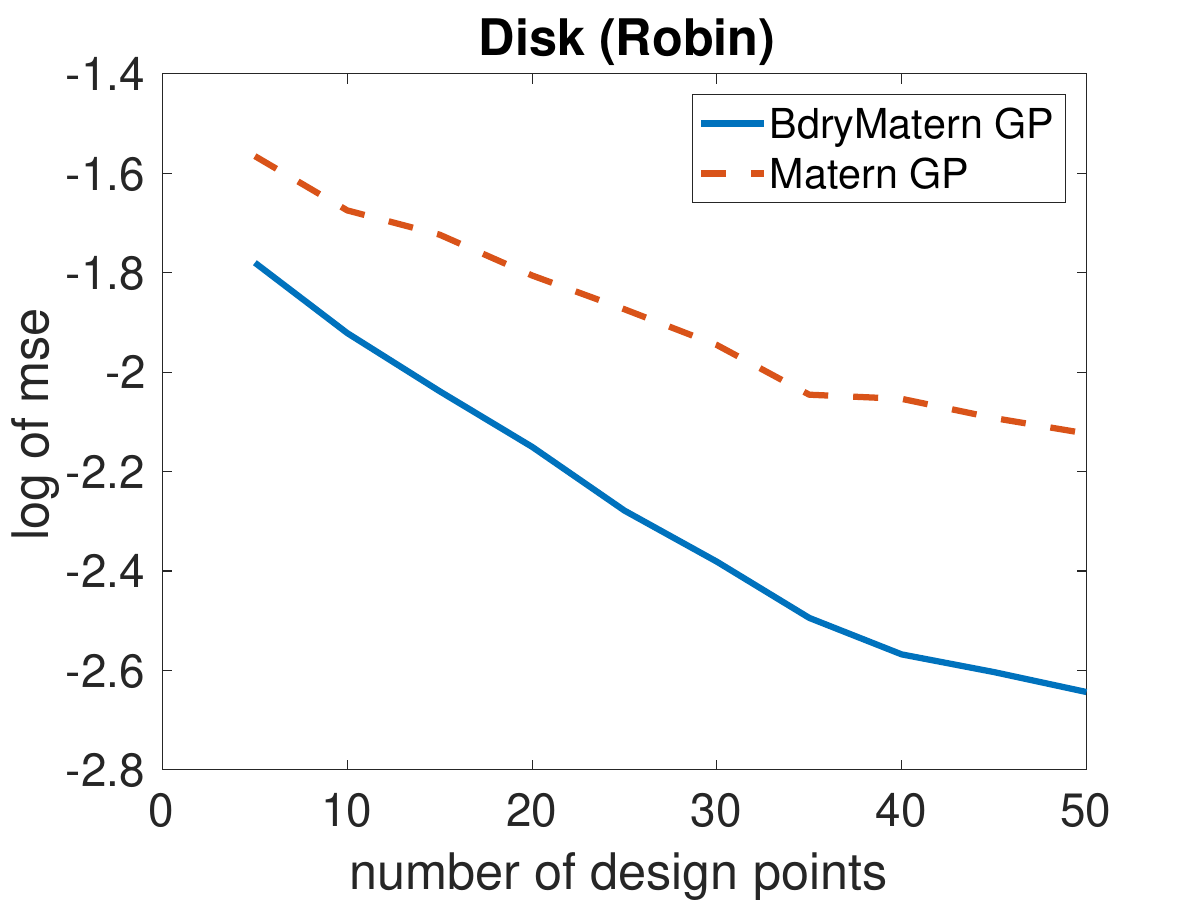}
\includegraphics[width=0.23\textwidth]{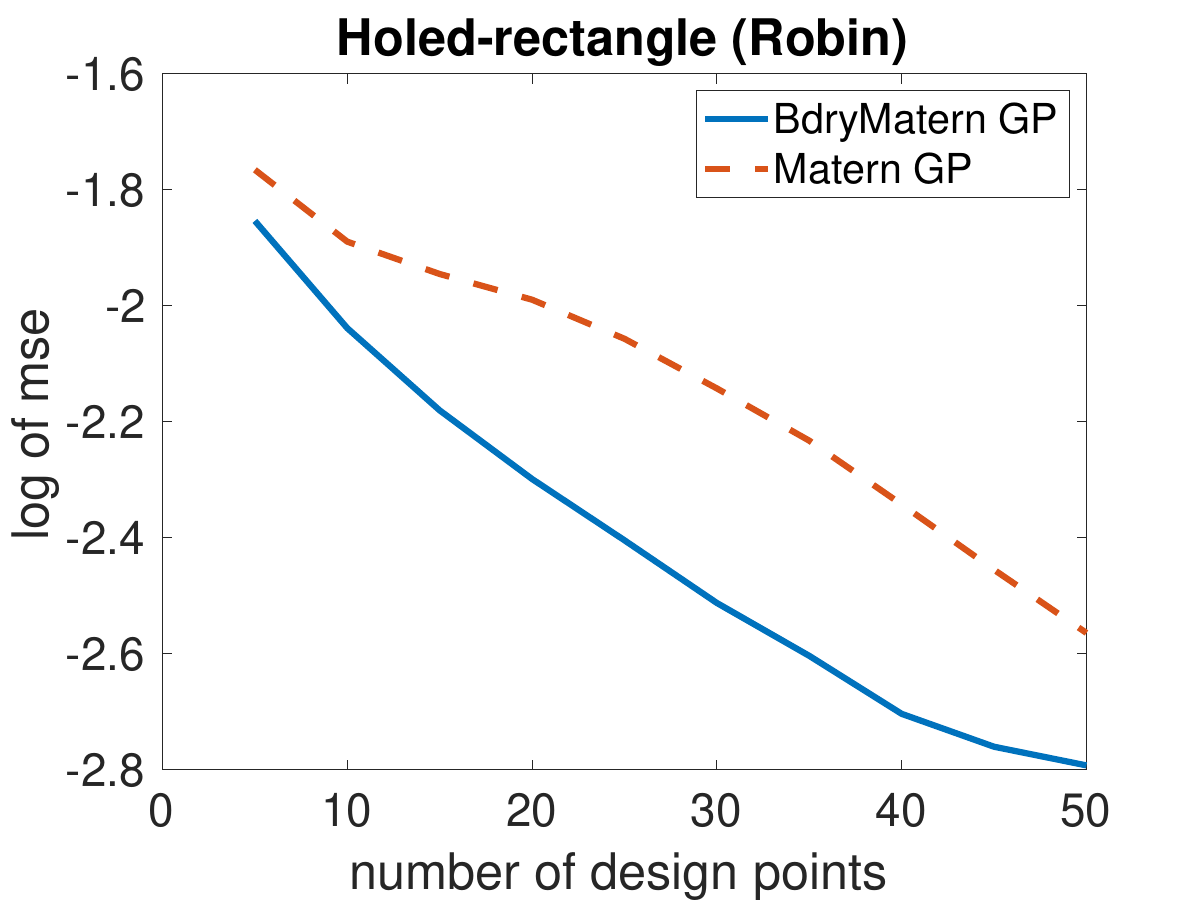}
\end{center}
\caption{Log-MSEs for the T-shaped, ring-shaped, disk-shaped, and holed-rectangle domain experiments, using the BdryMat\'ern GP and the Mat\'ern GP.}
\label{fig:2d_log_err}
\end{figure}

\begin{figure}
\hspace{2cm}
\begin{center}
\includegraphics[width=0.3\textwidth]
{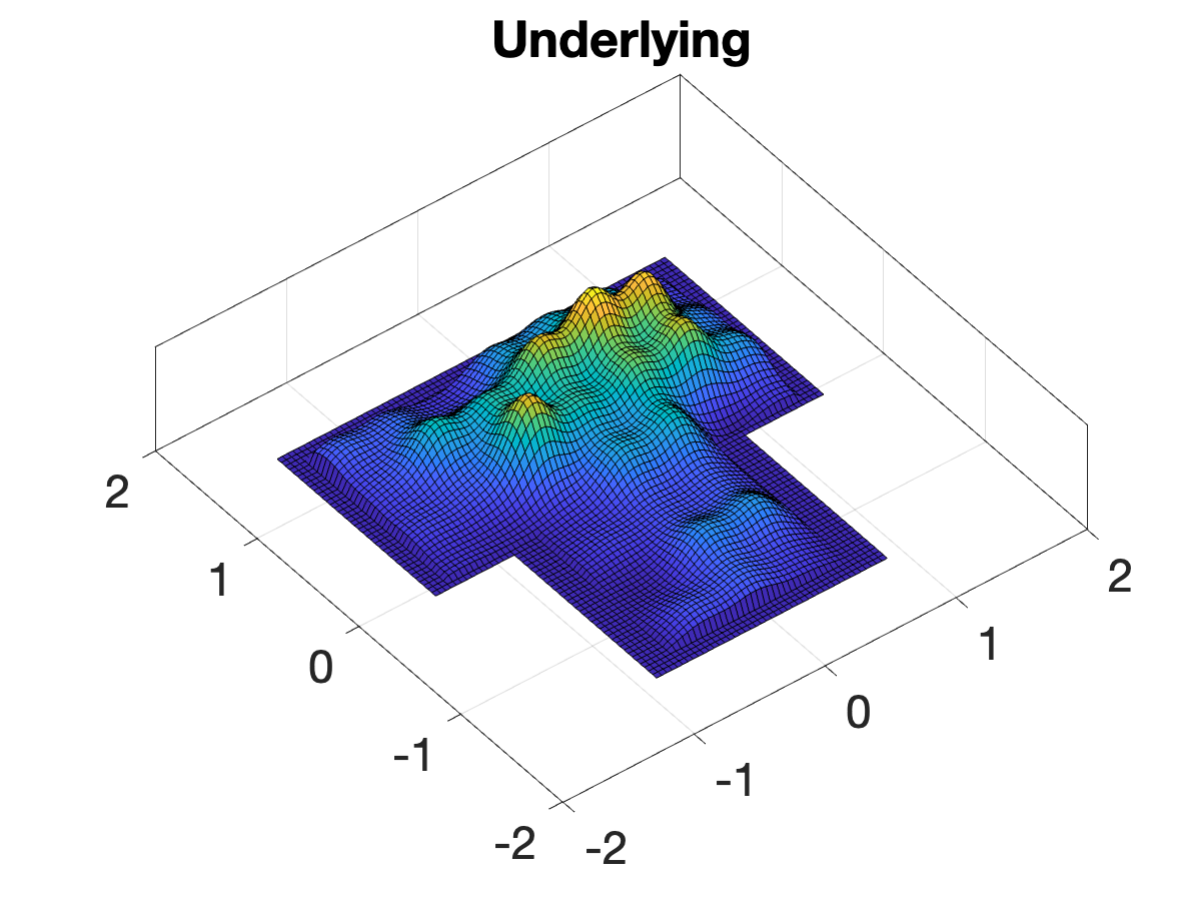}
\includegraphics[width=0.3\textwidth]{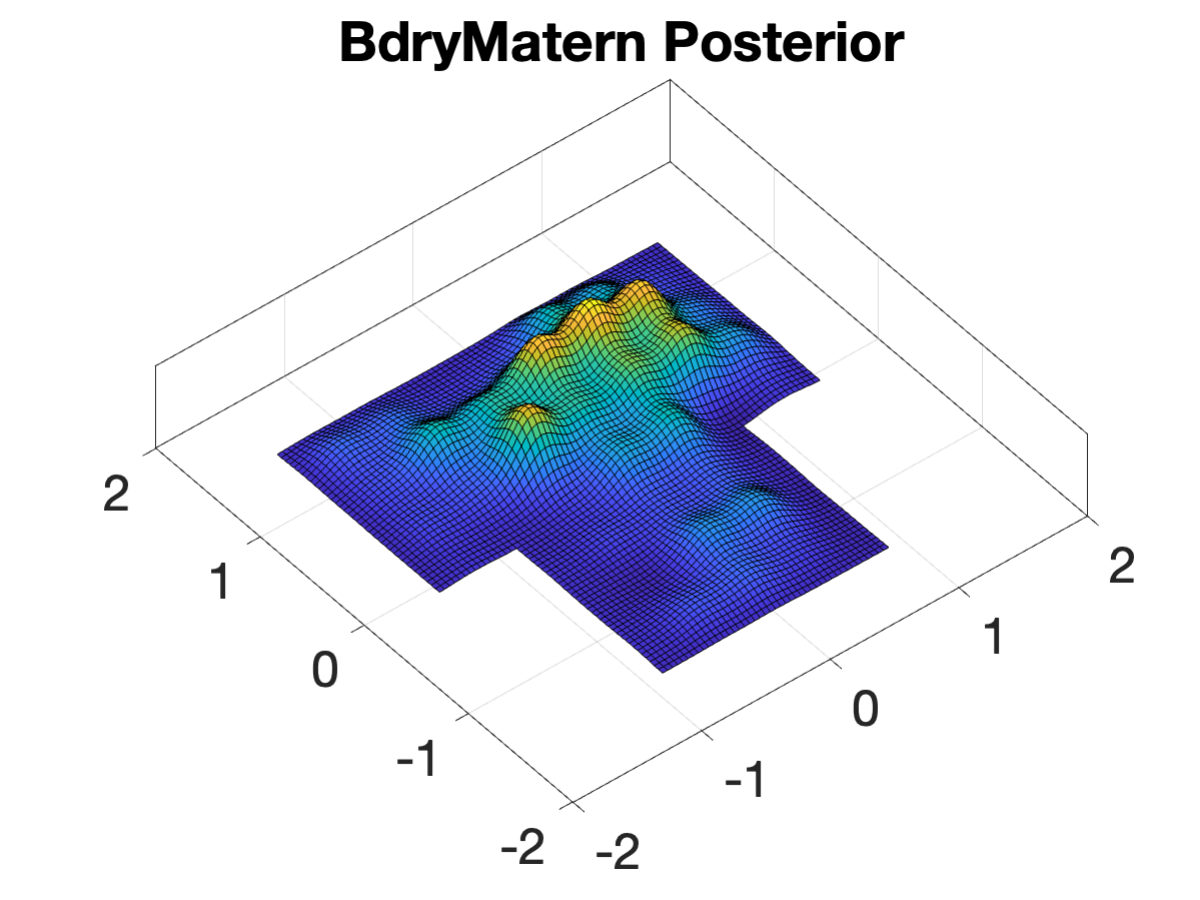}
\includegraphics[width=0.3\textwidth]{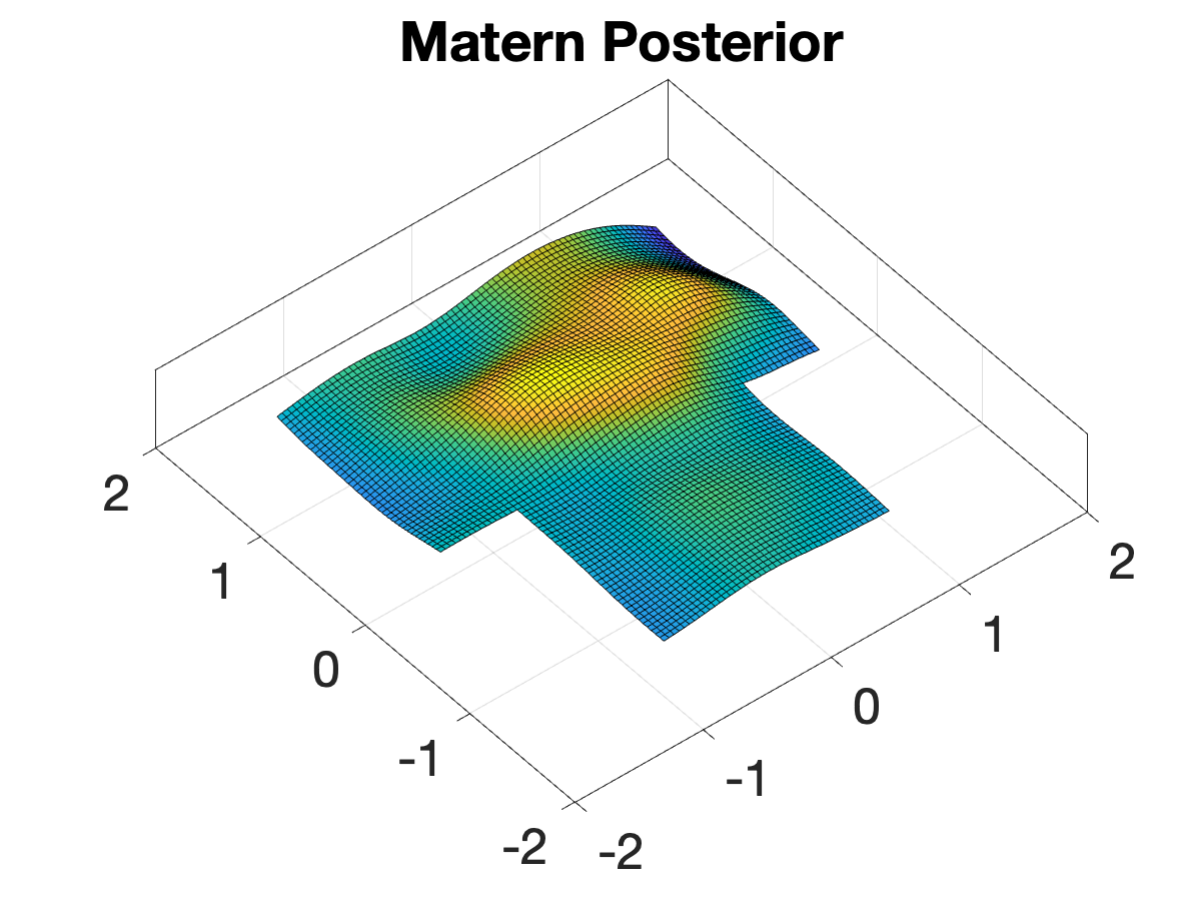}
\includegraphics[width=0.3\textwidth]
{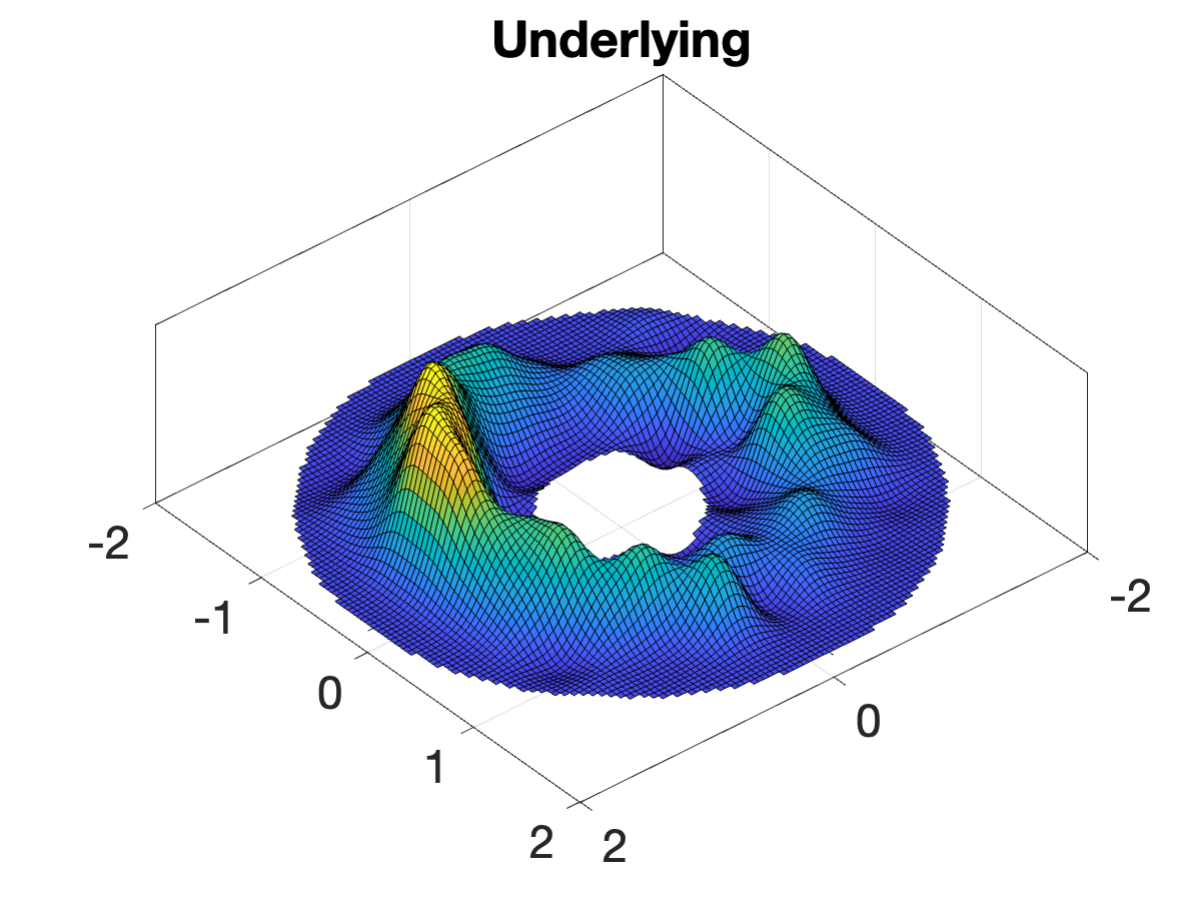}
\includegraphics[width=0.3\textwidth]{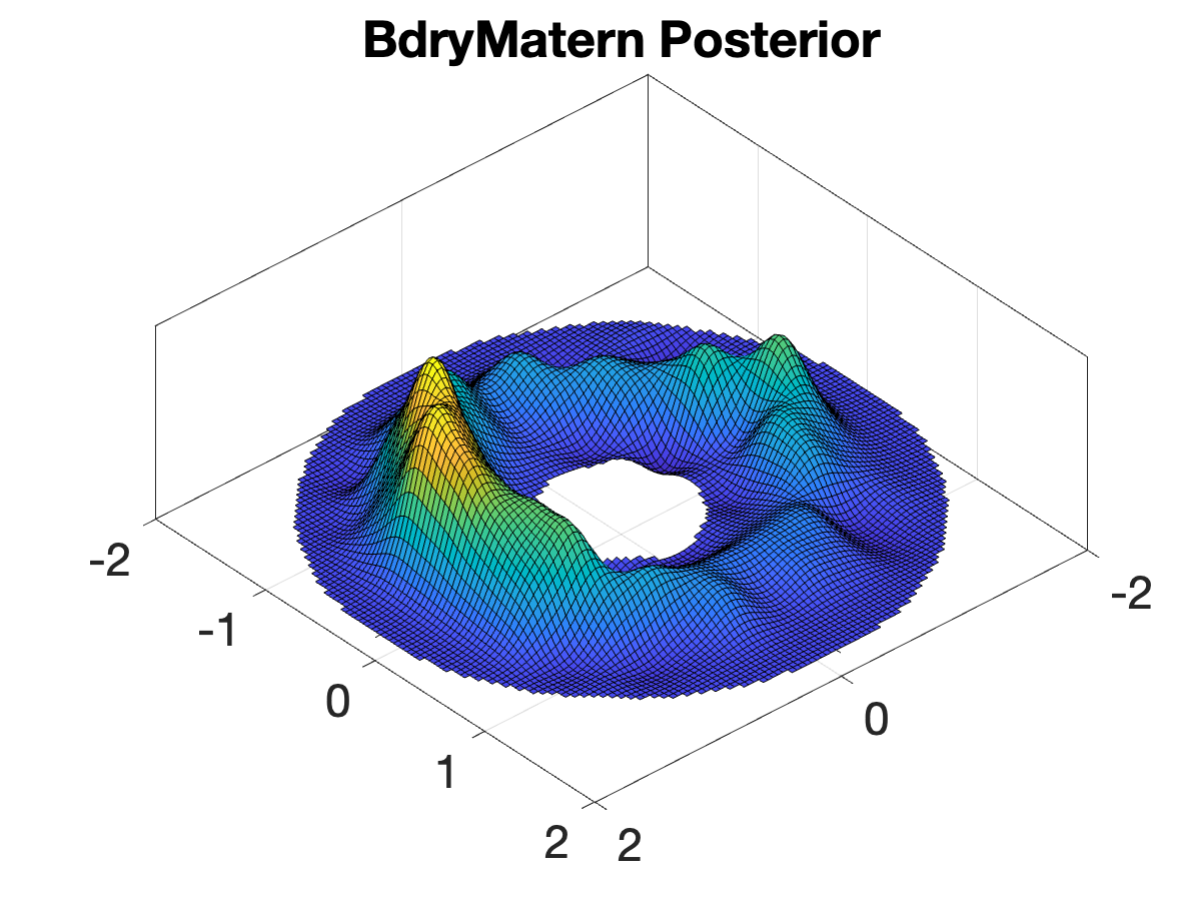}
\includegraphics[width=0.3\textwidth]{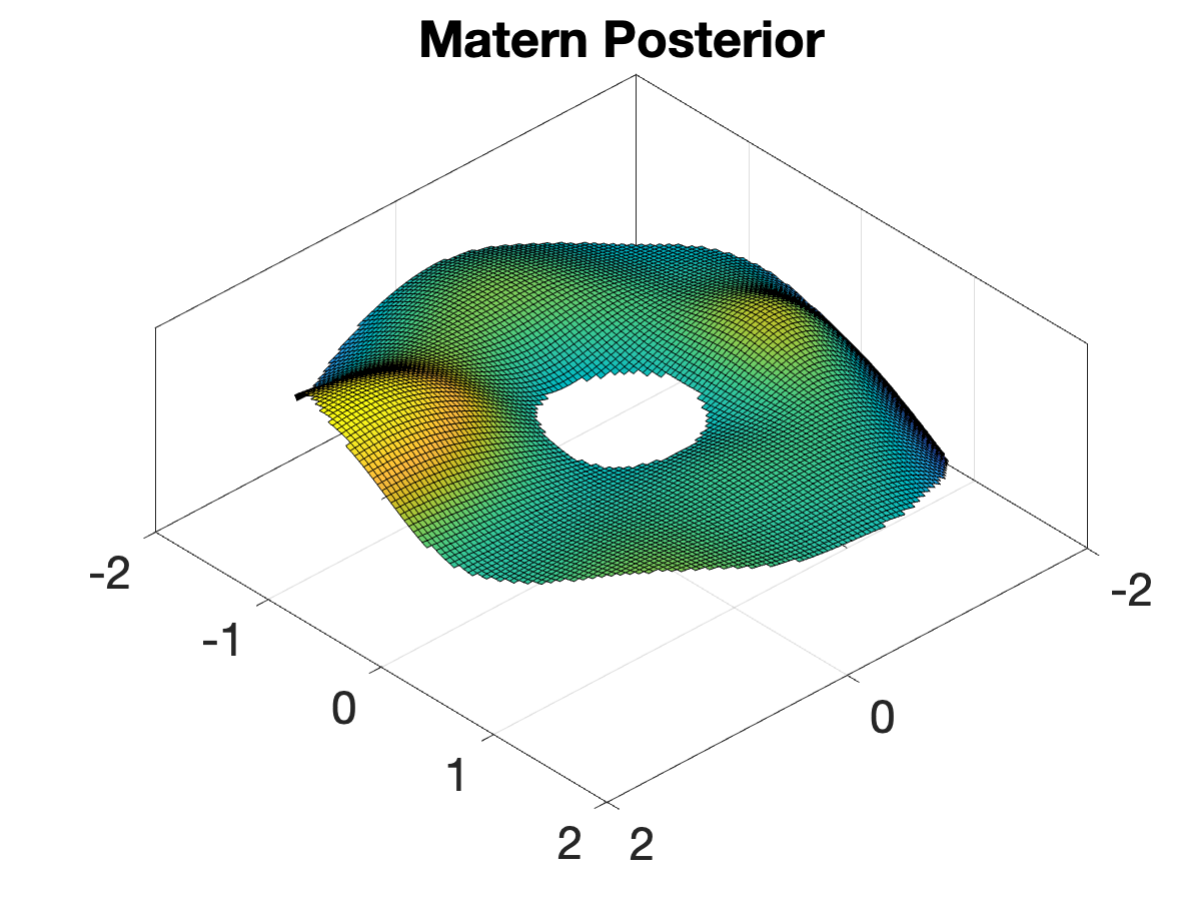}
\includegraphics[width=0.3\textwidth]
{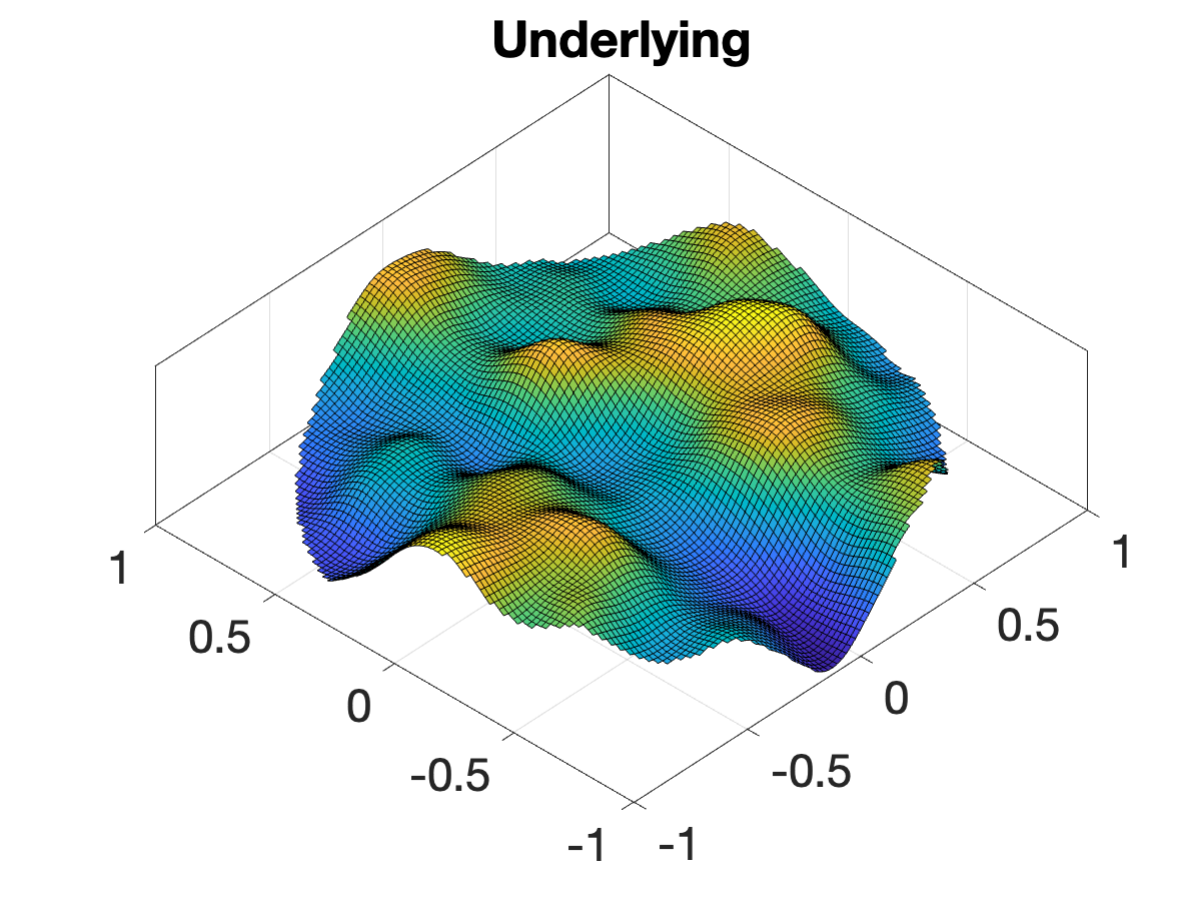}
\includegraphics[width=0.3\textwidth]{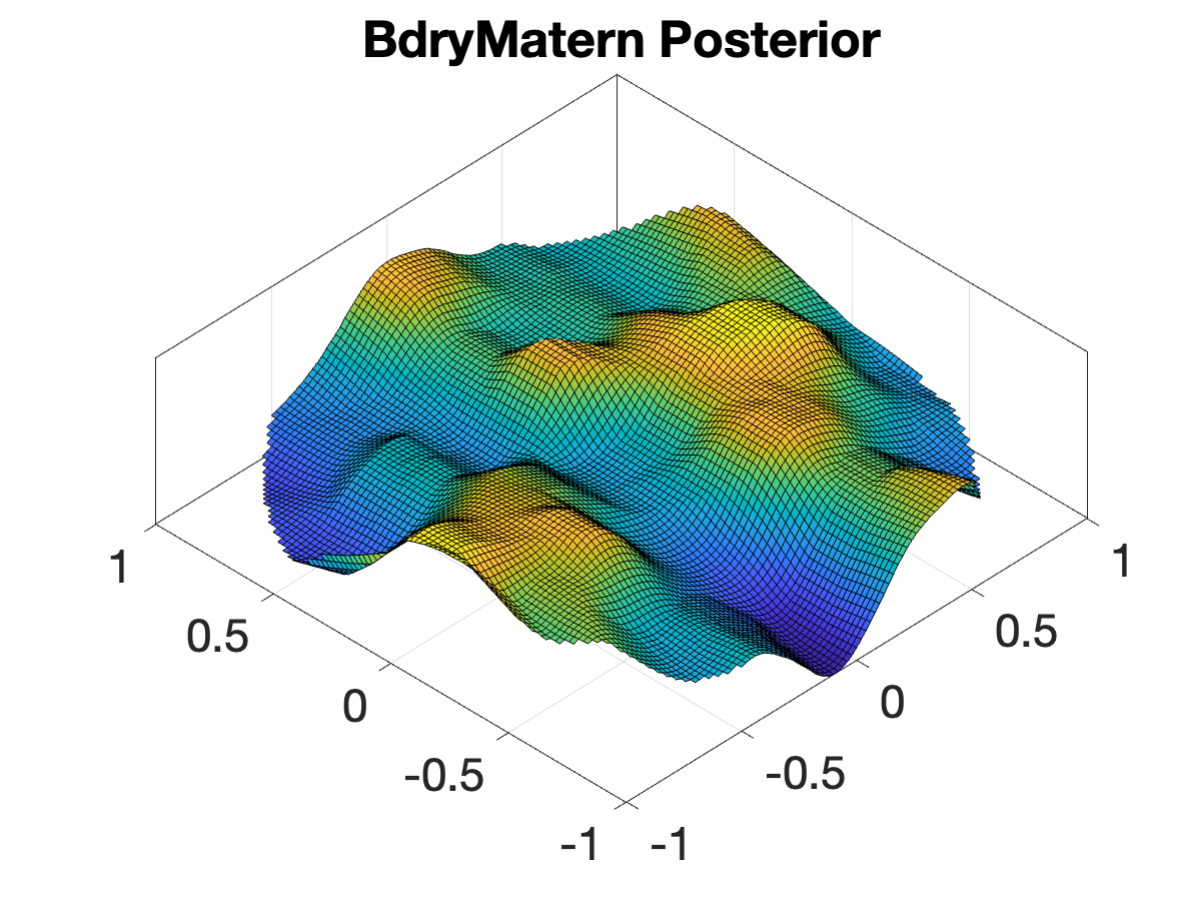}
\includegraphics[width=0.3\textwidth]{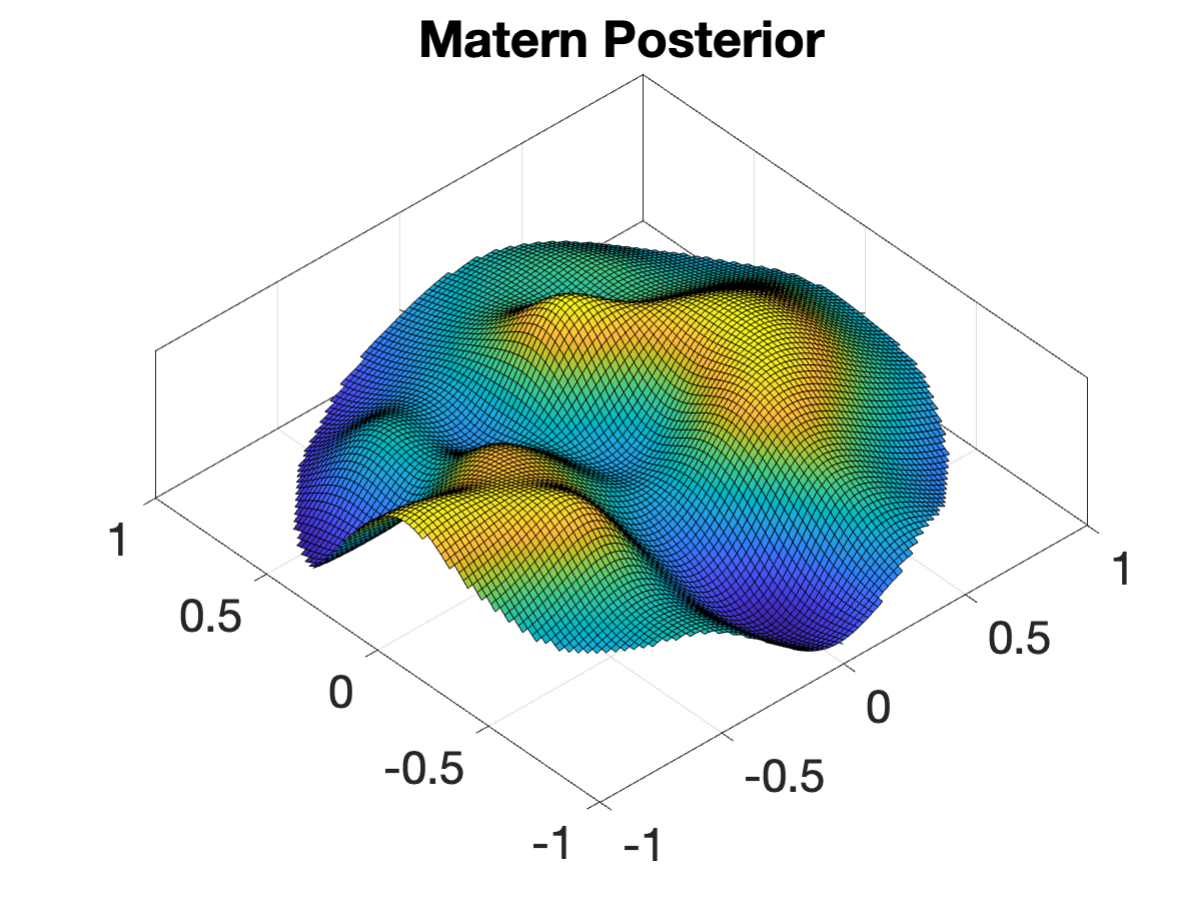}
\includegraphics[width=0.3\textwidth]
{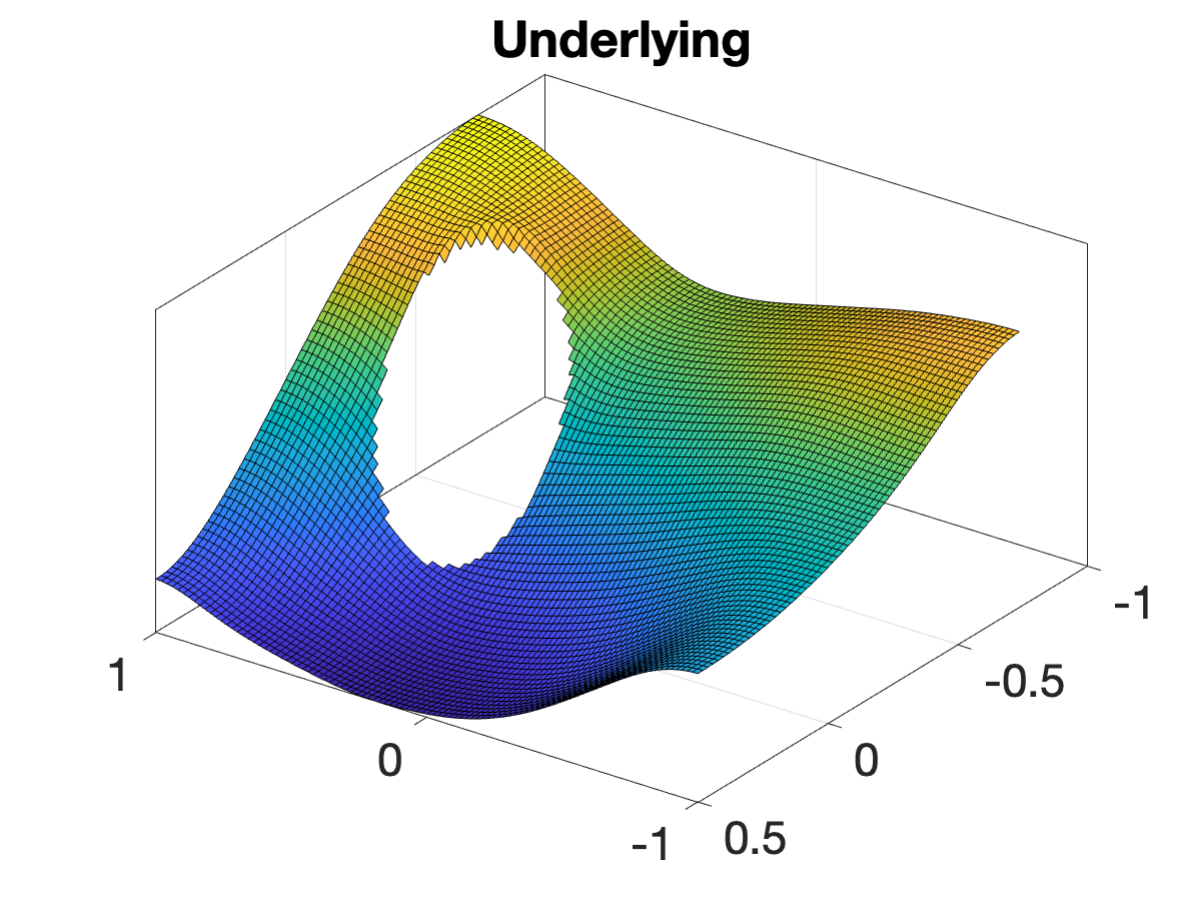}
\includegraphics[width=0.3\textwidth]{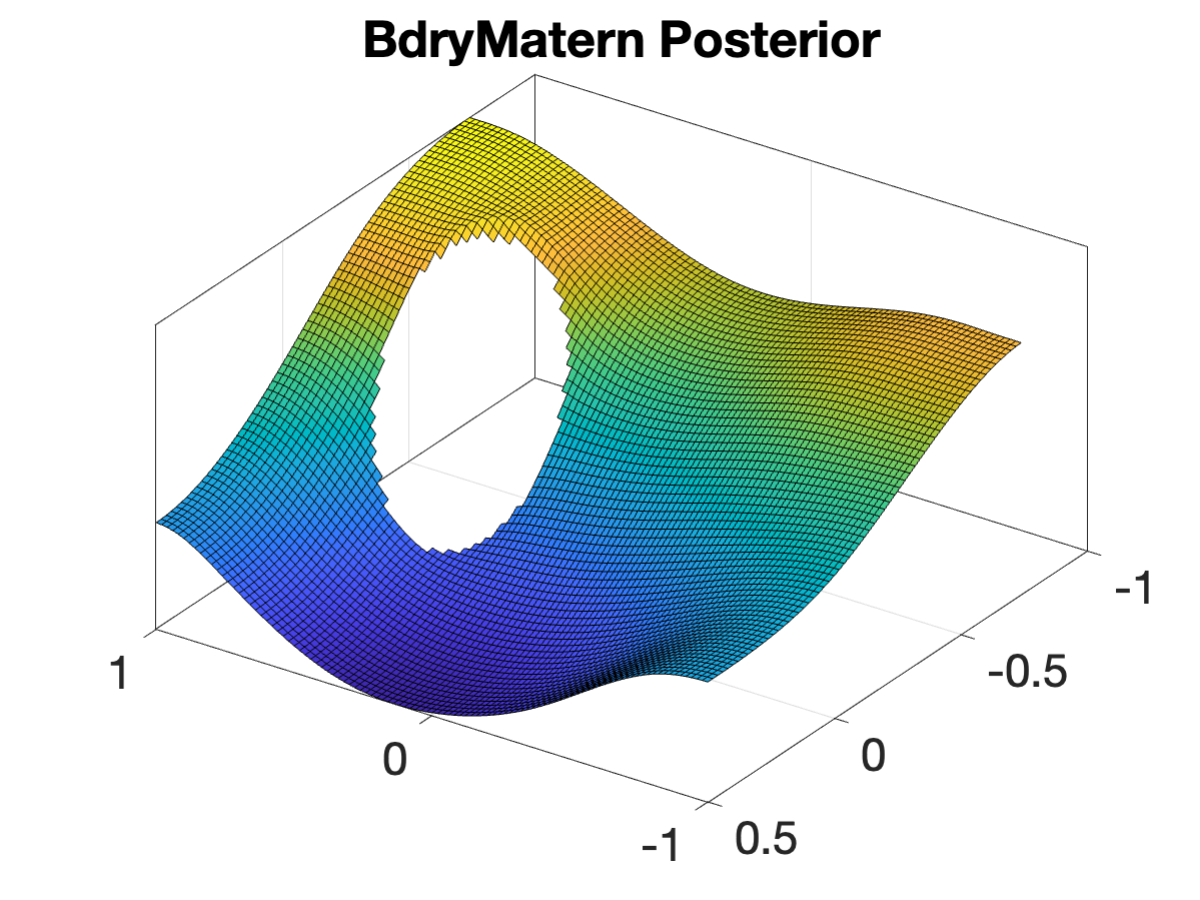}
\includegraphics[width=0.3\textwidth]{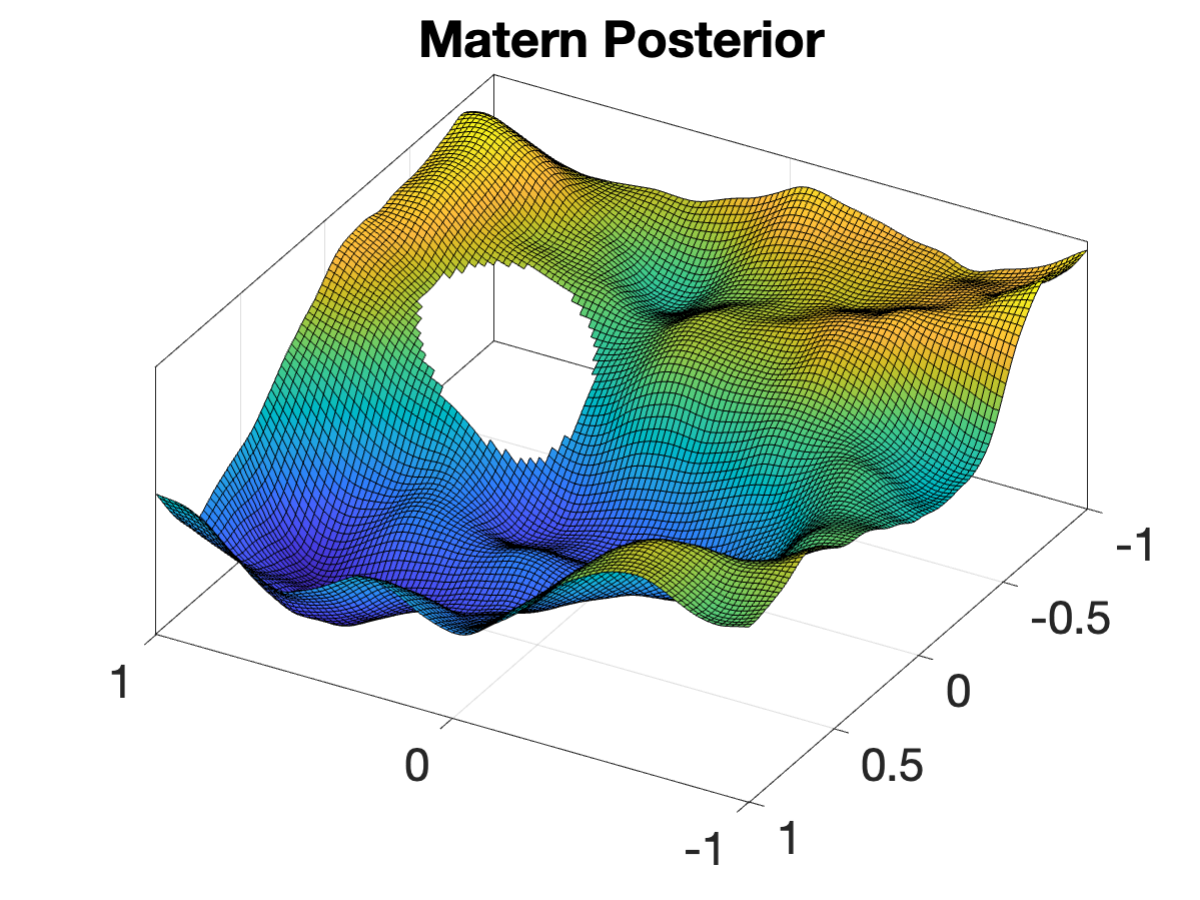}
\end{center}
\caption{
Visualizing the true underlying function $f$ (left), the BdryMat\'ern GP posterior mean (middle), and the Mat\'ern GP posterior mean functions (right) for the T-shaped, ring-shaped, disk-shaped, and holed-rectangle domain experiments (separated by columns).
}
\label{fig:2d_plot}
\end{figure}

\subsection{Tensor BdryMat\'ern GP}

We next investigate the performance of the tensor BdryMat\'ern GP (Section \ref{sec:tensor_mat}) for incorporating boundary information of the form \eqref{eq:boundar_mixed_full} on a higher-dimensional unit hypercube domain. We consider the two test functions from \cite{Surjano16}:
\begin{align}
\begin{split}
    \textit{Griewank:} \quad f(\Bx) &= \sum_{j=1}^{30}\frac{x_j^2}{4000}-\prod_{j=1}^{30}\cos\left(\frac{x_j}{\sqrt{j}}\right)+1 , \quad \Bx\in [-5,5]^{30},\\
    \textit{Product Sine:} \quad f(\Bx) &= \prod_{j=1}^{30}\left[1+\sin\left(\frac{\pi x_j}{2}\right)\right], \quad \Bx\in[-1,1 ]^{30}.
    \end{split}
\end{align}
Both domains are rescaled to the unit hypercube domain  prior to surrogate modeling. Dirichlet boundary information is provided for the Griewank test function; since this function has non-zero Dirichlet boundaries, the mean function $\mu_{\mathcal{B}}$ is specified via the approach in \cite{dalton2024boundary}. Neumann boundary information is provided for the product sine test function; since this function has a Neumann boundary of zero, the mean function is specified as $\mu_{\mathcal{B}} \equiv 0$.




With this, we then sample $n$ training design points from a sparse grid design \cite{Bungartz04}, where $n$ is varied depending on the order of the sparse grid. Sparse grid designs are broadly used for softening the curse-of-dimensionality in high-dimensional approximation problems \cite{Plumlee14}, thus our use of it here for high-dimensional boundary incorporation. With such data, we fit the tensor BdryMat\'ern GP (from Section \ref{sec:tensor_mat}) with smoothness parameters \(\nu=5/2\) and \(\nu=7/2\). As a benchmark, we then compare with a GP with a product Mat\'ern kernel with smoothness parameters $\nu = 5/2$ and $\nu=7/2$, which is the counterpart for the tensor BdryMat\'ern GP without boundary information. 


Figure \ref{fig:tensor} (left) shows the corresponding log-MSE of the compared models, and Figure \ref{fig:tensor} (middle and right) visualizes the fitted predictive models (and its uncertainties) over a slice of the prediction domain. From Figure \ref{fig:tensor} (left), we again see that for both test functions, the integration of boundary information via the tensor BdryMat\'ern GP can considerably improve predictive performance over its product Mat\'ern GP counterpart, which doesn't incorporate such information. The different error decay slopes over sparse grid level also suggest that this integration of boundary information may offer improved prediction rates; investigating this will be a topic of future work. Figure \ref{fig:tensor} (middle and right) shows that, in addition to improved point predictions, the incorporation of known boundaries  via the tensor BdryMat\'ern GP can further reduce predictive uncertainties, thus facilitating confident surrogate modeling with limited data.


\begin{figure}[!t]
\centering
\includegraphics[width=0.32\textwidth]{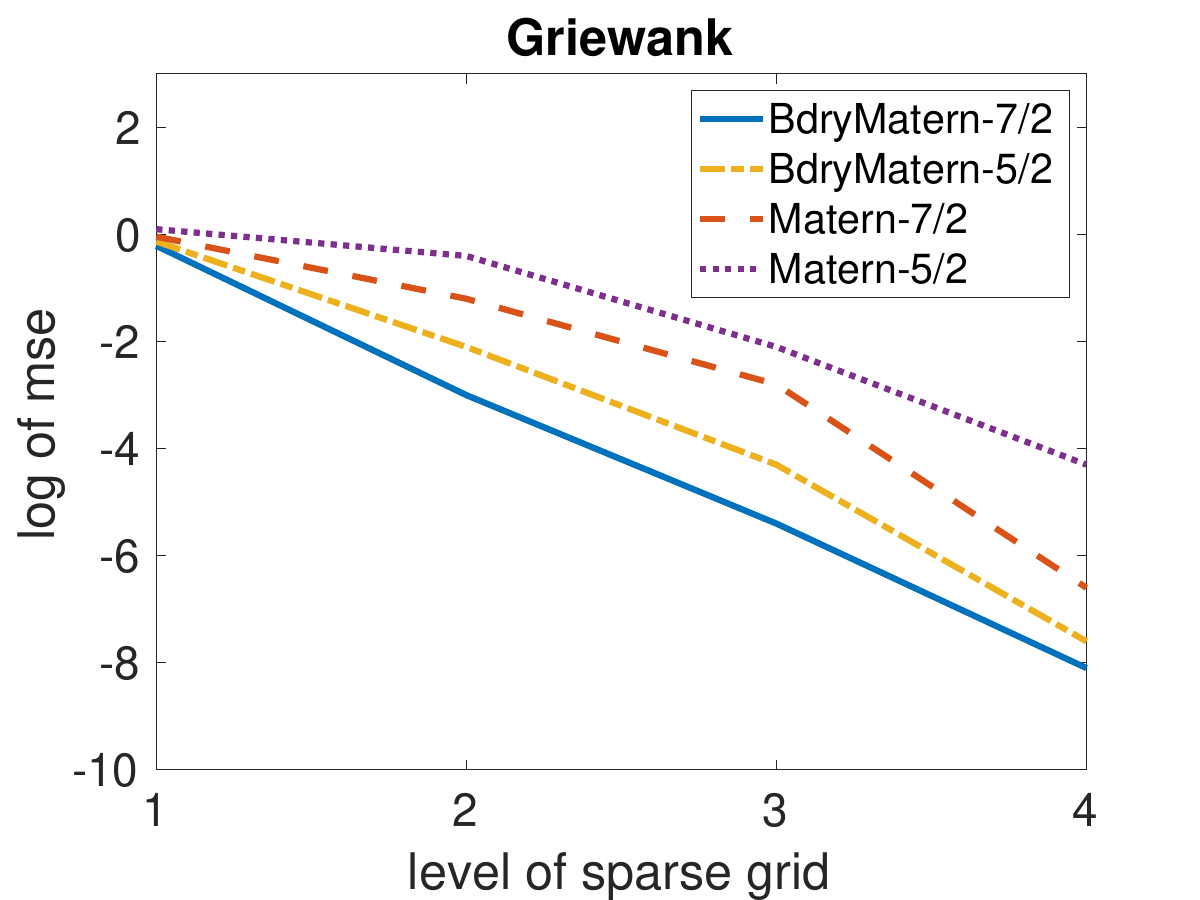}
\includegraphics[width=0.32\textwidth]{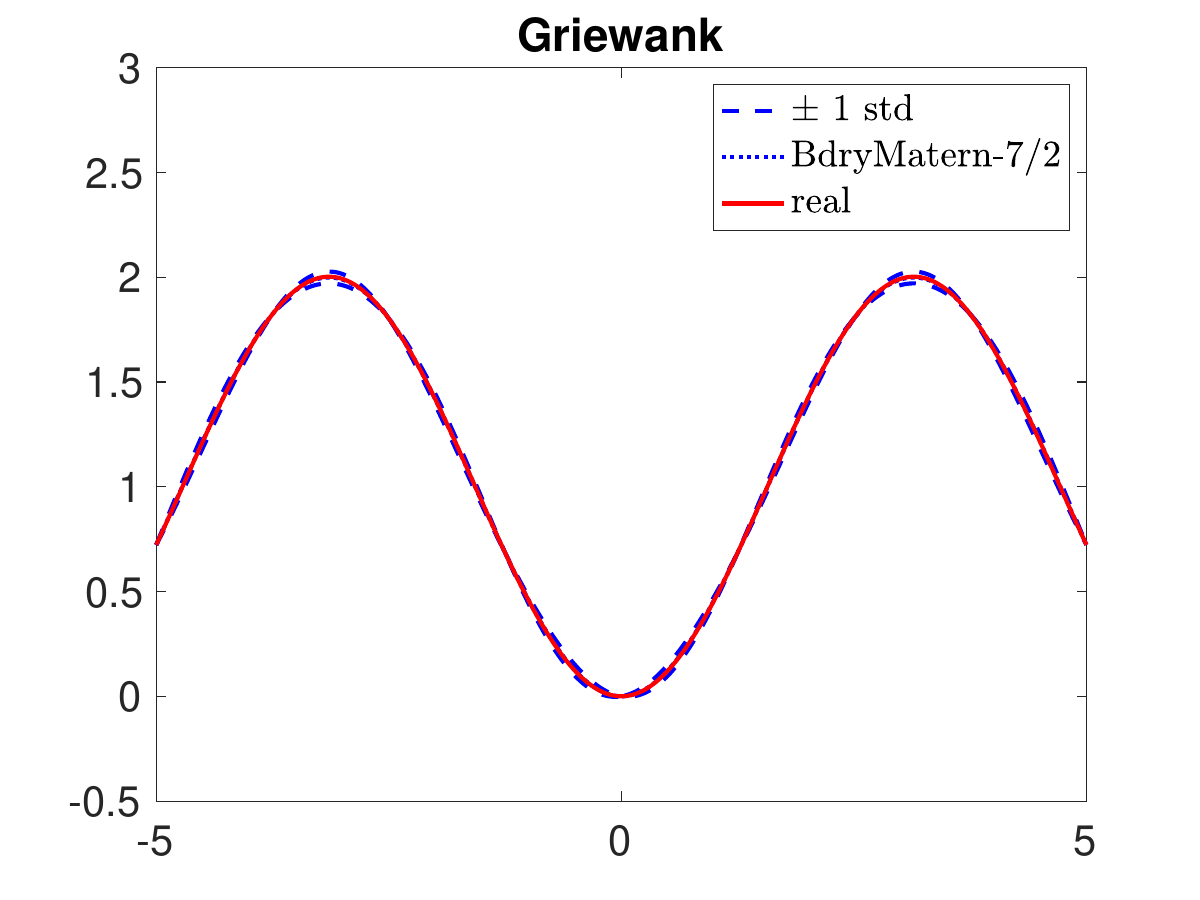}
\includegraphics[width=0.32\textwidth]{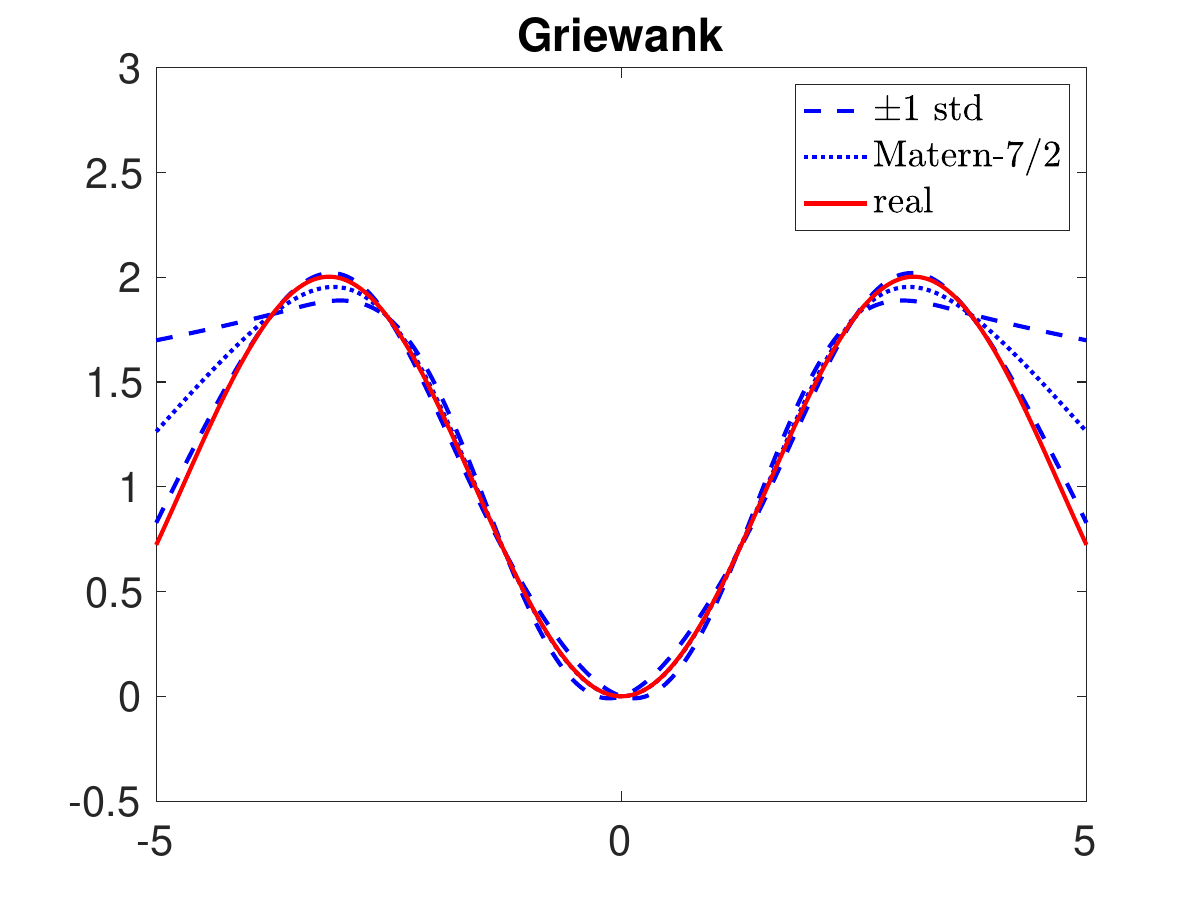}

\hspace{0.05cm}
\includegraphics[width=0.32\textwidth]{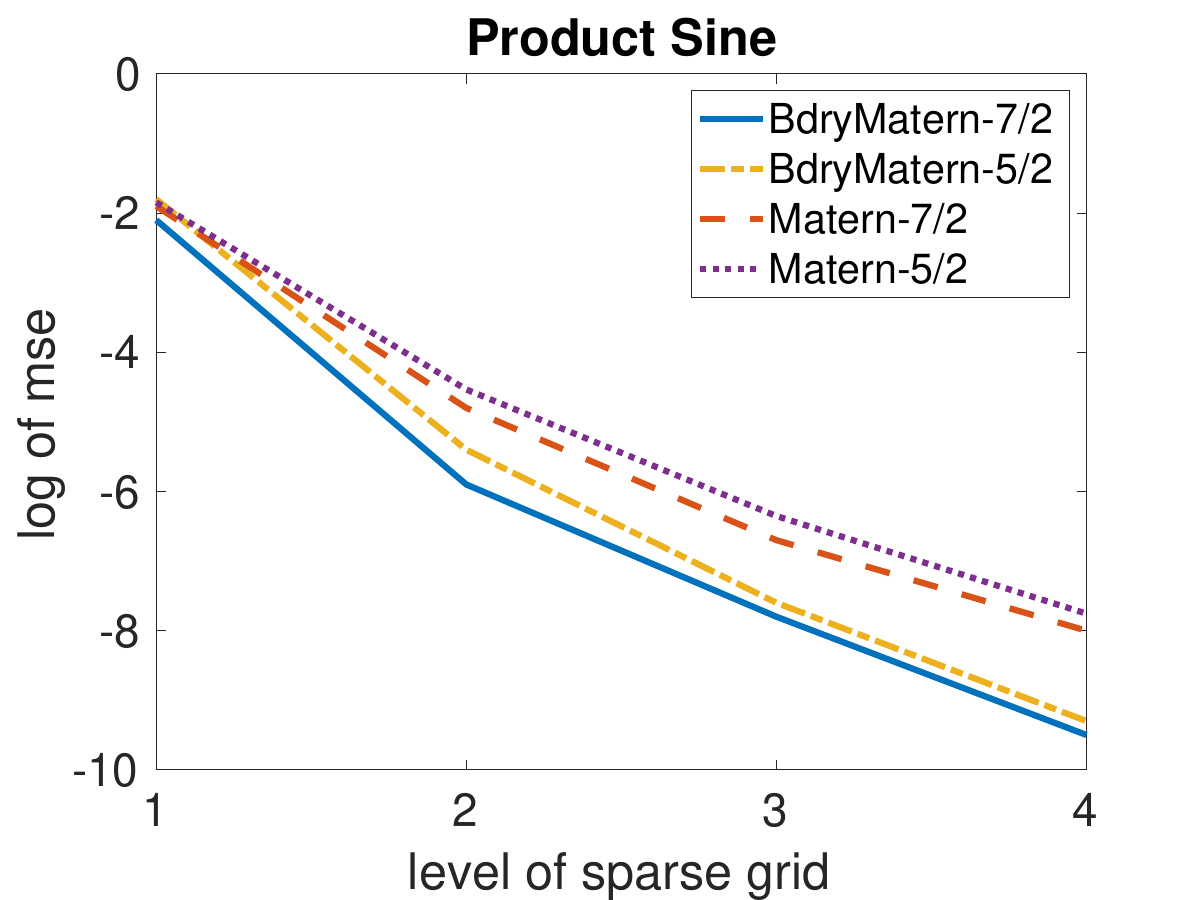}
\includegraphics[width=0.32\textwidth]{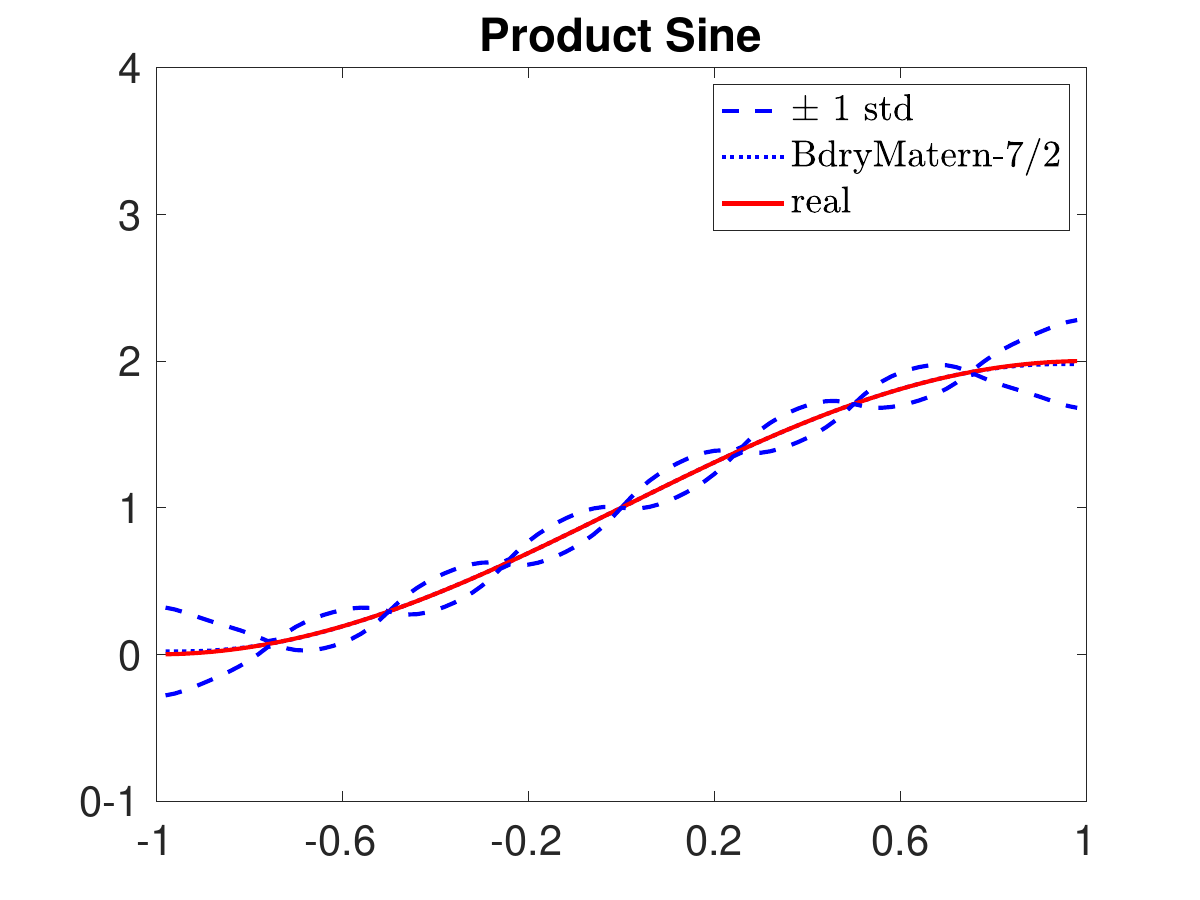}
\includegraphics[width=0.32\textwidth]{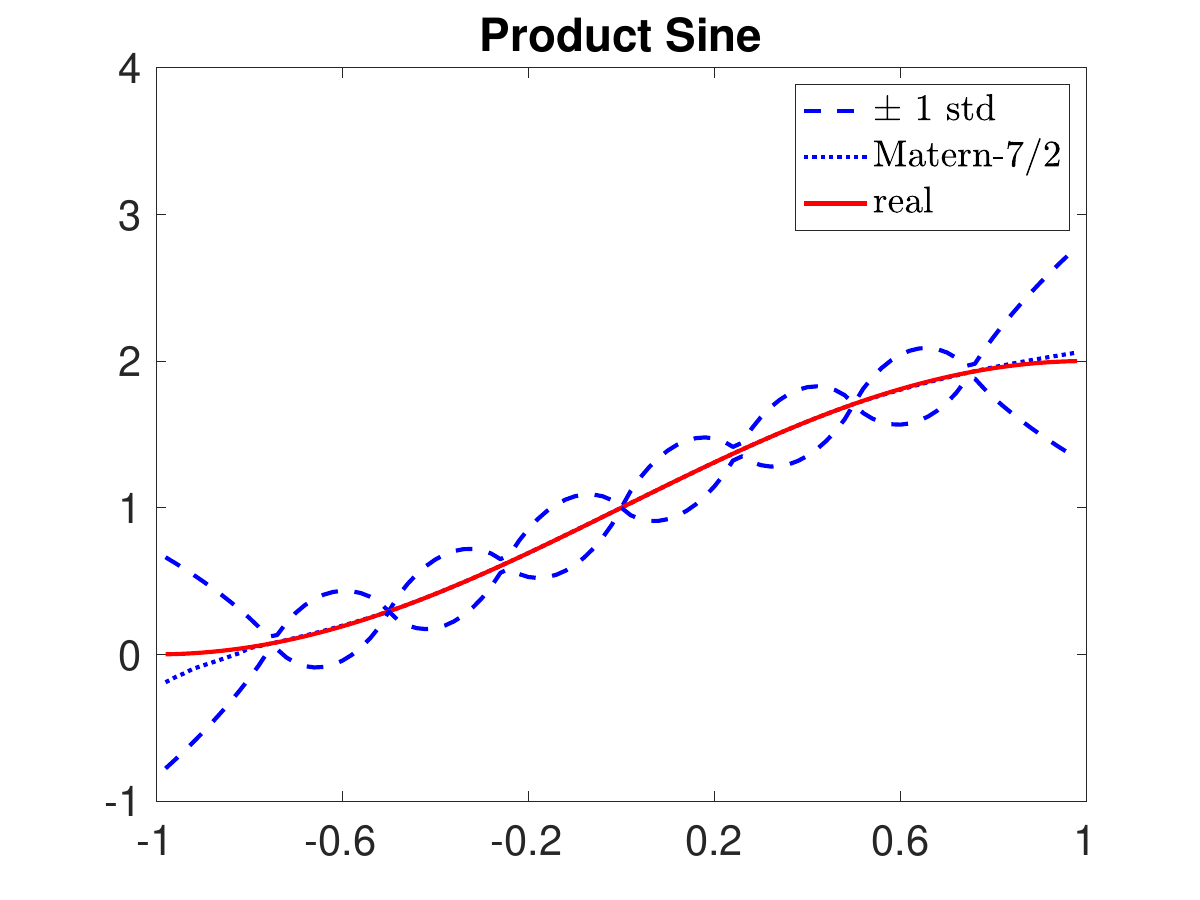}
\caption{(Left) Log-MSEs for the Griewank and product-sine test functions using the tensor BdryMat\'ern GP and the product Mat\'ern GP. (Middle) Posterior means and 68\% predictive intervals of the tensor BdryMat\'ern GP for the Griewank and product-sine functions. This is shown along a slice of the domain for input $x_1$ with other inputs set at their left end-points. (Right) Posterior means and 68\% predictive intervals of the product Mat\'ern GP for the Griewank and product-sine functions along the same domain slice.
}
\label{fig:tensor}
\end{figure}



\section{Conclusion}
\label{sec:conclusion}
This paper introduces a new Gaussian process model called the BdryMat\'ern GP, which targets the reliable incorporation of known Dirchlet, Neumann or Robin boundary information on a irregular (i.e., non-hypercube) domain $\mathcal{X}$. The BdryMat\'ern GP makes use of a new BdryMat\'ern covariance kernel, which we derive in path integral form using the Feynman-Kac formula \citep{ito1957fundamental,stroock1971diffusion,papanicolaou1990probabilistic}. Provided $\mathcal{X}$ is connected with a boundary set that is twice-differentiable almost surely, we show that sample paths from the BdryMat\'ern GP provide smoothness control (in terms of its derivatives) while satisfying the desired boundaries. We then present an efficient approach for approximating the BdryMat\'ern kernel on general domains $\mathcal{X}$ via finite element modeling with rigorous error analysis. When the domain takes the form of a high-dimensional unit hypercube, we present a tensor form of the BdryMat\'ern GP that facilitates efficient boundary integration via a closed-form covariance kernel. Finally, we demonstrate the effectiveness of the proposed model for reliable boundary integration in a suite of numerical experiments that feature broad boundary information on different domain choices.

 

Given the presented theoretical framework for boundary integration, there are several promising future directions for future work. One direction is the exploration of posterior contraction rates of the BdryMat\'ern GP on irregular domains $\mathcal{X}$, which can shed light on how much the integration of boundary information can improve probabilistic predictive performance. Another direction is the investigation of consistency and contraction rates for the BdryMat\'ern GP under inference (and potential misspecifications) of model parameters. Finally, it would be fruitful to investigate new space-filling experimental designs \citep{johnson1990minimax,mak2018minimax} that incorporate boundary information on irregular domains.

\begin{funding}
The second author gratefully recognizes funding from NSF CSSI Frameworks 2004571, NSF DMS 2210729, 2316012 and U.S. Department of Energy Grant DE-SC0024477.
\end{funding}


\bibliographystyle{imsart-number} 
\bibliography{bibliography}       


\pagebreak

\begin{appendix}
\stitle{}


\section{Proof of Theorem \ref{thm:smoothdir}}
\label{pf:smoothdir}
\begin{proof}[Proof of ${(a)}$]
Given our assumption that $\CalX$ is compact and connected with boundary set twice-differentiable almost everywhere, the GP with Mat\'ern covariance $k_{\nu}$ admits an eigendecomposition with respect to the Lebesgue measure on $\CalX$ of the form:
\begin{equation}
    \label{eq:smoothdir_pf_1}
    \CalW_\nu(\BFx)=\sum_i \sqrt{\lambda_i}\xi_i(\BFx)Z_i,
\end{equation}
where $\{Z_i\}_i$ are i.i.d. standard normal random variables, $\{\lambda_i\}_i$ and $\{\xi\}_i$ are the eigenvalues and eigenfunctions of $k_{\nu}$ under the Lebesgue measure. Moreover, since:
\[\int_\CalX k_\nu(\Bx,\Bx)d\Bx=\sum_i\lambda_i^2\int_\CalX\xi^2_i(\Bx)d\Bx\leq k_\nu(0) {\rm Vol}(\CalX),\]
it follows that $\sqrt{\lambda_i}\xi_i\in L^2(\CalX)$ for any $i$. Let $g_i$ be the solution to the following elliptical PDE:
\begin{equation}
    \label{eq:smoothdir_pf_2}
    (\kappa^2-\bigtriangleup)g_i=\sqrt{\lambda_i}\xi_i,  \quad g_i(\BFx)=0, \quad \forall\BFx \in\partial \CalX.
\end{equation}
Thus, given that $\partial \mathcal{X}$ is twice-differentiable almost everywhere and \cite[Theorem 4, Chapter 6]{Evans15}, we have that $g_i$ is $L^2$ twice-differentiable and $g_i=0$ on the boundary.

For finite $i^*$, define:
\[f_{i^*}=\sum_{i=1}^{i^*} g_iZ_i,\quad \CalW_{\nu,i^*}= \sum_{i=1}^{i^*} \sqrt{\lambda_i}\xi_i(\BFx)Z_i.\]
By the superposition principle of linear PDEs, we have:
\[   (\kappa^2-\bigtriangleup)\CalW_{\nu,i^*}=\sum_{i=1}^{i^*} Z_i   (\kappa^2-\bigtriangleup)\sqrt{\lambda_i}\xi_i=f_{i^*},\]
so $f_{i^*}$ satisfies the boundary condition almost surely for any $i^*$. Then by Theorem 12.1 in \cite{lifshits2012lectures}, we have:
\[f=\lim_{i^*\to\infty}f_{i^*},\quad {\rm a.s.,}\]
where $f$ is the solution to \eqref{eq:bdry_SPDE}. So $f$ satisfies the boundary condition almost surely.
\end{proof}

\begin{proof}[Proof of ${(b)}$] 
    We shall prove a stronger version of the theorem. We first define the concept of $\nu$-almost H\"older continuity as follows:
    \begin{definition}
        Let $C^\gamma(\CalX)$ be the $\gamma$-H\"older space. The $\nu$-almost H\"older space on $\CalX$ is defined as
        \[C^{-\nu}(\CalX)=\cap_{\gamma<\nu}C^{\gamma}(\CalX).\]
    \end{definition}
    From \cite[Proposition 10]{da2023sample}, we know that for any Mat\'ern-$\nu$ GP $\CalW_\nu$, its sample paths are $\nu$-almost H\"older. So the solution to \eqref{eq:bdry_SPDE} satisfies:
    \[\bigtriangleup f\in C^{-\nu}(\CalX),\quad {\rm a.s.}\]
    According to the definition of $\bigtriangleup=\sum_{l=1}^d\partial_{x_lx_l}$, we can conclude that $f\in C^{-(\nu+2)}(\CalX)$ almost surely. Hence, for any non-integer $\nu$, $f$ is $\lfloor\nu+2\rfloor$-differentiable and for any integer $\nu$, $f$ is $\nu+1$-differentiable.
\end{proof}

\section{Proof of Theorem \ref{thm:smooth}}
\label{pf:smoothrobin}
\begin{proof}
    The proof is identical to the proof of Theorem \ref{thm:smooth} except that \eqref{eq:smoothdir_pf_2} is replaced by
    \begin{equation}
    \label{eq:smooth_pf_1}
    (\kappa^2-\bigtriangleup)g_i=\sqrt{\lambda_i}\xi_i,  \quad \frac{\partial g_i(\Bx)}{\partial \bold{n}}+c(\mathbf{x})g_i(\Bx)=0, \quad \forall\BFx\in\CalX. 
\end{equation}
However, the change in boundary conditions does not affect our conclusions; all the theorems remain valid under Robin boundary conditions.
    
\end{proof}


\section{Proof of Theorem \ref{thm:convergence_general}}
\label{pf:convergence_general}

To prove the theorem, we first start with some useful lemmas. The first Lemma can be found in Section 5 of \cite{tropp2012user}:

\begin{lemma}[Chernoff inequality of random matrix]
\label{lem:matrix_chernoff}
    Let $p\in\NatInt$ and $\delta\in[0,1)$ be arbitrary and let $\bold{M}_1,\cdots, \bold{M}_m\in\Real^{p\times p}$ be independent, symmetric and positive semidefinite matrices with random entries. Let $R>0$ be such that $\lambda_{\rm max}(M_i)\leq R$ holds for all $i=1,...,m$. Then,it holds that:
    \begin{align*}
        &\pr\left(\lambda_{\rm min}(\sum_{i=1}M_i)\leq (1-\delta)\underline{\lambda}\right)\leq p\left(\frac{e^{-\delta}}{(1-\delta)^{1-\delta}}\right)^{\underline{\lambda}/R}\\
        &\pr\left(\lambda_{\rm max}(\sum_{i=1}M_i)\geq (1+\delta)\overline{\lambda}\right)\leq p\left(\frac{e^{\delta}}{(1+\delta)^{1+\delta}}\right)^{\overline{\lambda}/R},
    \end{align*}
    where $\lambda_{\max}(M)$ and $\lambda_{\min}(M)$ denote the maximum and minimum eigenvalues of a matrix $M$, respectively, and $\overline{\lambda}=\lambda_{\rm max}(\E \sum_{i=1}M_i)$ and $\underline{\lambda}=\lambda_{\rm min}(\E \sum_{i=1}M_i)$.
\end{lemma}

From Lemma \ref{lem:matrix_chernoff}, we can derive the following lemma in a straightforward way.
\begin{lemma}
\label{lem:spectral_empirical_gram}
    Let $M=[\int_{{\CalX}}\phi_i(\BFx)\phi_j(\BFx)d\BFx]$ be the Gram matrix of basis functions $\{\phi_i\}_{i=1}^p$ and let $\underline{\lambda}=\lambda_{\rm min}(M)$ and $\overline{\lambda}=\lambda_{\rm max}(M)$. Then:
    \begin{align}
        &\pr\left(\lambda_{\rm min}(\frac{1}{m}\sum_{i=1}\BPhi(\Bu_i)\BPhi(\Bu_i)^T)\leq \frac{\underline{\lambda}}{2} \quad \textrm{or} \quad \lambda_{\rm max}(\frac{1}{m}\sum_{i=1}\BPhi(\Bu_i)\BPhi(\Bu_i)^T)\geq \frac{3\overline{\lambda}}{2}\right)\nonumber\\
        \leq & p(\sqrt{2}e^{-1/2})^{m\underline{\lambda}/S_\phi}+p(\frac{e^{1/2}}{(3/2)^{3/2}})^{m\overline{\lambda}/S_\phi}. \label{eq:spectral_empirical_gram}
    \end{align}
\end{lemma}
\proof
Let $M_i$ denote the i.i.d. sample of Gram matrix $\frac{1}{m}\BPhi(\Bu_i)\BPhi^T(\Bu_i)$. Note that $\lambda_{\rm max}(M_1)\leq {S_\phi}/{m}$ almost surely because:
\[\lambda_{\rm max}(M_i)=\frac{1}{m}\BPhi(\Bu_i)^T\BPhi(\Bu_i)=\frac{1}{m}\sum_{j=1}\phi_j(\Bu_i)\phi_j(\Bu_i)\leq \frac{S_\phi}{m}.\]
Applying Lemma \ref{lem:matrix_chernoff} with $R=S_\phi/m$ and $\delta=1/2$, we can obtain the result.
\endproof

\proof[Proof of Theorem \ref{thm:convergence_general}]
Let $\CalP$ denote the $L_2$ projection onto the $\CalF_p=\text{span}\{\phi_i\}_{i=1}^p$. Similarly,
 the operator $\BPhi(\cdot)^T\BPhi(\BU)^\dagger$ can then be viewed as an empirical projection onto  $\CalF_p$ induced by $\BU$. Let $\CalP_m$ denote this empirical projection. For brevity, let $k_{\nu}$ denote the Mat\'ern-$\nu$ kernel in the following content.
 
 Based on the eigenfunction decomposition \eqref{lem:eigen_mat} and the path integral representation of BdryMat\'ern kernel (\eqref{eq:bdryMatern_path_integral_dirichlet} for Dirichlet setting; \eqref{eq:bdryMatern_path_integral_robin} for Robin setting), the BdryMat\'ern kernel and its regression kernel can be written as:
 \begin{align}
      k_{\nu,\mathcal{B}}=&\sum_i\lambda_ie_i(\BFx)e_i(\BFx')-\sum_i\lambda_ie_i(\Bx)\E[G_i(\BFx')]\nonumber\\
     &-\sum_i\lambda_ie_i(\Bx')\E[G_i(\Bx)]+\sum_i\lambda_i\E[G_i(\Bx)]E[G_i(\Bx')],\label{eq_pf_general_convergence_1}
 \end{align}
 and:
 \begin{align}
     \hat{k}_{Q,m}(\Bx,\Bx')=&\sum_i\lambda_i\CalP_m[e_i](\BFx)\CalP_m[e_i](\BFx')-\sum_i\lambda_i\CalP_m[e_i](\Bx)\CalP_m[G_i](\BFx')\nonumber\\
     &-\sum_i\lambda_i\CalP_m[e_i](\Bx')\CalP_m[G_i](\BFx)\nonumber\\
     &+\sum_i\lambda_i\CalP_m[G_i](\Bx)\CalP_m[G_i](\BFx'),\label{eq_pf_general_convergence_2}
 \end{align}
where:
\begin{equation*}
    G_i(\BFx)\coloneqq\begin{cases}
        &e^{-\kappa^2\tau} e_i\circ \BB_\tau\quad \quad \quad \quad \quad \quad \quad \quad \ \text{(Dirichlet)}\\
        & \int_0^\infty e^{-\kappa^2t-\int_0^tc(\BB_s)dLs}e_i\circ\BB_t dL_t\quad \text{(Robin)} 
    \end{cases}
\end{equation*}
 can be viewed as  random functions with $\BB_0=\BFx$. To construct our estimator $\hat{k}_{Q,m}$, we have one realizations $G_{i}(\Bu_j)$ for each $\Bu_j$. 

Let $\CalU_*$ denote the following event: 
\[\CalU_*=\left\{\lambda_{\rm min}(\frac{1}{m}\sum_{i=1}M_i)\geq \frac{\underline{\lambda}}{2} \quad \textrm{and} \quad \lambda_{\rm max}(M_i)\leq \frac{3\overline{\lambda}}{2}\right\},\]
where, following the notation from Lemma \ref{lem:spectral_empirical_gram},  $M_i=\BPhi(\Bu_i)\BPhi(\Bu_i)^T$ denotes the sample Gram matrix, and $\underline{\lambda}=\lambda_{\min}(M)$ and $\overline{\lambda}=\lambda_{\max}(M)$ denote the minimum and maximum eigenvalues of the Gram matrix $[M]_{i,j}=\langle\phi_i,\phi_j\rangle_{L^2}$, respectively.

From Lemma \ref{lem:spectral_empirical_gram}, it follows that the probability of $\CalU_*$ is at least:
\[\pr\left(\CalU_*\right)\geq 1-p(\sqrt{2}e^{-1/2})^{m\underline{\lambda}/S_\phi}-p(\frac{e^{1/2}}{(3/2)^{3/2}})^{m\overline{\lambda}/S_\phi}.\]

We first use \eqref{eq_pf_general_convergence_1} and \eqref{eq_pf_general_convergence_2} to have the following decomposition:
\begin{align}
    &\hat{k}_{m,Q}(\BFx,\BFx')- k_{\nu,\mathcal{B}}(\BFx,\BFx')=\CalK_1(\BFx,\BFx')-\CalK_2(\BFx,\BFx')-\CalK_2(\BFx',\BFx)+\CalK_3(\BFx,\BFx'),\label{eq_pf_general_convergence_4}
\end{align}   
where:
\begin{align*}
    &\CalK_1(\BFx,\BFx')= \sum_i\lambda_i\left(\CalP_m[e_i](\Bx)\CalP_m[e_i](\BFx')-e_i(\BFx)e_i(\BFx')\right),\\
    &\CalK_2(\BFx,\BFx')=\sum_i\lambda_i\left(\CalP_m[e_i](\BFx)\CalP_m[G_i](\BFx')-e_i(\BFx)\E[G_i(\BFx')]\right),\\
     &\CalK_3(\BFx,\BFx')=\sum_i\lambda_i\left(\CalP_m[G_i](\BFx)\CalP_m[G_i](\BFx')-\E[G_i(\BFx)]\E [G_i(\Bx')]\right).
\end{align*}

Write $k=k_{\nu,\mathcal{B}}$ for notation simplicity. Substitute \eqref{eq_pf_general_convergence_4} into the $L^2$ error between $\hat{k}_{m,Q}$ and $k$, we have on the event $U_*$
\begin{align}
   \|\hat{k}_{m,Q}-k\|_{L^2(\CalX\times\CalX)}^2 /4 \leq& \underbrace{\|\CalK_1\|_{L^2(\CalX\times\CalX)}^2}_{A_1}+\underbrace{2\|\CalK_2\|_{L^2(\CalX\times\CalX)}^2}_{A_2}+\underbrace{\|\CalK_3\|_{L^2(\CalX\times\CalX)}^2}_{A_3} .\label{eq_pf_general_convergence_5}
\end{align}

{\large \textbf{Estimation of $A_1$}}: $A_1$ is independent of boundary condition. On the set $\CalU_*$,
\begin{align}
\|\CalK_1\|_{L^2(\CalX\times\CalX)}^2=&\int_{\CalX\times\CalX}\bigg|\sum_i\lambda_i\big[\CalP_m[e_i](\Bx)\CalP_m[e_i](\BFx')-\CalP_m[e_i](\Bx)e_i(\Bx')\nonumber\\
    &+\CalP_m[e_i](\Bx)e_i(\Bx')-e_i(\BFx)e_i(\BFx')\big]\bigg|^2d\BFx d\Bx'\nonumber\\
    \leq &2\int_{\CalX\times\CalX}\bigg|\sum_i\lambda_i\CalP_m[e_i](\Bx)\big[\CalP_m[e_i](\BFx')-e_i(\Bx')\big]\bigg|^2d\Bx d\Bx'\nonumber\\
    &+2\int_{\CalX\times\CalX}\bigg|\sum_i\lambda_ie_i(\Bx')\big[\CalP_m[e_i](\BFx)-e_i(\Bx)\big]\bigg|^2d\Bx d\Bx'\nonumber\\
    \leq &2\underbrace{\int_\CalX\sum_i\lambda_i\big|\CalP_m[e_i](\BFx)\big|^2d\BFx }_{A_{11}}\underbrace{\int_\CalX\sum_i\lambda_i\big|\CalP_m[e_i](\BFx')-e_i(\Bx')\big|^2d\BFx'}_{A_{12}}\nonumber\\
    &+2\underbrace{\int_\CalX\sum_i\lambda_i\big|e_i(\BFx)\big|^2d\BFx }_{A_{13}}\underbrace{\int_\CalX\sum_i\lambda_i\big|\CalP_m[e_i](\BFx)-e_i(\Bx)\big|^2d\BFx }_{A_{12}}.\label{eq_pf_general_convergence_6}
\end{align}
For the term $A_{11}$, transform it back to the kernel representation, we can obtain:
\begin{align*}
    A_{11}=&\int_\CalX \BPhi(\BFx)^T\left(\frac{1}{m}\sum_iM_i\right)^{-1}\frac{\BPhi(\BU)}{m}k_\nu(\BU,\BU)\frac{\BPhi(\BU)^T}{m}\left(\frac{1}{m}\sum_{i=1}M_i\right)^{-1}\BPhi(\Bx)d\BFx\\
    =&\text{Tr}\left[\frac{1}{m}{\BPhi(\BU)^T}\left(\frac{1}{m}\sum_{i=1}M_i\right)^{-1}M\left(\frac{1}{m}\sum_iM_i\right)^{-1}{\BPhi(\BU)}\frac{k_\nu(\BU,\BU)}{m}\right]\\
    \leq & \text{Tr}[\frac{k_\nu(\BU,\BU)}{m}]\lambda_{max}\left(\frac{1}{m}{\BPhi(\BU)^T}\left(\frac{1}{m}\sum_{i=1}M_i\right)^{-1}M\left(\frac{1}{m}\sum_iM_i\right)^{-1}{\BPhi(\BU)}\right)\\
    \leq & C_1 \frac{\overline{\lambda}^2}{\underline{\lambda}^2},
\end{align*}
where the last line is because on event $U_*$, $\lambda_{\rm min}(\frac{1}{m}\sum_{i=1}M_i)\geq \frac{\underline{\lambda}}{2}$,  $ \lambda_{\rm max}(M_i)\geq \frac{3\overline{\lambda}}{2}$, and $\text{Tr}[\frac{k_\nu(\BU,\BU)}{m}]=k(0)$ for the Mat\'ern kernel.

For the term $A_{12}$, notice that the empirical projection $\CalP_m$ is invariant under the $L^2$ projection $\CalP$, i.e., $\CalP_m\CalP=\CalP$. We can obtain on $\CalU_*$:
\begin{align*}
   A_{12}\leq &\int_\CalX\sum_{i}\lambda_i\left(2|\CalP_m[e_i-\CalP e_i](\BFx)|^2+2|\CalP[e_i](\BFx)-e_i(\BFx)|^2\right)d\BFx\\
    =&\sum_i2\lambda_i \int_\CalX \|\BPhi^T(\BFx)(\frac{1}{m}\sum_{l=1}M_l)^{-1}\frac{1}{m}\BPhi(\BU)[e_i(\BU)-\CalP[e_i](\BU)]\|^2d\Bx\\
    &+2\sum_i\lambda_i\int_\CalX \left|\CalP[e_i](\BFx)-e_i(\BFx)\right|^2d\BFx\\
    \leq & \sum_iC_2\lambda_i \frac{\overline{\lambda}}{\underline{\lambda}^2}\|\frac{1}{m}\BPhi(\BU)[e_i(\BU)-\CalP[e_i](\BU)]\|^2+2\sum_i\lambda_i\int_\CalX \left|\CalP[e_i](\BFx)-e_i(\BFx)\right|^2d\BFx.
\end{align*}
 On $\CalU_*$, we have:
\begin{align*}
  \|\frac{1}{m}\BPhi(\BU)[e_i(\BU)-\CalP[e_i](\BU)]\|^2= & \sum_{j=1}^p\left(\frac{1}{m}\sum_{l=1}\phi_j(\Bu_l)\left[e_i(\Bu_l)-\CalP[e_i](\Bu_l)\right]\right)^2\\
  \leq & \|e_i-\CalP e_i\|_n^2\frac{1}{m}\sum_{j=1}^p\sum_{l=1}\phi_j(\Bu_l)^2\\
  \leq & \frac{S_\phi}{m}\left(\frac{1}{m}\sum_{j=1}\left|e_i(\Bu_j)-\CalP e_i(\Bu_j)\right|^2\right),
\end{align*}
where the second line is from H\"older inequality. Also, notice that the infinite sum of $|e_i(\BFx)-\CalP[e_i](\BFx)|^2$ can be represented as operators acting on kernel $k$:
\begin{align*}
    \sum_i\lambda_i |e_i(\BFx)-\CalP[e_i](\BFx)|^2=& \sum_i\lambda_i \left(\bold{I}-\CalP\right)e_ie_i\left(\bold{I}-\CalP\right)^T(\BFx,\BFx)\\
    =&\left(\bold{I}-\CalP\right)\left( \sum_i\lambda_i e_ie_i\right)\left(\bold{I}-\CalP\right)^T(\BFx,\BFx)\\
    =&\left(\bold{I}-\CalP\right)k_\nu(\cdot,\cdot)\left(\bold{I}-\CalP\right)^T(\BFx,\BFx).
\end{align*}
So $A_{11}\cdot A_{12}$ is bounded as:
\begin{align*}
    & \int_\CalX\sum_i\lambda_i\big|\CalP_m[e_i](\BFx)\big|^2d\BFx \int_\CalX\sum_i\lambda_i\big|\CalP_m[e_i](\BFx')-e_i(\Bx')\big|^2d\BFx' \\
    \leq&C_3(\frac{\overline{\lambda}}{\underline{\lambda}})^2\bigg[\frac{\overline{\lambda}S_\phi}{\underline{\lambda}^2m}\left(\frac{1}{m}\sum_{j=1}\left(\bold{I}-\CalP\right)k_\nu(\cdot,\cdot)\left(\bold{I}-\CalP\right)^T(\Bu_j,\Bu_j)\right)\\
    &+\int_\CalX\left(\bold{I}-\CalP\right)k_\nu(\cdot,\cdot)\left(\bold{I}-\CalP\right)^T(\BFx,\BFx)d\Bx\bigg].
\end{align*}

For $A_{13}$, it is straightforward to check that it is:
\[\int_\CalX d\BFx k(0)=k(0)\text{vol}(\CalX).\]
Substitute the above identities of $A_{11}$, $A_{12}$, and $A_{13}$ into \eqref{eq_pf_general_convergence_6}, we then have:
\begin{align}
\|\CalK_1\|_{L^2(\CalX\times\CalX)}^2\leq&C_4(\frac{\overline{\lambda}}{\underline{\lambda}})^2\bigg[\frac{\overline{\lambda}S_\phi}{\underline{\lambda}^2m}\left(\frac{1}{m}\sum_{j=1}\left(\bold{I}-\CalP\right)k_\nu(\cdot,\cdot)\left(\bold{I}-\CalP\right)^T(\Bu_j,\Bu_j)\right)\label{eq_pf_general_convergence_7}\\
    &+\int_\CalX\left(\bold{I}-\CalP\right)k_\nu(\cdot,\cdot)\left(\bold{I}-\CalP\right)^T(\BFx,\BFx)d\Bx\bigg].\nonumber
\end{align}

{\large \textbf{Estimation of $A_2$}}: For $A_2$, first decompose $\CalK_2$ as follows:
\begin{align*}
    \CalK_2(\Bx,\Bx')=\CalK_{21}(\Bx,\Bx')+\CalK_{22}(\Bx,\Bx')+\CalK_{23}(\Bx,\Bx'),
\end{align*}
where:
\begin{align*}
    &\CalK_{21}(\Bx,\Bx')=\sum_i\lambda_i\left(\CalP_m[e_i](\Bx)\CalP_m[G_i](\BFx')-\CalP_m[e_i](\BFx)\CalP_m\E[G_i](\BFx')\right)\\
    &\CalK_{22}(\Bx,\Bx')=\sum_i\lambda_i\left(\CalP_m[e_i](\BFx)\CalP_m\E[G_i](\BFx')-\CalP_m[e_i](\BFx)\E[G_i(\BFx')]\right)\\
    &\CalK_{23}(\Bx,\Bx')=\sum_i\lambda_i\left(\CalP_m[e_i](\BFx)\E[G_i(\BFx')]-e_i(\BFx)\E[G_i(\BFx')]\right).
\end{align*}
So $A_2$ can be further decomposed as:
\begin{equation}
    A_2/4\leq \underbrace{\|\CalK_{21}\|_{L^2(\CalX\times\CalX)}^2}_{A_{21}}+\underbrace{\|\CalK_{22}\|_{L^2(\CalX\times\CalX)}^2}_{A_{22}}+\underbrace{\|\CalK_{23}\|_{L^2(\CalX\times\CalX)}^2}_{A_{23}}. \label{eq_pf_general_convergence_8}
\end{equation}

For $A_{21}$:
\begin{align*}
    \|\CalK_{21}\|_{L^2(\CalX\times\CalX)}^2\leq &\int_{\CalX}\sum_i\lambda_i\left|\CalP_m[e_i](\BFx)\right|^2d\BFx\int_{\CalX}\sum_i\lambda_i\left|\CalP_m[ \E[G_i]-G_i](\BFx)\right|^2d\BFx\\
    = & A_{11} \int_{\CalX} \sum_i\lambda_i \|\BPhi(\Bx)^T(\frac{1}{m}\sum_{i=1}M_i)^{-1}\BPhi(\BU)\boldsymbol{\varepsilon}_i\|^2d\Bx,
\end{align*}
where $\boldsymbol{\varepsilon}$ are independent distributed zero-mean observation noise $\boldsymbol{\varepsilon}_{ij}=\E_{G_i}[G_i(\Bu_j)]-G_i(\Bu_j)$. Analogously to the estimation of $A_{12}$, 
\begin{align*}
    \int_{\CalX} \sum_i\lambda_i \|\BPhi(\Bx)^T(\frac{1}{m}\sum_{i=1}M_i)^{-1}\BPhi(\BU)\boldsymbol{\varepsilon}_i\|^2d\Bx\leq C_2\frac{\overline{\lambda}}{\underline{\lambda}^2}\sum_i\lambda_i\|\frac{1}{m}\BPhi(\BU)\boldsymbol{\varepsilon}_i\|^2.
\end{align*}
Because noise $\boldsymbol{\varepsilon}$ are zero-mean, we have:
\begin{align*}
    \sum_i\lambda_i\ \|\frac{1}{m}\BPhi(\BU)\boldsymbol{\varepsilon}_i\|^2 
    =& \sum_i \lambda_i \sum_{j=1}^p\left(\frac{1}{m}\sum_{l=1}\phi_j(\Bu_l)\boldsymbol{\varepsilon}_{il}\right)^2\\
    \leq & \sum_i\lambda_i \sum_{j=1}^p\left(\frac{1}{m}\sum_{l=1}\phi_j(\Bu_l)^2\right)\left(\frac{1}{m}\sum_{l=1}\boldsymbol{\varepsilon}_{il}^2\right)\\
    \leq & \frac{S_\phi}{m}\frac{1}{m}\sum_{l=1}\sum_i\lambda_i \boldsymbol{\varepsilon}_{il}^2.
\end{align*}

For the Dirichlet boundary setting:
\begin{align*}
    \sum_i\lambda_i \boldsymbol{\varepsilon}_{il}^2 = &  \sum_i\lambda_i\left(\E[e_i(\BB_{\tau})e^{-\kappa^2\tau}|\BB_0=\Bu_l]- e_i(\BB_{\tau})e^{-\kappa^2\tau}|\BB_0=\Bu_l\right)^2\\
    =&  \E_{\BB_t,\BB'_t,\tau,\gamma}[k_\nu(\BB_\tau,\BB'_\gamma)e^{-\kappa^2(\tau+\gamma)}|\BB_0=\BB'_0=\Bu_l]\\
    &+k_\nu(\BB_\tau,\BB_\tau)e^{-2\kappa\tau}|\BB_0=\Bu_l\\
    &-2\E_{\BB_t,\tau}[k_\nu(\BB_\tau,\BB'_\gamma)e^{-\kappa^2(\tau+\gamma)}|\BB_0=\BB'_0=\Bu_l]
    \leq 4k(0),
\end{align*}
where the last line is from the fact that Mat\'ern kernel satisfies $k_\nu(0)\geq k_\nu(\Bx,\Bx')$ for any $\Bx,\Bx'$.

For the Robin boundary setting:
\begin{align*}
   \sum_i\lambda_i \boldsymbol{\varepsilon}_{il}^2 = & \E \int_0^\infty \int_0^\infty e^{-\kappa^2(t+\tau)-\int_0^tc(\BB_s)dL_s-\int_0^\tau c(\BB'_s)dL'_s}k_\nu(\BB_t,\BB'_\tau)dL_tdL_\tau\\
   &+\int_0^\infty \int_0^\infty e^{-\kappa^2(t+\tau)-\int_0^tc(\BB_s)dL_s-\int_0^\tau c(\BB'_s)dL'_s}k_\nu(\BB_t,\BB'_\tau)dL_tdL_\tau\\
    &-2\E_{\BB_t,L_t}\int_0^\infty \int_0^\infty e^{-\kappa^2(t+\tau)-\int_0^tc(\BB_s)dL_s-\int_0^\tau c(\BB'_s)dL'_s}k_\nu(\BB_t,\BB'_\tau)dL_tdL_\tau\\
    &\leq 4k_\nu(0),
\end{align*}
where the last line is from \eqref{eq:pf_empirical_BdryMat_robin_1}.

Therefore, in both cases,  we always have:
\begin{align*}
     \sum_i\lambda_i\ \|\frac{1}{m}\BPhi(\BU)\boldsymbol{\varepsilon}_i\|^2
    \leq &\frac{4k_\nu(0)S_\phi}{m}.
\end{align*}

So we have the following estimation of $A_{21}$
\[A_{21}\leq A_{11}\frac{S_\phi\overline{\lambda}}{m\underline{\lambda}^2}4k_\nu(0)\leq C_5\frac{S_\phi\overline{\lambda}^3}{m\underline{\lambda}^4}.\]

For $A_{22}$, 
\begin{align*}
    \|\CalK_{22}\|^2_{L^2(\CalX\times \CalX)}\leq &\int_\CalX\sum_i\lambda_i\left|\CalP_m[e_i](\Bx)\right|^2d\BFx\int_\CalX\sum_i\lambda_i\left|\CalP_m\E[G_i](\BFx)-\E[G_i](\BFx)\right|^2d\BFx.
\end{align*}
So we can notice that its estimation is almost the same as that of $A_{11}\cdot A_{12} $. We only need to replace the eigenfunctions $\{e_i\}$ by the expectation  $\{\E G_i(\cdot)\}$ and calculate the infinite sum of $\left|\CalP\E[G_i](\BFx)-\E[G_i](\BFx)\right|^2$ as follows:
\begin{align*}
    &\sum_i\lambda_i \left|\CalP\E[G_i](\BFx)-\E[G_i](\BFx)\right|^2=\sum_i\lambda_i(\bold{I}-\CalP)\E[G_i]\E[G_i](\bold{I}-\CalP)^T(\BFx,\BFx)\\
     \coloneqq &(\bold{I}-\CalP)\bar{k}(\cdot,\cdot)(\bold{I}-\CalP)^T(\BFx,\BFx),
     \end{align*}
     where:
\begin{align*}     
    &\bar{k}(\cdot,\cdot)\\
    =& \begin{cases}
        &\E\left[\sum_i\lambda_ie_i(\BB_\tau)e_i(\BB_\gamma')e^{-\kappa^2(\tau+\gamma)}\bigg|\BB_0=\cdot,\BB'_0=\cdot\right]\quad\text{(Dirichlet)}\\
        & \E\left[\sum_i\lambda_i\int_{\Real_+^2}e^{-\kappa^2(t+\tau)-\int_0^tc(\BB_s)dL_s-\int_0^\tau c(\BB'_s)dL'_s}e_i(\BB_t)e_i(\BB'_\tau)dL_tdL'_\tau\bigg|\BB_0=\cdot,\BB'_0=\cdot\right]
    \end{cases}\\
    & {\quad\quad\quad\quad\quad\quad\quad\quad\quad\quad\quad\quad\quad\quad\quad\quad\quad\quad\quad\quad\quad\quad\quad\quad\quad\quad\quad\quad\quad\quad\quad\quad\quad \  \text{(Robin)}}\\
    =&\begin{cases}
        &\E\left[k_\nu(\BB_\tau,\BB'\gamma)\bigg|\BB_0=\cdot,\BB'_0=\cdot\right]\quad\text{(Dirichlet)}\\
        &\E\left[\int_{\Real_+^2}e^{-\kappa^2(t+\tau)-\int_0^tc(\BB_s)dL_s-\int_0^\tau c(\BB'_s)dL'_s}k_\nu(\BB_t,\BB'_\tau)dL_tdL'_\tau\bigg|\BB_0=\cdot,\BB'_0=\cdot\right]\  \text{(Robin)}
    \end{cases}
\end{align*}
So $A_{22}$ is bounded as:
\begin{align*}
 A_{22}
    \leq&C_3(\frac{\overline{\lambda}}{\underline{\lambda}})^2\bigg[\frac{\overline{\lambda}S_\phi}{\underline{\lambda}^2m}\left(\frac{1}{m}\sum_{j=1}\left(\bold{I}-\CalP\right)\bar{k}(\cdot,\cdot)\left(\bold{I}-\CalP\right)^T(\Bu_j,\Bu_j)\right)\\
    &+\int_\CalX\left(\bold{I}-\CalP\right) \bar{k}(\cdot,\cdot)\left(\bold{I}-\CalP\right)^T(\BFx,\BFx)d\Bx\bigg].
\end{align*}

For $A_{23}$,
\begin{align*}
    &\|\CalK_{23}\|^2_{L^2(\CalX\times\CalX)}\\ \leq &\int_\CalX \sum_i\lambda_i\left|\E[G_i(\BFx)]\right|^2d \BFx \int_\CalX \sum_i\lambda_i\left|e_i(\BFx)-\CalP_m[e_i](\BFx)\right|^2d\BFx\\
    =& \int_\CalX \sum_i\lambda_i\left|\E[G_i(\BFx)]\right|^2d \BFx  A_{12}.\\
    \end{align*}
From our previous calculations, the infinite sum is bounded as follows for both the Dirichlet and Robin boundary settings:
\begin{align*}
    & \sum_i\lambda_i\left|\E[G_i(\BFx)]\right|^2 \\
    = &\begin{cases}
        &\E\left[k_\nu(\BB_\tau,\BB_\gamma')e^{-\kappa^2(\tau+\gamma)}\bigg|\BB_0=\Bx,\BB'_0=\Bx\right]\\
        &\E\left[\int_{\Real_+^2}e^{-\kappa^2(t+\tau)-\int_0^tc(\BB_s)dL_s-\int_0^\tau c(\BB'_s)dL'_s}k_\nu(\BB_t,\BB'_\tau)dL_tdL'_\tau\bigg|\BB_0=\BFx,\BB'_0=\BFx\right]
    \end{cases}\\
    \leq &k_\nu(0).
\end{align*}
So:
\begin{align*}
    \|\CalK_{23}\|^2_{L^2(\CalX\times\CalX)}\leq & \text{vol}(\CalX)k_\nu(0) A_{12}\\
    \leq& C_6 \bigg[\frac{\overline{\lambda}S_\phi}{\underline{\lambda}^2m}\left(\frac{1}{m}\sum_{j=1}\left(\bold{I}-\CalP\right)k_\nu(\cdot,\cdot)\left(\bold{I}-\CalP\right)^T(\Bu_j,\Bu_j)\right)\\
    &+\int_\CalX\left(\bold{I}-\CalP\right)k_\nu(\cdot,\cdot)\left(\bold{I}-\CalP\right)^T(\BFx,\BFx)d\Bx\bigg],
\end{align*}
where the last line is from our precious calculations for $A_{12}$.

Substitute the identities of $A_{21}$, $A_{22}$, and $A_{23}$ into \eqref{eq_pf_general_convergence_8}, we then have:
\begin{align}
    &\|\CalK_2\|^2_{L^2(\CalX\times\CalX)}\label{eq_pf_general_convergence_9}\\
    \leq & C_7\bigg[\frac{S_\phi\overline{\lambda}^3}{m\underline{\lambda}^4}\nonumber\\
    &+(\frac{\overline{\lambda}}{\underline{\lambda}})^2\left(\frac{\overline{\lambda}S_\phi}{\underline{\lambda}^2m}\left[\frac{1}{m}\sum_{l=1}\left(\bold{I}-\CalP\right)\bar{k}\left(\bold{I}-\CalP\right)^T(\Bu_l,\Bu_l)\right]+\int_\CalX\left(\bold{I}-\CalP\right)\bar{k}\left(\bold{I}-\CalP\right)^T(\BFx,\BFx)d\Bx\right)\nonumber\\
    +& \left(\frac{\overline{\lambda}S_\phi}{\underline{\lambda}^2m}\left[\frac{1}{m}\sum_{l=1}\left(\bold{I}-\CalP\right)k_\nu(\cdot,\cdot)\left(\bold{I}-\CalP\right)^T(\Bu_l,\Bu_l)\right]+\int_\CalX\left(\bold{I}-\CalP\right)k_\nu(\cdot,\cdot)\left(\bold{I}-\CalP\right)^T(\BFx,\BFx)d\Bx\right)\bigg]. \nonumber
\end{align}

{\large \textbf{Estimation of $A_3$}}: For $A_3$, first decompose $\CalK_3$ as follows:
\begin{align*}
    \CalK_3(\Bx,\Bx')=\CalK_{31}(\Bx,\Bx')+\CalK_{32}(\Bx,\Bx')+\CalK_{33}(\Bx,\Bx')+\CalK_{34}(\Bx,\Bx'),
\end{align*}
where:
\begin{align*}
    &\CalK_{31}(\Bx,\Bx')=\sum_i\lambda_i\left(\CalP_m[G_i](\Bx)\CalP_m[G_i](\BFx')-\CalP_m[G_i](\BFx)\CalP_m\E[G_i](\BFx')\right)\\
    &\CalK_{32}(\Bx,\Bx')=\sum_i\lambda_i\left(\CalP_m[G_i](\BFx)\CalP_m\E[G_i](\BFx')-\CalP_m\E[G_i](\BFx)\CalP_m\E[G_i](\Bx')\right)\\
    &\CalK_{33}(\Bx,\Bx')=\sum_i\lambda_i\left(\CalP_m\E[G_i](\BFx)\CalP_m\E[G_i](\BFx')-\E[G_i(\BFx)]\CalP_m\E[G_i](\BFx')\right)\\
    &\CalK_{34}(\Bx,\Bx')=\sum_i\lambda_i\left(\E[G_i(\BFx)]\CalP_m\E[G_i](\BFx')-\E[G_i(\BFx)]\E[G_i](\BFx')\right).
\end{align*}
So $A_3$ can be further decomposed as:
\begin{align}
    A_3/8\leq &\underbrace{ \|\CalK_{31}\|_{L^2(\CalX\times\CalX)}^2 }_{A_{31}}+\underbrace{ \|\CalK_{32}\|_{L^2(\CalX\times\CalX)}^2 }_{A_{32}} +\underbrace{ \|\CalK_{33}\|_{L^2(\CalX\times\CalX)}^2 }_{A_{33}}+\underbrace{ \|\CalK_{34}\|_{L^2(\CalX\times\CalX)}^2 }_{A_{34}}.\label{eq_pf_general_convergence_10}
\end{align}
For $A_{31}$:
\begin{align*}
    \|\CalK_{31}\|_{L^2(\CalX\times\CalX)}^2\leq &\int_{\CalX}\sum_i\lambda_i\left|\CalP_m[G_i](\BFx)\right|^2d\BFx\int_{\CalX}\sum_i\lambda_i\left|\CalP_m[ \E[G_i]-G_i](\BFx)\right|^2d\BFx.\\
\end{align*}
Notice that the second multiplier on the right hand side has been estimated in the calculations for $A_{21}$:
\begin{equation}
    \int_{\CalX}\sum_i\lambda_i\left|\CalP_m[ \E[G_i]-G_i](\BFx)\right|^2d\BFx \leq C_8\frac{S_\phi\overline{\lambda}}{m\underline{\lambda}^2}. \label{eq_pf_general_convergence_11}
\end{equation}
We only need to estimate the first multiplier:
\begin{align*}
   &\int_{\CalX}\sum_i\lambda_i\left|\CalP_m[G_i](\BFx)\right|^2d\BFx\\
   =&\int_\CalX \BPhi(\BFx)^T\left(\frac{1}{m}\sum_iM_i\right)^{-1}\frac{\BPhi(\BU)}{m}\bold{H}\frac{\BPhi(\BU)^T}{m}\left(\frac{1}{m}\sum_{i=1}M_i\right)^{-1}\BPhi(\Bx)d\BFx\\
    =&\text{Tr}\left[\frac{1}{m}{\BPhi(\BU)^T}\left(\frac{1}{m}\sum_{i=1}M_i\right)^{-1}M\left(\frac{1}{m}\sum_iM_i\right)^{-1}{\BPhi(\BU)}\frac{\bold{H}}{m}\right]\\
    \leq & \text{Tr}\left[\frac{\bold{H}}{m}\right]\lambda_{max}\left(\frac{1}{m}{\BPhi(\BU)^T}\left(\frac{1}{m}\sum_{i=1}M_i\right)^{-1}M\left(\frac{1}{m}\sum_iM_i\right)^{-1}{\BPhi(\BU)}\right)\\
    \leq & \text{Tr}\left[\frac{\bold{H}}{m}\right] C_1\left(\frac{\overline{\lambda}}{\underline{\lambda}}\right)^2,
\end{align*}
where for the Dirichlet boundary setting:
\begin{align*}
    \bold{H}_{ij}=&\sum_i\lambda_ie_i(\BB^{(i)}_{\tau_i})e_i(\BB^{(j)}_{\tau_j})e^{-\kappa^2(\tau_i+\tau_j)}\bigg|\BB^{(i)}_0=\Bu_i,\BB^{(j)}_0=\Bu_j\\
    =& k_\nu(\BB^{(i)}_{\tau_i},\BB^{(j)}_{\tau_j})e^{-\kappa^2(\tau_i+\tau_j)}\bigg|\BB^{(i)}_0=\Bu_i,\BB^{(j)}_0=\Bu_j\leq k(0).
\end{align*}
Similarly, for the Robin boundary setting:
\begin{align*}
    \bold{H}_{ij}
    =&\int_{\Real_+^2}e^{-\kappa^2(t+\tau)-\int_0^tc(\BB_s)dL_s-\int_0^\tau c(\BB'_s)dL'_s}k_\nu(\BB_t,\BB'_\tau)dL_tdL'_\tau\bigg|\BB_0=\Bu_i,\BB'_0=\Bu_j\leq k(0).
\end{align*}

In both cases, we have $\text{Tr}\left[\frac{\bold{H}}{m}\right]\leq k(0)$ so we can obtain:
\begin{align}
   &\int_{\CalX}\sum_i\lambda_i\left|\CalP_m[G_i](\BFx)\right|^2d\BFx
    \leq  C_1\left(\frac{\overline{\lambda}}{\underline{\lambda}}\right)^2.\label{eq_pf_general_convergence_12}
\end{align}
So $A_{31}$ is bounded as:
\[A_{31}\leq C_9\frac{S_\phi\overline{\lambda}^3}{m\underline{\lambda}^4}.\]

For $A_{32}$,
\begin{align*}
   &\|\CalK_{32}\|^2_{L^2(\CalX\times\CalX)}\\ \leq&\int_\CalX\sum_i\lambda_i\left|\CalP_m\E[G_i](\BFx)\right|^2d\BFx\int_\CalX\sum_i\lambda_i\left(\CalP_m[G_i](\BFx)-\CalP_m\E[G_i](\BFx)\cal\right)^2d\Bx \\
  \leq&\E_{\{G_i\}}\left[\int_\CalX\sum_i\lambda_i\left|\CalP_m[G_i](\BFx)\right|^2d\BFx\right]\int_\CalX\sum_i\lambda_i\left(\CalP_m[G_i](\BFx)-\CalP_m\E[G_i](\BFx)\cal\right)^2d\Bx \\
  \leq & C_1\left(\frac{\overline{\lambda}}{\underline{\lambda}}\right)^2\int_\CalX\sum_i\lambda_i\left(\CalP_m[G_i](\BFx)-\CalP_m\E[G_i](\BFx)\cal\right)^2d\Bx \\
  \leq & A_{31},
\end{align*}
where the fourth line is from \eqref{eq_pf_general_convergence_12} and the last line is from \eqref{eq_pf_general_convergence_11}.

For $A_{33}$:
\begin{align*}
     &\|\CalK_{33}\|^2_{L^2(\CalX\times\CalX)}\\ 
     \leq &    \int_\CalX\sum_i\lambda_i\left|\CalP_m\E[G_i](\Bx)\right|^2d\Bx \int_\CalX\sum_i\lambda_i\left(\CalP_m\E[G_i](\Bx)-\E[G_i](\Bx)\right)^2d\Bx  \\
     \leq &  C_1\left(\frac{\overline{\lambda}}{\underline{\lambda}}\right)^2 \int_\CalX\sum_i\lambda_i\left(\CalP_m\E[G_i](\Bx)-\E[G_i](\Bx)\right)^2d\Bx  \\
     \leq & C_1 \left(\frac{\overline{\lambda}}{\underline{\lambda}}\right)^2 \left(\frac{\overline{\lambda}S_\phi}{\underline{\lambda}^2m}\left[\frac{1}{m}\sum_{l=1}\left(\bold{I}-\CalP\right)\bar{k}\left(\bold{I}-\CalP\right)^T(\Bu_l,\Bu_l)\right]+\int_\CalX\left(\bold{I}-\CalP\right)\bar{k}\left(\bold{I}-\CalP\right)^T(\BFx,\BFx)d\Bx\right),
\end{align*}
where the last line is from the estimation of $A_{22}$.

For $A_{44}$, calculations similar to what we did for $A_{22}$ and $A_{23}$ gives:
\begin{align*}
    &\|\CalK_{33}\|^2_{L^2(\CalX\times\CalX)}\\ 
     \leq &  \int_{\CalX}\sum_i\lambda_i \left|\E [G_i(\BFx)]\right|^2d\BFx \int_\CalX\sum_i\lambda_i\left(\CalP_m\E[G_i](\Bx)-\E[G_i](\Bx)\right)^2d\Bx \\
     \leq & k_\nu(0)\text{vol}(\CalX)\left(\frac{\overline{\lambda}S_\phi}{\underline{\lambda}^2m}\left[\frac{1}{m}\sum_{l=1}\left(\bold{I}-\CalP\right)\bar{k}\left(\bold{I}-\CalP\right)^T(\Bu_l,\Bu_l)\right]+\int_\CalX\left(\bold{I}-\CalP\right)\bar{k}\left(\bold{I}-\CalP\right)^T(\BFx,\BFx)d\Bx\right).
\end{align*}
Substitute the identities of $A_{31}$, $A_{32}$, $A_{33}$ and $A_{34}$ into \eqref{eq_pf_general_convergence_10}, we then have:
\begin{align}\label{eq_pf_general_convergence_13}
    &\|\CalK_3\|^2_{L^2(\CalX\times\CalX)}\leq  C_{10}\bigg[\frac{S_\phi\overline{\lambda}^3}{m\underline{\lambda}^4}+(\frac{\overline{\lambda}}{\underline{\lambda}})^2\big(\frac{\overline{\lambda}S_\phi}{\underline{\lambda}^2m}\left[\frac{1}{m}\sum_{l=1}\left(\bold{I}-\CalP\right)\bar{k}\left(\bold{I}-\CalP\right)^T(\Bu_l,\Bu_l)\right]\\
    &+\int_\CalX\left(\bold{I}-\CalP\right)\bar{k}\left(\bold{I}-\CalP\right)^T(\BFx,\BFx)d\Bx\big)\bigg]\nonumber.
\end{align}
Substitute   \eqref{eq_pf_general_convergence_7}, \eqref{eq_pf_general_convergence_9}, and \eqref{eq_pf_general_convergence_13} into \eqref{eq_pf_general_convergence_5}, we have:
\begin{align}
   &\|\hat{k}_{m,Q}-k\|_{L^2(\CalX\times\CalX)}^2\leq C\left(\frac{\overline{\lambda}}{\underline{\lambda}}\right)^2\bigg[\frac{\overline{\lambda}S_\phi}{\underline{\lambda}^2m}\left(1+\frac{1}{m}\sum_{l=1} (\bold{I}-\CalP)[k_\nu+\bar{k}](\bold{I}-\CalP)^T(\Bu_l,\Bu_l) \right)\nonumber\\
    &+\int_\CalX(\bold{I}-\CalP)[k_\nu+\bar{k}](\bold{I}-\CalP)^T(\Bx,\Bx)d\Bx\bigg].
\end{align}
By applying Chebyshev's inequality on the averaged $\frac{1}{m}\sum_{l=1} (\bold{I}-\CalP)[k_\nu+\bar{k}](\bold{I}-\CalP)^T(\Bu_l,\Bu_l)$, together with the probability of the event $\CalU_*$, we have the following bound with probability at least $1-v_*-p(\sqrt{2}e^{-1/2})^{m\underline{\lambda}/S_\phi}-p(\frac{e^{1/2}}{(3/2)^{3/2}})^{m\overline{\lambda}/S_\phi}$:
\begin{equation*}
    \|\hat{k}_{m,Q}-k\|_{L^2(\CalX\times\CalX)}^2\leq C \left(\frac{\overline{\lambda}}{\underline{\lambda}}\right)^2 \left[\left(\frac{\overline{\lambda}S_\phi}{\underline{\lambda}^2m}+1\right)\int_\CalX(\bold{I}-\CalP)[k_\nu+\bar{k}](\bold{I}-\CalP)^T(\Bx,\Bx)d\Bx+\frac{\overline{\lambda}S_\phi}{\underline{\lambda}^2m}\right],
\end{equation*}
where $v_*=\int_\CalX\left|(\bold{I}-\CalP)[k+\bar{k}](\bold{I}-\CalP)^T(\Bx,\Bx)\right|^2d\Bx$. In our analysis, the statistical error $\frac{\overline{\lambda}S_\phi}{\underline{\lambda}^2m}=o(1)$, otherwise the overall error $ \|\hat{k}_{m,Q}-k\|_{L^2(\CalX\times\CalX)}^2$ blows up. We thus have the final result \eqref{eq:convergence_general}.
\endproof

\section{Proof of Theorem \ref{thm:b_spline_regression}}
\label{pf:b_spline_regression}

We first present the  following lemma on important properties of the Mat\'ern kernel. These are well-known results and can be found in standard references (e.g., \cite[Sec 1.2.1]{yang2022graph}). 
\begin{lemma}
    \label{lem:eigen_mat}
    The Mat\'ern kernel $k_\nu$ has the following eigenfunction decomposition on the compact domain $\CalX$:
    \begin{equation}
        \label{eq:eign_mat}
        k_\nu(\BFx,\BFx')=\sum_i\lambda_ie_i(\BFx)e_i(\BFx'),\quad \Bx,\Bx'\in\CalX.
    \end{equation}
    Here, the eigenvalues $\lambda_i=\CalO(i^{-(2\nu/d+1)})$, the eigenfunctions $\{e_i\}$ are mutually orthogonal under the $L^2$ inner product $\langle\cdot,\cdot\rangle_2$ (see below), and the reproducing kernel Hilbert space (RKHS) inner product $\langle\cdot,\cdot\rangle_{k_\nu}$ induced by the Mat\'ern kernel $k_\nu$ is given by:
    \begin{align*}
        &\langle e_i,e_j\rangle_{L^2}\coloneqq\int_\CalX e_i(\Bs)e_j(\Bs)d\Bs=\delta_{ij},\ \langle e_i,e_j\rangle_{k_\nu}=\lambda_i^{-1}\delta_{ij},
    \end{align*}
   where $\delta_{ij}$ is the Kronecker delta. Moreover, the RKHS induced by $k_\nu$ is equivalent to the Sobolev space with smoothness order  $\alpha=\nu+d/2$.
\end{lemma}
\noindent \cite{Wendland10,paulsen2016introduction} provides further details on this RKHS connection.

The following Lemma is a Sobolev inequality for the RKHS induced by Mat\'ern-$\nu$ kernel, which is also a Sobolev space with fractional order of smoothness $\alpha=\nu+d/2$.
\begin{lemma}
    \label{lem:RKHS_interpolation inequality}
    For any function $f$ in the RKHS induced by Mat\'ern-$\nu$ kernel $k_\nu$ on $\CalX$, we have:
    \[\|f\|_{L^4(\CalX)}\leq C\|f\|_{\CalH_\nu(\CalX)}^{\frac{d}{4\alpha}}\|f\|_{L^2(\CalX)}^{1-\frac{d}{4\alpha}},\quad \alpha=\nu+d/2,\]
    where $C$ is some constant independent of $f$.
\end{lemma}
\proof
Let $\alpha^*=\lfloor\alpha\rfloor$ be the largest integer no greater than $\alpha$. We can use the Gagliardo-Nirenberg inequality (Sobolev inequality for Sobolev space with integer order of smoothness, see \cite{nirenberg1959elliptic,gagliardo1959ulteriori}) to obtain:
\[\|f\|_{L^4(\CalX)}\leq C_1  \|f\|_{\CalH^{\alpha^*}(\CalX)}^{\frac{d}{4\alpha^*}}\|f\|_{L^2(\CalX)}^{1-\frac{d}{4\alpha^*}},\]
where  $\|f\|_{\CalH^{\alpha^*}(\CalX)}$ is the Sobolev norm of $f$ with smoothness order $\alpha^*$. It  can be defined via the following RKHS defined on the whole space $\Real^d$:
\begin{align*}
    &\|f\|_{\CalH^{\alpha^*}(\CalX)}=\min_{h\in\CalH^{\alpha^*}(\Real^d)}\|h\|_{\CalH^{\alpha^*}(\Real^d)},\quad h=f,\quad \text{on}\ \CalX\\
    \text{where}\ &\|h\|_{\CalH^{\alpha^*}(\Real^d)}=\int_{\Real^d}\left|\CalF[h](\Bomega)\right|^2(\kappa+\|\Bomega\|^2)^{\alpha^*} d\Bomega.
\end{align*}
Let $h^*=\arg \min_{h\in\CalH^{\alpha^*}(\Real^d)}\|h\|_{\CalH^{\alpha^*}(\Real^d)}$ such that $h^*=f$ on $\CalX$. From the norm equivalence between RKHS and Sobolev spaces (see \cite{Wendland10,Evans15}), we have:
\begin{align*}
    &C_2\|f\|_{{\CalH^{\alpha^*}(\CalX)}}\leq  \|h^*\|_{\CalH^{\alpha^*}(\Real^d)} \leq C_3\|f\|_{{\CalH^{\alpha^*}(\CalX)}}\\
    &C_4\|f\|_{{L^2(\CalX)}}\leq  \|h^*\|_{L^2(\Real^d)} \leq C_5\|f\|_{{L^2(\CalX)}},
\end{align*}
where $C_2-C_5$ are some constants independent of $f$. Then the Gagliardo-Nirenberg inequality becomes:
\begin{equation}
    \|f\|_{L^4(\CalX)}\leq C_6  \|h^*\|_{\CalH^{\alpha^*}(\Real^d)}^{\frac{d}{4\alpha^*}}\|f\|_{L^2(\CalX)}^{1-\frac{d}{4\alpha^*}}.\label{eq:RKHS_interpolation_inequality_1}
\end{equation}
Let $p=\alpha/\alpha^*\geq 1$. The RKHS norm in the above equation can be written as:
\begin{align}
    \|h^*\|_{\CalH^{\alpha^*}(\Real^d)}^2=&\int_{\Real^d}(\kappa+\|\Bomega\|^2)^{\alpha^*}|\CalF[h^*](\Bomega)|^2 d\Bomega\nonumber\\
    =& \int_{\Real^d}(\kappa+\|\Bomega\|^2)^{\alpha^*}|\CalF[h^*](\Bomega)|^{2/p}|\CalF[h^*](\Bomega)|^{2-2/p} d\Bomega\nonumber\\
    \leq & \left(\int_{\Real^d}(\kappa+\|\Bomega\|^2)^{\alpha}||\CalF[h^*](\Bomega)|^{2} d\Bomega\right)^{\alpha^*/\alpha}\left(\int_{\Real^d}|\CalF[h^*](\Bomega)|^{2}d\Bomega\right)^{1-\alpha^*/\alpha}\nonumber\\
    \leq & C_7 \left(\int_{\Real^d}(\kappa+\|\Bomega\|^2)^{\alpha}||\CalF[h^*](\Bomega)|^{2} d\Bomega\right)^{\alpha^*/\alpha}\|f\|_{L^2(\CalX)}^{2-2\alpha^*/\alpha} \label{eq:RKHS_interpolation_inequality_2},
\end{align}
where the third line is from H\"older's inequality. Substituting \eqref{eq:RKHS_interpolation_inequality_2} into \eqref{eq:RKHS_interpolation_inequality_1} and applying the RKHS norm equivalence identities, we can have the final result.
\endproof

The following lemma is a restatement of  \cite[Theorem 1.2]{dung2011optimal} or \cite[Lemma 3.1]{dung2011adaptive} using the fact that the RKHS $\CalH_\nu(\CalX)$ induced by Mat\'ern-$\nu$ kernel $k_\nu$  on $\CalX$ is a Sobolev space with (non-integer) smoothness order $\alpha=\nu+d/2$.
\begin{lemma}
    \label{lem:spline_approx_err}
     Let $\CalP$ be the projection onto the set of B-spline $F_l$ with $|F_l|=p=\CalO(2^{ld})$ and smoothness $s-1/2>\alpha$. Then for any function $f\in\CalH_\nu$, we have the following $L^r$ error for any $r<\infty$:
    \[\| f-\CalP f\|_{L^r}\leq C \min\{p^{-\frac{\alpha}{d}+\max\{1/2-1/r,0\}}\|f\|_{\CalH_\nu},\|f\|_r\},\]
    where $C$ is some universal constant independent of $f$ and $F_l$.
\end{lemma}
The follow lemma provides estimation for the eigenvalues of the Gram matrix induced by B-splines. It is a special case of the general de Boor's conjecture proved in \cite{shadrin2001norm}.
\begin{lemma}
\label{lem:de_Boor}
    Let B-splines with smoothness $s<\infty$ be the set of basis function $\{\phi_i\}_{i=1}^p=F_l$. Let $M=[\langle \phi_i,\phi_j\rangle_{L^2}]_{ij}$ be the Gram matrix. Then:
    \[\frac{\lambda_{\rm max}(M)}{\lambda_{\rm min}(M)}=\CalO(1)\quad\text{and}\quad \lambda_{\rm min}(M)\geq \CalO(p^{-1}).\]
\end{lemma}
\proof
The entries of Gram matrix can be labeled using the indices $(\bold{j},\bold{u})$:
\[M_{\bold{j},\bold{u}}=\langle M^d_{\bold{l},\bold{j}},M^d_{\bold{l},\bold{u}}\rangle_{L^2}.\]
Because the B-spline is in product form $M^d_{\bold{l},\bold{j}}=\prod_{i=1}^dN_s(2^{l_i}x_i-j_i)$, each entry of $M$ is also the product of $d$ one-dimensional $L^2$ inner products:
\begin{align*}
    M_{\bold{j},\bold{u}}=\prod_{i=1}^d \langle N_s(2^{l}\cdot -j_i),N_s(2^{l}\cdot -u_i)\rangle_{L^2}.
\end{align*}
Therefore, the Gram matrix $M$ is the Kronecker product of $d$ Gram matrices $M^{(i)}$, i.e., $M=\bigotimes_{i=1}^dM^{(i)}$, with:
\begin{align*}
    M^{(i)}_{ju}=&\langle N_s(2^{l}\cdot -j),N_s(2^{l}\cdot -u) \rangle_{L^2},\quad,j,u\in\{-s,-s+1,\cdots,2^l-1,2^l\}.
\end{align*}
Also, the value of $M^{(i)}_{ju}$ only depends on $|j-u|$ and $M^{(i)}_{ju}=0$ for any $|j-u|>s$:
\begin{align*}
    \int_0^1  N_s(2^{l}x -j)N_s(2^{l}x -u)dx= &2^{-l}\int_{-j}^{2^l-j}N_s(x)N_s(x-(u-j))dx.
\end{align*}

The matrix $\int_{-j}^{2^l-j}N_s(x)N_s(x-(u-j))dx$ is  de Boor $L^2$ projection matrix and its maximum and minimum eigenvalues are all independent of the number of grid points $\{-s,-s+1,\cdots,2^l-1,2^l\}$ according to \cite{shadrin2001norm}.

Therefore, we have:
\[\frac{\lambda_{\max}(M^{(i)})}{\lambda_{\min}(M^{(i)})}=\CalO(1),\quad \lambda_{\min}(M(i))\geq 2^{-l}.\]
This leads to the final result:
\[\frac{\lambda_{\max}(M)}{\lambda_{\min}(M)}=\prod_{i=1}^d \frac{\lambda_{\max}(M^{(i)})}{\lambda_{\min}(M^{(i)})}=\CalO(1)  \quad\text{and}\quad \lambda_{\min}(M)=\prod_{i=1}^d \lambda_{\min}(M^{(i)})\geq \CalO(2^{-dl})=\CalO(p^{-1} ).\]

\proof[Proof of Theorem \ref{thm:b_spline_regression}]

The conditional number $\overline{\lambda}/\underline{\lambda}$ of the Gram matrix $M$ is $\CalO(1)$ and the minimum eigenvalue $\underline{\lambda}^{-1}=\CalO(p)$ according to Lemma \ref{lem:de_Boor}. Also, the quantity $S_\phi=\CalO(1)$ because almost all $\phi_j\in F_l$ are disjoint. We now prove the approximation error.

From Lemma \ref{lem:spline_approx_err} for the projection error for functions in the RKHS  $\CalH_\nu$ induced by the Mat\'ern-$\nu$ kernel $k$, we have the following $L^2$ (i.e., $r=2$) approximation error for eigenfunction $e_i$:
\begin{align*}
    \int_\CalX\left|\left(\BI-\CalP\right)[e_i](\Bx)\right|^2d\Bx\leq C\min\{p^{-{2\nu}/{d}-1}\|e_i\|_{\CalH_\nu}^2,1\}\leq C_1  \min\{p^{-{2\nu}/{d}-1}\lambda_i^{-1},1\},
\end{align*}
where the last equality is from the properties of eigenvalues and eigenfunctions for the Mat\'ern kernel in Lemma \ref{lem:eigen_mat}. Therefore,  the approximation error of $k=k_\nu$ is:
\begin{align*}
    \int_\CalX\left(\BI-\CalP\right)k\left(\BI-\CalP\right)(\Bx,\Bx)d\Bx
    =&\sum_i\lambda_i\int_\CalX\left|\left(\BI-\CalP\right)[e_i](\Bx)\right|^2d\Bx\\
    \leq & C_1\sum_i \lambda_i  \min\{p^{-{2\nu}/{d}-1}\lambda_i^{-1},1\}\\
    \leq & C_2\left(\sum_{i\leq p}  p^{-{2\nu}/{d}-1}+\sum_{i> p}  \lambda_i\right)
    \leq  C_3 p^{-2\nu/d}.
\end{align*}
Define:
\begin{equation*}
    g_i(\BFx)\coloneqq\begin{cases}
        &\E[e_i(\BB_\tau)e^{-\kappa^2\tau}|\BB_0=\Bx]\quad (\text{Dirichlet})\\
        &\E\left[\int_{0}^{\infty}e^{-\kappa^2t-\int_0^tc(\BB_s)dL_s}e_i(\BB_t)dL_t\bigg|\BB_0=\BFx\right]\ ( \text{Robin})
    \end{cases}
\end{equation*}

Because $g_i(\Bx)$ is also the solution to the PDE from the path integral identity, we have the following identity:
\begin{align*}
    (\kappa^2-\bigtriangleup)g_i(\Bx)=0,
\end{align*}
corresponding to the Dirichlet or Robin boundary conditions. Therefor, the RKHS norm and $L^2$ norm of $g_i$ can be bounded as follows:
\begin{align*}
    \|g_i\|_2^2=&\int_\CalX \int_{\CalX\times\CalX}G(\BFx,\Bu)e_i(\Bu)e_i(\bold{v})G(\BFx,\bold{v})d\Bu d\bold{v} d\BFx\\
    \leq & \lambda^2_{min}(\kappa^2-\bigtriangleup)\|e_i\|_2^2=\CalO(1),
\end{align*}
\begin{align*}
    \langle g_i,g_i\rangle_{\CalH_\nu}=& \langle \int_\CalX G(\cdot,\Bu)e_i(\Bu)d\Bu,\int_\CalX G(\cdot,\Bu)e_i(\Bu)d\Bu \rangle_{\CalH_\nu}\\
    \leq &  \lambda^2_{min}(\kappa^2-\bigtriangleup)\langle e_i,e_i\rangle_{\CalH_\nu}=\CalO(\lambda_i^{-1}),
\end{align*}
where  $G(\BFx,\BFx')$ is the Green's function of operator $(\kappa^2-\bigtriangleup)$ corresponding to different boundary conditions such that, i.e., $(\kappa^2-\bigtriangleup)G(\BFx,\BFx')=\delta_{\BFx-\BFx'}$ 
and $ \lambda_{min}(\kappa^2-\bigtriangleup)$ denotes the minimum eigenvalues of the operator $(\kappa^2-\bigtriangleup)$ and it is bounded for both the Dirichlet and Robin boundary conditions.

Therefore, the $L^2$ approximation error of $\CalP$ for $\bar{k}$ can be computed as follows:
\begin{align*}
    &\int_\CalX\left(\BI-\CalP\right)\bar{k}\left(\BI-\CalP\right)^T(\Bx,\Bx)d\Bx\\
    =& \sum_i\lambda_i \int_\CalX\left|\left(\BI-\CalP\right)[g_i](\Bx) \right|^2d\BFx\\
    \leq & C_3 \sum_i \lambda_i \min\{p^{-{2\nu}/{d}-1}\|g_i\|^2_{\CalH_\nu},1\}\leq  C_4\left(\sum_{i\leq p}  p^{-{2\nu}/{d}-1}+\sum_{i> p}  \lambda_i\right)
    \leq  C_5 p^{-2\nu/d}.
\end{align*}
For the  $v_*$ in the probability lower bound:
\begin{align}
    v
    _*=&\int_\CalX \left| (\BI-\CalP)[k+\bar{k}](\BI-\CalP)^T(\Bx,\Bx)\right|^2d\Bx\nonumber\\
    =& \int_\CalX\left|(\BI-\CalP)[\sum_i\lambda_i(e_ie_i+g_ig_i)](\BI-\CalP)^T(\Bx,\Bx)\right|^2d\Bx\nonumber\\
    =& \int_\CalX\left|\sum_i \lambda_i \left[|(\BI-\CalP)[e_i](\Bx)|^2+|(\BI-\CalP)[g_i](\Bx)|^2\right]\right|^2d\Bx\nonumber\\
    \leq & 2\int_\CalX \left(\sum_i \lambda_i\right)\left(\sum_i\lambda_i\left[|(\BI-\CalP)[e_i](\Bx)|^2+|(\BI-\CalP)[g_i](\Bx)|^2\right]^2\right)d\Bx\nonumber\\
    \leq & C_4 \sum_i\lambda_i\left(\int_\CalX  |(\BI-\CalP)[e_i](\Bx)|^4d\Bx+\int_\CalX  |(\BI-\CalP)[g_i](\Bx)|^4d\Bx\right)\nonumber\\
    \leq & C_5\sum_i\lambda_i \left(\min\{p^{-4\alpha/d+1}\|e_i\|_{\CalH_v}^4,\|e_i\|_{L^4}^4\}+\min\{p^{-4\alpha/d+1}\|g_i\|_{\CalH_v}^4,\|g_i\|_{L^4}^4\}\right) \nonumber\\
    \leq & C_5\left(\sum_{i\leq p} \lambda_i^{-1}p^{-4\alpha/d+1}+\sum_{i>p}\lambda_i^{1-\frac{d}{2\alpha }}\right)\nonumber\\
    \leq & C_6 p^{-\frac{2\nu}{d}+1}\nonumber,
\end{align}
where the second-to-last line is from Lemma \ref{lem:RKHS_interpolation inequality}, our assumption $\nu>d/2$,  and the fact that $\|e_i\|^2_{\CalH_\nu}=\lambda_i^{-1}$, $\|g_i\|^2_{\CalH_\nu}\leq \CalO(\lambda_i^{-1})$, $\|e_i\|_{L^2}=1$, and $\|g_i\|_{L^2}\leq \CalO(1)$. For the remaining two terms in our probability bound, we have:
\begin{align*}
    &p(\sqrt{2}e^{-1/2})^{m\underline{\lambda}/S_\phi} +p\left(\frac{e^{1/2}}{(3/2)^{3/2}}\right)^{m\overline{\lambda}/S_\phi}\leq 2p\left(\frac{e^{1/2}}{(3/2)^{3/2}}\right)^{m\overline{\lambda}/S_\phi}\\
    \leq & 2p \left(e^{-0.1}\right)^{m\overline{\lambda}/S_\phi}\leq \exp\left[-\frac{m}{10p}+\log 2p\right]\leq \exp\left[-C_2(\frac{m}{p}+\log p)\right].
\end{align*}
Finally, setting $p=\CalO(m^{\frac{1}{2\nu/d+1}})$, the final result is proven.
\endproof


\section{Computational Time Complexity of  Kernel Regression by B-Spline}
\label{sec:B_spline_time_complexity}
Recall that the kernel regression is written as follows:
\begin{equation}
\label{eq:kernel_reg2_appendix}
 \hat{k}_{m,Q}(\BFx, \BFx') = \sum_{q,q'=1}^Q m_{q,q'}{\phi}_q(\BFx) {\phi}_{q'}(\BFx'), \; \mathbf{M} = [m_{q,q'}]_{q,q'=1}^Q = [\BPhi(\BU)]^\dagger \hat{\BK}^{\rm C}(\BU, \BU) [\BPhi(\BU)^T]^\dagger.
\end{equation}
If the basis function $\{\phi_q\}_{q=1}^Q=F_\zeta=\{M_{\zeta, \mathbf{j}}^d, \mathbf{j}\in \mathcal{J}(\zeta)\}$ are B-splines defined as \eqref{eq:B_spline} and the support of each $M_{\zeta, \mathbf{j}}^d$ is:
\begin{equation}
\label{eq:B_spline_support_appendix}
    {\rm supp}\{M_{\zeta, \mathbf{j}}^d\}=[0,2^{-\zeta}(s+1)]^d+\bold{j}2^{-\zeta},
\end{equation}
we can have the following lemma regarding the number of non-zero elements of the vector $\{M_{\zeta, \mathbf{j}}^d(\BFx), \mathbf{j}\in \mathcal{J}(\zeta)\}$ for any $\BFx\in\CalX$:
\begin{lemma}
    \label{lem:nnz_B_spline}
     The number of non-zero element of $[\phi_1(\BFx),\cdots,\phi_Q(\BFx)]$ is at most $(2s)^d$ for any $\BFx\in\CalX$.
\end{lemma}
\begin{proof}
    Let $\CalX$ be partitioned into small, non-overlapping rectangles of the form ${[0, 2^{-\zeta}(s+1)]^d + \boldsymbol{j} 2^{-\zeta}}$ for integer vectors $\boldsymbol{j}$. Then, every point $\BFx \in \CalX$ lies in exactly one such rectangle. Without loss of generality, assume $\BFx \in [0, 2^{-\zeta}(s+1)]^d$. We then count the basis functions $\phi_q$ that are nonzero on this rectangle. According to \eqref{eq:B_spline_support_appendix}, there are at most $(2s)^d$ such functions.
\end{proof}
By Lemma \ref{lem:nnz_B_spline}, to evaluate \eqref{eq:kernel_reg2_appendix} at any given $\BFx$, it suffices to compute only those entries ${m_{q,q'}}$ for which the basis functions $\phi_q$ and $\phi_{q'}$ have overlapping supports. As a result, we only need to compute at most $(2s)^d Q$ entries of ${m_{q,q'}}$, rather than the full inverse of the matrix $\boldsymbol{M}$. This leads to the following lemma:
\begin{lemma}
    \label{lem:time_complexity_sparse_M}
     Let $\mathbf{M} = [m_{q,q'}]_{q,q'=1}^Q = [\BPhi(\BU)]^\dagger \hat{\BK}^{\rm C}(\BU, \BU) [\BPhi(\BU)^T]^\dagger.$ Define the set:
     \[\mathbf{M}^{>0}=\{m_{q,q'}:  {\rm supp}\{\phi_q\}\bigcap   {\rm supp}\{\phi_q'\}\neq \emptyset\}.\]
     Then the size $\bold{M}^{>0}$ is at most $(2s)^dQ$ and it can be computed in $\CalO((2s)^dQ^3)$ time.
\end{lemma}
\begin{proof}
    We first count the size of $\bold{M}^{>0}.$ According to Lemma \ref{lem:nnz_B_spline}, there are at most $(2s)^d$ B-spline basis functions with overlapping support with each $\phi_q$. So we can directly conclude that $|\bold{M}^{>0}|=\CalO((2s)^d)Q$. 
    
    For computational time complexity, we directly provide an algorithm to solve for $\bold{M}^{>0}$ in $\CalO((2s)^dQ^3)$ time. Define:
    \begin{align*}
        \bold{S}=\left[\BPhi(\BU)\BPhi(\BU)^T\right]^{-1},\quad \bold{R}=\BPhi(\BU) \hat{\BK}^{\rm C}(\BU, \BU) \BPhi(\BU)^T,
    \end{align*}
    so that $\forall m_{q,q'}$, $ m_{q,q'}= \bold{S}_{q,:}\bold{R}\bold{S}_{:,q'}$. Note that $\bold{S}\in\Real^{Q\times Q}$, so it can be computed in $\CalO(Q^3)$ time. Also, because of the sparsity of $ \BPhi(\BU)$, $\bold{R}$ can be computed in $\CalO(Q^3)$ time. 
    
    After we have computed $\bold{R}$, the computational time complexity for matrix-vector multiplication $m_{q,q'}=\bold{S}_{q,:}\bold{R}\bold{S}_{:,q'}$ is $\CalO(Q^2)$. In total, we only need to compute those $m_{q,q'}\in\bold{M}^{>0}$, so we need to repeat this process $(2s)^dQ$ times.

    Therefore, the total time complexity for all $m_{q,q'}\in\bold{M}^{>0}$ is $\CalO((2s)^dQ^3)$.
\end{proof}
After we can have the values of $\bold{M}_0$, we have the following efficient computation for kernel values:

\begin{theorem}
    \label{thm:time_complexity_kernel}Given $\bold{M}^{>0}$,  $\hat{k}_{m,Q}(\BFx, \BFx') $ can be computed in $(2s)^{2d}$ steps at most for any pair of $\BFx$ and $\BFx'$.
\end{theorem}
\begin{proof}
    By Lemma \ref{lem:nnz_B_spline}, for any $\BFx$ and $\BFx'$, the vectors $\BPhi(\BFx)$ and $\BPhi(\BFx')$ each contain at most $n^* = (2s)^d$ nonzero entries. Denote these nonzero basis functions as $[\phi_{q_1}, \ldots, \phi_{q_{n^*}}]$ and $[\phi_{q_1'}, \ldots, \phi_{q_{n^*}'}]$, respectively. Then:
    \[\hat{k}_{m,Q}(\BFx, \BFx')=\sum_{i=1}^{n^*}\sum_{j=1}^{n^*}m_{q_i,q_j}\phi_i(\BFx)\phi_j(\BFx).\]
    There are thus at most $(2s)^{2d}$ computations in total.
\end{proof}
\end{appendix}

\end{document}